\documentclass{article} 
\usepackage{PRIMEarxiv} 
\usepackage[utf8]{inputenc} 
\usepackage[T1]{fontenc}    
\usepackage{hyperref}       
\usepackage{url}            
\usepackage{booktabs}       
\usepackage{amsfonts}       
\usepackage{nicefrac}       
\usepackage{microtype}      
\usepackage{xcolor}         
\usepackage{natbib}
\usepackage{algorithm}
\usepackage{algorithmic}
\usepackage{mathtools}
\usepackage{amsthm}
\usepackage{amssymb}
\usepackage{accessibility}

\usepackage[capitalize,noabbrev]{cleveref}

\usepackage{amsthm}                
\usepackage{bm}                    

\usepackage{aliascnt}

\newcommand{\reals}{\mathbb{R}}

\newcommand{\ba}{{a}}

\newcommand{\calA}{\mathcal{A}}
\newcommand{\calS}{\mathcal{S}}

\theoremstyle{plain}
\newtheorem{theorem}{Theorem}[section]

\newaliascnt{proposition}{theorem}
\newtheorem{proposition}[proposition]{Proposition}
\aliascntresetthe{proposition}
\crefname{proposition}{proposition}{propositions}

\newaliascnt{lemma}{theorem}
\newtheorem{lemma}[lemma]{Lemma}
\aliascntresetthe{lemma}
\crefname{lemma}{lemma}{lemmas}

\newaliascnt{corollary}{theorem}

\aliascntresetthe{corollary}
\crefname{corollary}{corollary}{corollaries}

\theoremstyle{definition}
\newtheorem{definition}{Definition}[section]

\theoremstyle{remark}
\newtheorem{remark}{Remark}[section]

\title{Game-Theoretic Understandings of Multi-Agent Systems with Multiple Objectives}

\author{
  Zixin Lin \\
  University of Central Florida \\
  Orlando, FL, USA \\
  \And
  Qianyi Tan \\
  University of Texas at Arlington \\
  Arlington, TX, USA \\
  \AND
  Sihong He \\
  University of Texas at Arlington \\
  Arlington, TX, USA \\
  \And
  Yue Wang \\
  University of Central Florida \\
  Orlando, FL, USA \\
}
\begin{document}
\maketitle

\begin{abstract}
Agents in many multi-agent systems must balance several, possibly conflicting, objectives, and the trade-offs an agent can achieve depend on how the other agents behave. We study this setting through the Multi-Objective Markov Game (MOMG), a Markov game in which each agent receives a vector of rewards, and adopt the Pareto-Nash Equilibrium (PNE) as the solution concept: at a PNE, no agent can unilaterally improve one of its objectives without sacrificing another. We show that PNEs always exist and coincide with the Nash equilibria of the linearly scalarized games with strictly positive weights, so standard single-objective Markov-game solvers suffice to compute them. We further define the Weak Pareto-Nash Equilibrium (WPNE), which excludes only unilateral deviations that strictly improve all objectives at once. The associated gap (WPNG) is zero exactly at WPNEs and never exceeds the Nash gap of any linearly scalarized game; hence an approximate Nash equilibrium of a scalarized game is also an approximate WPNE. When the preference profile is known in advance, learning in an MOMG therefore reduces to learning in a single scalarized Markov game, for which we provide a sample-efficient algorithm. When preferences are revealed only after data collection, we design a two-phase algorithm that first explores without any preference information; with high probability, every Nash equilibrium of the learned scalarized game is then an approximate WPNE for every queried preference profile, and the WPNG of any policy can be estimated accurately from the learned model. Experiments on tabular games and four multi-objective multi-agent environments illustrate the theory and show the effect of preference-free exploration.
\end{abstract}
\maketitle

\section{Introduction}
Multi-agent systems \cite{shoham2008multiagent,wooldridge2009introduction} are becoming integral to complex real-world domains such as robotic warehouses \cite{lowe2017multi, matignon2012independent}, automated trading \cite{mi2023taxai, huang2024multi}, and network management \cite{zhang2018fully, abrol2024deep, shabestary2022adaptive}. To model the complex interactions among agents and with the environment, the framework of Markov games \cite{fink1964equilibrium,filar2012competitive,shapley1953stochastic} has been proposed to capture the sequential, strategic nature of such systems. Building on this framework, Multi-Agent Reinforcement Learning (MARL) \cite{zhang2021multi, hernandez2019survey} has emerged as the principal learning paradigm for engineering intelligent behavior in multi-agent settings \cite{lowe2017multi, yu2022surprising}.

Although extensively studied, standard Markov games and MARL rest upon a key simplifying assumption: that each agent's goal can be captured by a single scalar reward function. In most practical settings, however, agents must navigate diverse and often conflicting objectives \cite{eriksson2008adaptive,osika2024navigating, chapman2023bridging}. A financial trading agent, for instance, must balance profit maximization against risk mitigation \cite{addy2024machine, olanrewaju2025artificial, bhardwaj2024risk}. While Multi-Objective Reinforcement Learning (MORL) \cite{van2013scalarized, yang2019generalized,hayes2022practical} provides a framework for single-agent decision-making with vectorial rewards, it assumes a stationary, non-strategic environment—an assumption that breaks down when multiple agents interact.

The difficulty in multi-objective multi-agent settings is that the trade-offs available to one agent depend on the other agents' policies. Consider traffic signal control at two adjacent intersections. Each controller trades throughput against fairness across its queues, and vehicles leaving one intersection join the queues of the other. Whether a controller can increase its throughput without reducing fairness therefore depends on its neighbor's signal plan. Moreover, the appropriate balance between throughput and fairness differs between peak and off-peak hours, so the weights on the objectives may not be known when the controllers are trained. 
The first observation means that the agents' efficient policies cannot be computed separately and then combined: the outcome must be stable against unilateral deviations, which calls for combining game-theoretic equilibrium \cite{fudenberg1991game} with Pareto optimality \cite{ehrgott2005multicriteria}. The second means that a standard Markov game \cite{littman1994markov} is also insufficient, since its scalar rewards require the weights to be fixed before learning begins.

In this paper, we address this gap by introducing a theoretical and algorithmic framework for multi-objective multi-agent reinforcement learning. Our contributions are as follows.

\textbf{Multi-Objective Markov Games and Pareto-Nash Equilibrium.} We formalize the Multi-Objective Markov Game (MOMG), extending standard Markov games to vectorial rewards to model each agent's multiple objectives. We propose the Pareto-Nash Equilibrium (PNE) as its solution concept:  a joint policy from which no agent can unilaterally improve any objective without worsening another, thereby unifying Nash-type strategic stability with Pareto optimality in a single notion. We prove that every finite MOMG admits a PNE in randomized Markov policies.
The natural gap associated with PNE vanishes at every PNE but can also vanish outside the PNE set. We therefore also introduce the Weak Pareto-Nash Equilibrium (WPNE), which rules out only unilateral deviations that strictly improve every objective. Its gap vanishes exactly on the WPNE set, can be computed by one linear program per player in a known model, and is bounded above by the Nash gap of every linearly scalarized game.

\textbf{Equivalence via Linear Scalarization.} We establish that the set of all PNEs of a MOMG is exactly the union of Nash equilibria across all linearly scalarized games with strictly positive preference weights, and that the analogous equivalence holds between WPNEs and non-negative weights. These results are central to the paper: they reduce finding PNEs and WPNEs to solving standard single-objective Markov games, so existing solvers apply directly. Unlike in normal-form games, a player's value in a Markov game depends nonlinearly on its policy, so the normal-form arguments do not apply directly. We use occupancy measures to show that the set of value vectors a player can reach against fixed opponents is a polytope, which yields strictly positive supporting weights for every PNE.

\textbf{Sample-Efficient Algorithms for Fixed and Post-hoc Preferences.} When the preference profile is fixed in advance, the bound of WPNG by the scalarized Nash gap reduces learning a WPNE to learning a Nash equilibrium of one scalarized Markov game. As a concrete instance, we adapt optimistic Nash value iteration to MOMGs with unknown transitions and stochastic vector rewards. It achieves Nash regret $\tilde{\mathcal{O}}(H^2S\sqrt{AT})$ and returns an $\epsilon$-WPNE after $\tilde{\mathcal{O}}(H^4S^2A/\epsilon^2)$ episodes.
When preferences are unknown during data collection, we propose a two-phase algorithm. 
In Phase 1, all agents explore with a common reward that is large where the model is still uncertain, without using any preference information. 
In Phase 2, a Nash equilibrium of the scalarized game on the learned model is computed for each queried preference profile. 
We show that $\tilde O(H^4S^2A/\varepsilon^2)$ episodes of Phase 1 suffice for all later queries. With high probability, for every preference profile, including profiles chosen after seeing the data, every Nash equilibrium of the learned scalarized game is an $\epsilon$-WPNE of the MOMG; in addition, the WPNG of every policy computed on the learned model is within $\epsilon$ of its true value. Like Nash value iteration, both algorithms assume an exact solver for general-sum Nash equilibria, so our guarantees concern sample complexity rather than computation.

\textbf{Related work.} Pareto-Nash equilibria and their characterization by linear scalarization were first studied in normal-form games with vector payoffs~\cite{shapley1959equilibrium,zhao1991equilibria,lozovanu2005multiobjective}; utility-based equilibria of such games are analyzed in~\cite{ropke2022nash,ruadulescu2020multi}. Our proof of \Cref{thm:scalar} uses a classical fact from multiple-objective linear programming~\cite{EvansSteuer1973}: every Pareto optimal point of a polytope has a strictly positive supporting weight. Mifrani and Noll~\cite{MifraniNoll2025} apply this fact to the occupancy measures of a single vector-valued MDP; we apply it to each player's deviation problem against fixed opponents. 
For learning, the closest result is Liu et al.~\cite{liu2021sharp}, who show that after reward-free exploration of a general-sum Markov game, every Nash equilibrium of the learned game is an approximate Nash equilibrium of the true game, simultaneously for finitely many reward functions fixed in advance. 
Our Phase 1 guarantee covers the continuum of agent-specific linear preferences over vector rewards, including preferences chosen after the data are collected. The closest single-agent analogue is preference-free exploration in MORL~\cite{wu2021accommodating,qiu2024traversing}. Other multi-objective multi-agent formulations~\cite{yu2021provably,chang2014partially} target different goals, such as approachability. Appendix~\ref{sec:related works} discusses further related work.

\section{Background and Preliminaries}
\subsection{Multi-Objective Markov Decision Processes}
An MOMDP \cite{chatterjee2006markov,
wakuta1998solution,
wiering2007computing} is formally defined as a tuple $(\calS, \calA, P, {r},H)$, where 
$\calS,\calA$ are the finite state and action spaces, and $P: \calS \times \calA \times [H]\to \Delta(\calS)$ is the transition kernel, where $P_h(s'|s, a)$ is the probability of transitioning from state $s$ to $s'$ after taking action $a$. The reward vector is $\mathbf{r}:\calS\times\calA\times [H]\to \mathbb{R}^M$, where $\mathbf{r}_h(s, a)\in\mathbb{R}^M$ is an $M$-dimensional vector that the agent receives at step $h$ after taking action $a$ at state $s$. 

A policy is a sequence of decision rules $\pi = (\pi_1, \dots, \pi_{H})$, where each $\pi_h: \calS \to \Delta(\calA)$ specifies the action to take at step $h\in [H]$ under different states. The performance of  a policy $\pi$ starting from step $h$ is captured by an $M$-dimensional vector of the expected cumulative reward following $\pi$:
$$
\mathbf{V}^{\pi}_h(s) = \mathbb{E}_{\pi,P} \big[ \sum_{\ell=h}^{H} \mathbf{r}_\ell(s_\ell, a_\ell) \mid s_h = s \big]\in\mathbb{R}^M,
$$
where the expectation is taken w.r.t. the randomness in policies and transition kernels. 
Since the objectives may conflict,  there is generally no single policy that can maximize every entry of $\mathbf{V}$ for all objectives. Therefore, the goal in an MOMDP is to find \textbf{Pareto optimal policies} \cite{wiering2007computing,
pirotta2015multi,
qiu2024traversing}, defined as follows.  
\begin{definition}[Pareto optimal policy]\label{def:pareto-optimal-policy} Let $\mathbf{u} = (u_1, \dots, u_M)$ and $\mathbf{v} = (v_1, \dots, v_M)$ be two vectors in $\reals^M$. The vector $\mathbf{u}$ {Pareto dominates} the vector $\mathbf{v}$ if (1)
$u_i \ge v_i$ for all components $i \in [M]$; and (2) there exists some $j \in [M]$ such that $u_j > v_j$.
A policy $\pi^*$ is {Pareto optimal} if no other policy $\pi$ Pareto dominates it, i.e., there is no other $\pi$ such that $\mathbf{V}^{\pi}_1(s_1)$ Pareto dominates $\mathbf{V}^{\pi^*}_1(s_1)$.
\end{definition}

\subsection{Finite-Horizon Single-Objective Markov Games}
\label{subsec:finite SOMG}
A Finite-Horizon Single-Objective Markov Game \cite{liu2021sharp}, or Stochastic Game, models multi-agent interactions. Each agent seeks to maximize its own scalar reward. An $N$-player Markov Game is defined by the tuple $G = (\mathcal{N}, \mathcal S, \{\mathcal A_k\}_{k \in \mathcal{N}}, P, \{r_k\}_{k \in \mathcal{N}}, H)$. Here, $\mathcal{N} = \{1, 2, \dots, N\}$ is the finite set of agents, and $P_h: \mathcal S \times \mathcal A(\triangleq \times_{k}\mathcal A_k)\to\Delta(\mathcal S)$ is the transition function, dependent on a joint action $\bm a = (a_1, \dots, a_N)$. Each agent $k$ then receives a scalar reward $r_{k,h}(s,\bm a)\in[0,1]$ after all agents take the joint action $\bm a$ at state $s$ and step $h$. 

A policy for agent $k$ is a finite-horizon randomized Markov policy
$\pi_k=(\pi_{k,1},\ldots,\pi_{k,H})$, where
$\pi_{k,h}(\cdot\mid s)\in\Delta(\mathcal A_k)$.
We denote the corresponding policy class by $\Pi_k:=\prod_{h=1}^H\prod_{s\in\mathcal S}\Delta(\mathcal A_k),$ and
the product policy
$\bm\pi=(\pi_1,\ldots,\pi_N)\in\Pi:=\prod_{k=1}^N\Pi_k$.
The value function for agent $k$ under $\pi$ at step $h$ is $V^\pi_{k,h}(s)
=
\mathbb E_{\pi,P}
\left[
\sum_{\ell=h}^{H} r_{k,\ell}(s_\ell,a_\ell)
\,\middle|\, s_h=s
\right].$
Throughout, equilibrium conditions are evaluated from a fixed initial
state $s_1$.

The standard solution concept is the \textbf{Nash Equilibrium (NE)} \cite{shapley1953stochastic,littman1994markov,fink1964equilibrium}. A product policy $\pi^*$ is an NE if no agent $k$ can improve its expected return by unilaterally changing its policy $\pi_{k}$, given that all other agents stick to $\pi^*_{-k}$, i.e., for all $k \in \mathcal{N}$ and any alternative policy $\pi_k\in\Pi_k$:
$$
V^{(\pi^*_{k}, \pi^*_{-k})}_{k,1}(s_1) \geq V^{(\pi_{k}, \pi^*_{-k})}_{k,1}(s_1).
$$

\section{Multi-Objective Markov Games and Solution Concepts}
\label{sec:momg}

In this section, we formulate multi-objective Markov games (MOMGs) as the framework of multi-objective multi-agent systems, and propose solution concepts for them.

\begin{definition}[General-Sum MOMG]
\label{def:momg}
An $N$-player MOMG $\mathcal{G}=(\mathcal{N},\mathcal{S},\{\mathcal{A}_k\}_{k=1}^N,P,H,\{\mathbf{r}_k\}_{k=1}^N)$ is defined as a standard finite-horizon Markov game, except that player $k$ receives a vector reward $\mathbf{r}_{k,h}(s,\mathbf{a})\in[0,1]^M$ at step $h$, state $s$ and joint action $\mathbf{a}$. Under a joint policy $\bm{\pi}$, the value of player $k$ is the vector $\mathbf{V}^{\bm{\pi}}_{k,h}(s)=(V^{\bm{\pi},1}_{k,h}(s),\dots,V^{\bm{\pi},M}_{k,h}(s))\in\mathbb{R}^M$, where $V^{\bm{\pi},i}_{k,h}$ is the standard value function for the $i$-th reward coordinate.
\end{definition}

Policies are as in Section~\ref{subsec:finite SOMG}: player $k$ uses a time-dependent randomized Markov policy $\pi_k\in\Pi_k$, and a product policy is a product $\bm{\pi}=(\pi_1,\dots,\pi_N)$, i.e., the players randomize independently. Values are evaluated from the fixed initial state $s_1$. Besides Pareto dominance (Definition~\ref{def:pareto-optimal-policy}), we say that $\mathbf{u}\in\mathbb{R}^M$ \emph{strictly dominates} $\mathbf{v}\in\mathbb{R}^M$ if $u_i>v_i$ for all $i$. We write $\Delta_M$ for the probability simplex over $[M]$ and $\Delta_M^\circ=\{\bm{\lambda}\in\Delta_M:\lambda_i>0\ \forall i\}$ for its relative interior.

Our formulation extends both MOMDPs ($N=1$) and single-objective Markov games ($M=1$), but neither of their solution concepts applies directly. Pareto optimality treats the other agents as part of a fixed environment, while a Nash equilibrium needs a scalar payoff to compare deviations. In an MOMG, the trade-offs available to a player depend on the other players' policies: changing an opponent's policy can change which of the player's responses are Pareto efficient. One therefore cannot compute each player's efficient policies separately and simply combine them into an equilibrium.

\subsection{Pareto-Nash Equilibrium}
\label{sec:pne}

We propose the following solution concept for MOMGs.

\begin{definition}[Pareto-Nash Equilibrium (PNE)]
\label{def:pne}
A joint policy $\bm{\pi}^*$ is a PNE if, for every player $k$, there is no policy $\pi_k\in\Pi_k$ such that $\mathbf{V}^{(\pi_k,\bm{\pi}^*_{-k})}_{k,1}(s_1)$ Pareto dominates $\mathbf{V}^{\bm{\pi}^*}_{k,1}(s_1)$.
\end{definition}

At a PNE, no agent can improve one of its objectives without worsening another, as long as the other agents keep their policies. When $N=1$, a PNE is a Pareto optimal policy of an MOMDP; when $M=1$, it is a Nash equilibrium of a Markov game. PNE hence combines strategic stability from game theory with Pareto optimality from multi-objective optimization.

To characterize and compute PNEs, we further connect them to standard Markov games, through the scalarization technique in multi-objective studies. Given a preference profile $\bm{\Lambda}=(\bm{\lambda}_1,\dots,\bm{\lambda}_N)$ with $\bm{\lambda}_k\in\Delta_M$, the linearly scalarized game $\mathcal G_{\bm{\Lambda}}$ is the single-objective Markov game with the same states, actions and transitions as $\mathcal{G}$, in which player $k$ receives the scalar reward $\bm{\lambda}_k^\top\mathbf{r}_{k,h}\in[0,1]$. We write $\operatorname{NE}(\mathcal G_{\bm{\Lambda}})$ for its set of Nash equilibria in product Markov policies. With this scalarization, we can prove the following existence of PNE and equivalence between PNE and scalarized NE. 

\begin{theorem}[Existence]\label{thm:pne-exist}
Every finite MOMG has a PNE in time-dependent randomized Markov policies.
\end{theorem}

\begin{proof}[Proof sketch (Appendix~\ref{sec:pne-appendix})]
Take any $\bm{\Lambda}\in(\Delta_M^\circ)^N$. The finite-horizon game $\mathcal G_{\bm{\Lambda}}$ has a Nash equilibrium in randomized Markov policies, which can be constructed by backward induction, solving a finite normal-form game at every step and state. If some player had a deviation that Pareto dominates its equilibrium value, the deviation would strictly increase the player's scalarized value, because all weights in $\bm{\lambda}_k$ are positive. This contradicts the Nash property, so the equilibrium is a PNE.
\end{proof}

The argument above shows that every Nash equilibrium of a positively scalarized game is a PNE. The next result shows that the converse also holds, i.e., scalarization finds all PNEs.

\begin{theorem}[Scalarization equivalence]\label{thm:scalar} For a finite MOMG, 
\[
\operatorname{PNE}(\mathcal{G})=\bigcup_{\bm{\Lambda}\in(\Delta_M^\circ)^N}\operatorname{NE}(\mathcal G_{\bm{\Lambda}}).
\]
\end{theorem}

Theorem~\ref{thm:scalar} reduces finding PNEs to solving standard single-objective Markov games over linear scalarizations. Any Nash equilibrium solver applied to $\mathcal G_{\bm{\Lambda}}$ with strictly positive weights returns a PNE, and every PNE can be obtained in this way for some choice of weights. Which PNE is returned depends on the weights and on the solver.

\begin{proof}[Proof sketch (Lemmas~\ref{lem:achievable-polytope}--\ref{lem:positive-support} and Theorem~\ref{thm:pne-scalarization})]
The inclusion ``$\supseteq$'' was shown above. For ``$\subseteq$'', fix a PNE $\bm{\pi}^*$ and a player $k$. With the other players fixed at $\bm{\pi}^*_{-k}$, player $k$ faces a finite-horizon MOMDP. Its state-action occupancy measures are described by finitely many linear constraints, every feasible occupancy measure is induced by some Markov policy, and the value vector is linear in the occupancy measure. Hence the set $\mathcal{V}_k(\bm{\pi}^*_{-k})$ of value vectors that player $k$ can reach by deviating is a polytope (Lemma~\ref{lem:achievable-polytope}). Since $\bm{\pi}^*$ is a PNE, $\mathbf{V}^{\bm{\pi}^*}_{k,1}(s_1)$ is a Pareto optimal point of this polytope, and a Pareto optimal point of a polytope has a strictly positive supporting weight $\bm{\lambda}_k\in\Delta_M^\circ$ (Lemma~\ref{lem:positive-support}). Thus $\pi^*_k$ is a best response to $\bm{\pi}^*_{-k}$ in the game scalarized by $\bm{\lambda}_k$. Repeating this for every player gives a profile $\bm{\Lambda}$ with $\bm{\pi}^*\in\operatorname{NE}(\mathcal G_{\bm{\Lambda}})$.
\end{proof}

\begin{remark}[Relation to normal-form games]
PNE and its scalarization characterization were first studied in normal-form games with vector payoffs~\cite{shapley1959equilibrium,borm1988pareto,zhao1991equilibria,lozovanu2005multiobjective,ropke2022nash}. The existence of strictly positive weights for efficient points of multiple-objective linear programs is classical~\cite{EvansSteuer1973}. In a normal-form game, once the opponents are fixed, a player's achievable payoffs are the convex hull of finitely many pure-strategy payoffs. But our results for a Markov game are not direct extensions: value functions now depend nonlinearly on policies through the trajectory distribution, and we instead restore these results through the occupancy-measure argument.
\end{remark}

\subsection{Weak Pareto-Nash Equilibrium}
\label{sec:wpne}

After introducing PNE, a natural question is how to check whether a given policy is a PNE, or how far it is from being one. In single-objective games, the standard measure is the Nash gap. For a unilateral deviation $\pi'_k$ from $\bm{\pi}$, let
\[
\mathbf{d}_k(\pi'_k;\bm{\pi}):=\mathbf{V}^{(\pi'_k,\bm{\pi}_{-k})}_{k,1}(s_1)-\mathbf{V}^{\bm{\pi}}_{k,1}(s_1)
\]
be its gain vector. For a preference profile $\bm{\Lambda}=(\bm{\lambda}_1,...,\bm{\lambda}_N)$, the Nash gap of $\bm{\pi}$ in $\mathcal G_{\bm{\Lambda}}$ is
\[
\operatorname{NashGap}_{\mathcal G_{\bm{\Lambda}}}(\bm{\pi}):=\max_{k\in[N]}\,\sup_{\pi'_k\in\Pi_k}\,\bm{\lambda}_k^\top\mathbf{d}_k(\pi'_k;\bm{\pi}).
\]
It is nonnegative and equals zero exactly at a Nash equilibrium of $\mathcal G_{\bm{\Lambda}}$. We call $\bm{\pi}$ an $\epsilon$-NE of $\mathcal G_{\bm{\Lambda}}$ if $\operatorname{NashGap}_{\mathcal G_{\bm{\Lambda}}}(\bm{\pi})\le\epsilon$.

By Theorem~\ref{thm:scalar}, $\bm{\pi}$ is a PNE if and only if its Nash gap is zero for some $\bm{\Lambda}\in(\Delta_M^\circ)^N$. This motivates the following candidate gap, called the Pareto-Nash gap (PNG), in which the weight is chosen separately for each deviation:
\[
\operatorname{PNG}(\bm{\pi})=\max_{k\in[N]}\,\sup_{\pi'_k\in\Pi_k}\,\inf_{\bm{\lambda}\in\Delta_M^\circ}\,\bm{\lambda}^\top\mathbf{d}_k(\pi'_k;\bm{\pi}).
\]
However, a zero PNG does not guarantee a PNE.

\begin{proposition}\label{prop:png}
If $\bm{\pi}$ is a PNE, then $\operatorname{PNG}(\bm{\pi})=0$; But  the converse does not hold.
\end{proposition}

The reason is that the infimum over the open set $\Delta_M^\circ$ can be approached at its boundary without being attained. For example (Proposition~\ref{prop:png-appendix}), consider one player with two objectives, one state, one step, and two actions with rewards $\mathbf{r}(a_1)=(0,0)$ and $\mathbf{r}(a_2)=(1,0)$. Playing $a_1$ is not a PNE, since $a_2$ Pareto dominates it. A deviation that plays $a_2$ with probability $p$ has gain $\mathbf{d}=(p,0)$, and $\inf_{\bm{\lambda}\in\Delta_2^\circ}\bm{\lambda}^\top\mathbf{d}=\inf_{\lambda_1\in(0,1)}\lambda_1 p=0$. Hence $\operatorname{PNG}=0$, although the policy is not a PNE. Adding players whose actions affect neither rewards nor transitions turns this into a multi-player example.

In fact, for every $\mathbf{d}\in\mathbb{R}^M$, the infimum over $\Delta_M^\circ$ equals the minimum over the closed simplex $\Delta_M$, which is $\min_i d_i$. A zero PNG therefore only says that no deviation strictly improves \emph{all} objectives. This leads to the following weaker solution concept.

\begin{definition}[Weak Pareto-Nash Equilibrium (WPNE)]
\label{def:wpne}
A joint policy $\bm{\pi}^*$ is a WPNE if, for every player $k$, there is no policy $\pi_k\in\Pi_k$ such that $\mathbf{V}^{(\pi_k,\bm{\pi}^*_{-k})}_{k,1}(s_1)$ strictly dominates $\mathbf{V}^{\bm{\pi}^*}_{k,1}(s_1)$.
\end{definition}

A WPNE rules out only deviations that improve every objective, whereas a PNE also rules out deviations that improve some objectives and keep the others unchanged. Hence every PNE is a WPNE, and a WPNE always exists by Theorem~\ref{thm:pne-exist}. We define the weak Pareto-Nash gap (WPNG) by replacing the open simplex in PNG with the closed one:
\[
\begin{aligned}
\operatorname{WPNG}(\bm{\pi})&:=\max_{k\in[N]}\,\sup_{\pi'_k\in\Pi_k}\,\min_{\bm{\lambda}\in\Delta_M}\,\bm{\lambda}^\top\mathbf{d}_k(\pi'_k;\bm{\pi})\\
&=\max_{k\in[N]}\,\sup_{\pi'_k\in\Pi_k}\,\min_{i\in[M]}\,d_{k,i}(\pi'_k;\bm{\pi}).
\end{aligned}
\]
We call $\bm{\pi}$ an $\epsilon$-WPNE if $\operatorname{WPNG}(\bm{\pi})\le\epsilon$, i.e., no player can improve all of its objectives by more than $\epsilon$ through a unilateral deviation. By the identity above, PNG and WPNG are numerically identical. Their zero value characterizes WPNE (Theorem~\ref{thm:wpne} below), but is only a necessary condition for PNE (Proposition~\ref{prop:png}).

\begin{theorem}\label{thm:wpne}
(1) $\operatorname{WPNE}(\mathcal{G})=\bigcup_{\bm{\Lambda}\in(\Delta_M)^N}\operatorname{NE}(\mathcal G_{\bm{\Lambda}})$. 
(2) $\operatorname{WPNG}(\bm{\pi})=0$ if and only if $\bm{\pi}$ is a WPNE.
\end{theorem}

\begin{proof}[Proof sketch (Theorems~\ref{thm:wpne-scalarization}--\ref{thm:wpne-zero-gap})]
For (1), a Nash equilibrium under nonnegative weights that sum to one cannot admit a deviation that improves every objective, since such a deviation would increase the weighted sum. Conversely, at a WPNE, $\mathbf{V}^{\bm{\pi}^*}_{k,1}(s_1)$ is a weakly Pareto optimal point of the convex set $\mathcal{V}_k(\bm{\pi}^*_{-k})$, and a separating hyperplane gives a nonnegative supporting weight. Unlike in Theorem~\ref{thm:scalar}, the weight need not be strictly positive, so convexity suffices. For (2), $\min_i d_{k,i}>0$ holds exactly when the deviation improves every objective, and the trivial deviation $\pi'_k=\pi_k$ gives zero.
\end{proof}

We remark that for a policy $\bm\pi$, its PNG is equivalent to WPNG, and both gaps can also be written through the scalarized Nash gap (see Appendix~\ref{sec:wpne-appendix}):
\[
\begin{aligned}
&\operatorname{PNG}(\bm{\pi})= \inf_{\bm{\Lambda}\in(\Delta_M^\circ)^N}\operatorname{NashGap}_{\mathcal G_{\bm{\Lambda}}}(\bm{\pi})\\
&=\min_{\bm{\Lambda}\in(\Delta_M)^N}\operatorname{NashGap}_{\mathcal G_{\bm{\Lambda}}}(\bm{\pi})=\operatorname{WPNG}(\bm{\pi}).
\end{aligned}
\]
The difference is that the minimum over the closed simplex is attainable, but the infimum over the interior may not be, as in the example above. Hence $\bm{\pi}$ is a WPNE when its smallest scalarized Nash gap is zero, and a PNE when this zero is attained at an interior profile. Therefore, WPNE enjoys a better guarantee than PNE.

From a learning perspective, the most useful consequence is that, for every fixed $\bm{\Lambda}\in(\Delta_M)^N$,
\begin{equation}\label{eq:pub}
\operatorname{WPNG}(\bm{\pi})\le\operatorname{NashGap}_{\mathcal G_{\bm{\Lambda}}}(\bm{\pi}),
\end{equation}
since $\min_i d_{k,i}\le\bm{\lambda}_k^\top\mathbf{d}_k$ for every player and deviation. Hence every $\epsilon$-NE of a scalarized game is an $\epsilon$-WPNE, and learning a WPNE reduces to learning in a scalarized Markov game. When the model is known, WPNG can be computed exactly by solving one linear program per player and taking the maximum of their optimal values. Each program maximizes the smallest coordinate gain over the player's feasible occupancy measures. A related linear program checks whether a given policy is a PNE~\cite{benson1978existence}. Both formulations are given in Appendix~\ref{sec:wpne-appendix}. In the following learning part, we will mainly focus on learning WPNEs.

\subsection{Utility-Based Equilibria}
\label{sec:utility}

Besides PNE, another class of solution concepts for multi-objective problems is based on utility functions~\cite{borm1988pareto,lozovanu2005multiobjective,ruadulescu2020multi,rodriguez2024analytical,ropke2023bridging,ropke2022nash}. A utility function $u_k:\mathbb{R}_+^M\to\mathbb{R}$ encodes a preference of agent $k$ that is fixed in advance. It maps a vector to a scalar and hence reduces the multi-objective problem to a single-objective one. We extend the two standard utility-based criteria to MOMGs as follows.

\begin{definition}[ESR-NE]
Given a utility function $u_k$ for each agent $k$, a joint policy $\bm{\pi}$ is an Expected Scalarized Return Nash Equilibrium (ESR-NE) if, for every agent $k$ and every policy $\pi'_k\in\Pi_k$, $V^{\bm{\pi},\tilde r}_{k,1}(s_1)\ge V^{(\pi'_k,\bm{\pi}_{-k}),\tilde r}_{k,1}(s_1)$, where $V^{\cdot,\tilde r}_{k,1}$ is the value for the scalar reward $\tilde r_{k,h}=u_k(\mathbf{r}_{k,h})$.
\end{definition}

\begin{definition}[SER-NE]
Given a utility function $u_k$ for each agent $k$, a joint policy $\bm{\pi}$ is a Scalarized Expected Return Nash Equilibrium (SER-NE) if, for every agent $k$ and every policy $\pi'_k\in\Pi_k$, $u_k\big(\mathbf{V}^{\bm{\pi}}_{k,1}(s_1)\big)\ge u_k\big(\mathbf{V}^{(\pi'_k,\bm{\pi}_{-k})}_{k,1}(s_1)\big)$.
\end{definition}

ESR applies the utility to the reward of each step and then takes the expected sum, while SER applies the utility to the expected return. Both reduce to linear scalarization when $u_k(\mathbf{v})=\bm{\lambda}_k^\top\mathbf{v}$ with $\bm{\lambda}_k\in\Delta_M$. Note that our ESR applies the utility stage by stage. For nonlinear $u_k$, this differs from applying the utility to the whole return, $\mathbb{E}[u_k(\sum_h\mathbf{r}_{k,h})]$, although the two coincide in the one-step examples below (Appendix~\ref{sec:utility-appendix}). Both notions are widely studied in multi-objective normal-form games and multi-objective RL. We argue, however, that they are less suitable than (W)PNE for MOMGs, for two reasons: existence and consistency with Pareto optimality.

\begin{proposition}[Existence]\label{prop:utility-existence-main}
An ESR-NE always exists, but an SER-NE may not exist.
\end{proposition}

An ESR-NE exists because, with the transformed rewards $\tilde r_{k,h}=u_k(\mathbf{r}_{k,h})$, the problem is a finite single-objective Markov game, and the construction in Theorem~\ref{thm:pne-exist} applies. For SER, Proposition~\ref{prop:utility-existence} gives a one-state, one-step two-player game, adapted from R\"{o}pke et al.~\cite{ropke2022nash}, with rewards in $[0,1]^2$ and continuous utilities that are strictly increasing in each coordinate, which has no SER-NE. In contrast, a PNE always exists (Theorem~\ref{thm:pne-exist}).

\begin{proposition}[Inconsistency]\label{prop:utility-inconsistency-main}
There exist a one-state, one-step MOMDP with rewards in $[0,1]^2$ and a concave utility function that is strictly increasing in each coordinate, whose ESR-optimal policy is strictly dominated in expectation by another policy.
\end{proposition}

For example, let the three actions have rewards $(0.4,0.4)$, $(0,1)$ and $(1,0)$, and let $u(\mathbf{v})=\sqrt{v_1}+\sqrt{v_2}$. ESR selects the first action, since $2\sqrt{0.4}>1$, but the equal mixture of the other two actions has expected reward $(0.5,0.5)$, which strictly dominates $(0.4,0.4)$ (Proposition~\ref{prop:esr-pareto}). The reason is that ESR applies the utility before taking the expectation. Under SER, Pareto consistency holds whenever $u_k$ is strictly increasing in each coordinate, since a Pareto improvement of the expected return then strictly increases the utility. Concavity is not needed for this, and without strict monotonicity SER can also select a dominated policy, e.g., with $u(\mathbf{v})=\min\{v_1,v_2\}$.

More importantly, a utility-based solution requires each agent to commit to a utility function in advance. A policy computed for one utility need not remain an equilibrium when the utility changes. In contrast, (W)PNE does not fix a utility in advance. By Theorems~\ref{thm:scalar} and~\ref{thm:wpne}, the PNEs (WPNEs) are exactly the equilibria of the linearly scalarized games with positive (nonnegative) weights. For learning, the useful feature is that the transitions and reward vectors do not depend on these weights. The learner can therefore collect data before a preference profile is specified and then solve the corresponding scalarized game using the learned model. Section~\ref{sec:twophase} gives a uniform error bound that guarantees approximate equilibria for all such queries. We therefore focus on algorithms for finding WPNEs.

\section{Algorithmic Solutions for MOMGs}

We now study learning in MOMGs with unknown transitions and rewards. Section~\ref{sec:single-solution} considers a preference profile that is known before learning, and Section~\ref{sec:twophase} considers preferences that are specified only after data collection. In both cases, data are collected in a centralized way: the learner chooses joint actions and observes the next state and the reward vectors of all players, each in $[0,1]^M$ (Appendix~\ref{sec:nash}).

\subsection{Single Solution Identification}\label{sec:single-solution}

When the preference profile $\bm{\Lambda}\in(\Delta_M)^N$ is fixed in advance, \Cref{eq:pub} reduces the problem to learning an approximate Nash equilibrium of the scalarized game $\mathcal G_{\bm{\Lambda}}$: an $\epsilon$-NE of $\mathcal G_{\bm{\Lambda}}$ is an $\epsilon$-WPNE of the MOMG. Hence any sample-efficient algorithm for general-sum Markov games can be applied to the scalarized rewards $\bm{\lambda}_k^\top\mathbf{r}_{k,h}$.

\textbf{Optimistic Nash value iteration (ONVI-MG).} As a concrete instance, we adapt optimistic Nash value iteration~\cite{liu2021sharp} to MOMGs with unknown transitions and stochastic vector rewards (Algorithm~\ref{alg:onvi-mg} in Appendix~\ref{sec:nash}). In each episode, the algorithm computes optimistic Q-values from the empirical model and exploration bonuses, solves the optimistic stage game at every step and state for a Nash equilibrium by backward induction, and executes the resulting policy. It also evaluates the same policy with pessimistic values. With high probability, the optimistic values upper bound the value of every unilateral deviation from the executed policy, and the pessimistic values lower bound the value of the executed policy itself. The largest difference between the two values at $s_1$, over all players, is therefore a computable upper bound on the policy's scalarized Nash gap. We call it the certificate of the episode. The algorithm returns the policy with the smallest certificate.

ONVI-MG achieves a cumulative Nash regret (Definition~\ref{def:nash regret}) of $\tilde O(H^2S\sqrt{AT})$ over $T$ episodes, where $A=\prod_k|\mathcal{A}_k|$. The certificates admit a bound of the same order, so their average over $T$ episodes is $\tilde O(H^2S\sqrt{A/T})$, and the smallest certificate is no larger than this average. Hence, with probability at least $1-\delta$, the returned policy is an $\epsilon$-WPNE after $\tilde O(H^4S^2A/\epsilon^2)$ episodes, i.e., $\tilde O(H^5S^2A/\epsilon^2)$ environment steps (Theorem~\ref{thm:NASH-reg}). As in Nash value iteration, the algorithm assumes an exact Nash equilibrium solver for the general-sum stage games, so these are sample complexity results rather than computational ones.

This approach requires the preference profile as input. When preferences are unknown during data collection, change over time, or are queried repeatedly, running a separate learning procedure for every profile is wasteful. We address this setting next.

\subsection{Preference-Free Learning for Post-hoc Equilibrium Queries}\label{sec:twophase}

A key observation is that the transitions and reward vectors of an MOMG do not depend on the agents' preferences. We therefore separate learning the model from computing equilibria (Algorithm~\ref{alg:combined}). In Phase~1, the agents explore the environment without any preference information and build an empirical model. In Phase~2, given any preference profile $\bm{\Lambda}$, we compute a Nash equilibrium of the scalarized game on the empirical model, without collecting new samples. The challenge is to explore in a way that makes the learned model accurate enough for every preference profile at once, not only for a single one.

\textbf{Phase 1: exploration.} Since the target preference is unknown, Phase~1 does not try to collect high rewards. Instead, all agents receive the same exploration reward $\bar r^t_h(s,\ba)=\max\{\Phi^t_h(s,\ba)/H,\ \Psi^t_h(s,\ba)\}$, where $\Psi^t_h$ and $\Phi^t_h$ are confidence radii of the reward and transition estimates at $(s,\ba)$ (Lines~\ref{alg1:radii}--\ref{alg1:rbar}). Because this reward is shared by all agents, Phase~1 is a planning problem in a single MDP whose actions are the joint actions. The algorithm computes optimistic values for this MDP by backward induction (Line~\ref{alg1:Q}) and, at every state, executes the joint action with the largest optimistic value (Line~\ref{alg1:argmax}). This steers the agents towards state-action pairs where some part of the model, either a transition probability or some agent's reward in some objective, is still uncertain. Maximizing over all joint actions matters here. A Nash equilibrium of this common-payoff stage game would not suffice, because coordination games can have equilibria that collect less uncertainty than other joint actions. Choosing the maximizing joint action is where the centralized data collection described above is needed; equivalently, the players can share all observations and use the same tie-breaking rule.

\textbf{Phase 2: planning.} For a queried profile $\bm{\Lambda}$, Phase 2 scalarizes the empirical reward estimates and computes a Nash equilibrium of the resulting empirical game (Lines~\ref{alg1:p2}--\ref{alg1:ret}).

\begin{algorithm}[!htb]
\caption{Two-Phase Preference-Free Learning}
\label{alg:combined}
\begin{algorithmic}[1]
\item[] \COMMENT{Phase 1: Preference-free exploration}
\STATE \textbf{Input:} number of episodes $T$, confidence $\delta\in(0,1)$.
\STATE $L\leftarrow\log(16NMSAH(T+1)/\delta)$.
\STATE Initialize $N_h^0(s,\ba)=0$, $\hat{\mathbf r}_{k,h}^0(s,\ba)=\mathbf 0$, and $\hat P_h^0(\cdot\mid s,\ba)=\mathrm{Uniform}(\mathcal S)$ for all $k,h,s,\ba$.
\FOR{$t=1,\ldots,T$}
    \STATE Set $\bar V_{H+1}^t\equiv0$.
    \FOR{$h=H,\ldots,1$}
        \STATE For each $(s,\ba)$, compute
        \STATE $\displaystyle\Psi_h^t(s,\ba)\leftarrow\min\Big\{1,\sqrt{\frac{2L}{N_h^{t-1}(s,\ba)\vee1}}\Big\},$\label{alg1:radii}
        \STATE $\displaystyle\Phi_h^t(s,\ba)\leftarrow\min\Big\{H,H\sqrt{\frac{2SL}{N_h^{t-1}(s,\ba)\vee1}}\Big\},$
        \STATE $\bar r_h^t(s,\ba)\leftarrow\max\{\Phi_h^t(s,\ba)/H,\Psi_h^t(s,\ba)\}$,\label{alg1:rbar}
        \STATE $\bar Q_h^t(s,\ba)\leftarrow\min\{H-h+1,$\label{alg1:Q}
        \item[] \hspace{\algorithmicindent}
        $\bar r_h^t(s,\ba)+\sum_{s'}\hat P_h^{t-1}(s'\mid s,\ba)\bar V_{h+1}^t(s')+\Phi_h^t(s,\ba)\}$.
        \STATE For each $s$, choose $\bar{\ba}_h^t(s)\in\arg\max_{\ba}\bar Q_h^t(s,\ba)$ with a fixed deterministic tie-breaking rule, and set $\bar V_h^t(s)\leftarrow\bar Q_h^t(s,\bar{\ba}_h^t(s))$.\label{alg1:argmax}
    \ENDFOR
    \STATE Starting from $s_1$, play $\ba_h^t=\bar{\ba}_h^t(s_h^t)$ at each step $h$, and record the next state and all players' reward vectors.
    \STATE Update $N_h^t$, $\hat{\mathbf r}_{k,h}^t$, and $\hat P_h^t$ at the visited triples $(h,s_h^t,\ba_h^t)$.
\ENDFOR
\STATE \textbf{Output:} the empirical model $\hat{\mathbf r}_{k,h}:=\hat{\mathbf r}_{k,h}^T$ and $\hat P_h:=\hat P_h^T$.
\item[] \COMMENT{Phase 2: Planning for a preference query}
\STATE \textbf{Input:} the model from Phase~1 and a profile $\bm{\Lambda}=(\bm{\lambda}_1,\ldots,\bm{\lambda}_N)\in(\Delta_M)^N$.\label{alg1:p2}
\STATE Compute a Nash equilibrium $\hat{\bm{\pi}}_{\bm{\Lambda}}$ of the empirical game with transitions $\hat P_h$ in which player $k$ receives the reward $\bm{\lambda}_k^\top\hat{\mathbf r}_{k,h}$, using an exact general-sum Nash oracle.
\STATE \textbf{Return:} $\hat{\bm{\pi}}_{\bm{\Lambda}}$.\label{alg1:ret}
\end{algorithmic}
\end{algorithm}

The following theorem shows that one run of Phase~1, consisting of $HT=\tilde O(H^5S^2A/\epsilon^2)$ environment steps, suffices for all later queries.

\begin{theorem}\label{thm:pf-main}
Run Algorithm~\ref{alg:combined} with $T=\tilde O(H^4S^2A/\epsilon^2)$ exploration episodes in Phase~1. With probability at least $1-\delta$, the following holds simultaneously for every preference profile $\bm{\Lambda}\in(\Delta_M)^N$: the policy $\hat{\bm{\pi}}_{\bm{\Lambda}}$ returned by Phase~2 is an $\epsilon$-NE of the true scalarized game $\mathcal G_{\bm{\Lambda}}$, and hence an $\epsilon$-WPNE of the MOMG.
\end{theorem}

\begin{proof}[Proof sketch (Theorem~\ref{thm:pf-uniform})]
Phase~2 computes an exact Nash equilibrium of the empirical scalarized game. It therefore remains to control how model estimation errors affect the equilibrium guarantee.

\emph{Uniform value accuracy.} On one event of probability at least $1-\delta$, the empirical and true scalarized values at $s_1$ differ by at most $\eta_T$, simultaneously for every player $k$, preference $\bm{\lambda}\in\Delta_M$, and product Markov policy $\bm{\pi}$. Instead of a union bound over finitely many rewards \citep{liu2021sharp}, we bound every reward coordinate at once: all reward-coordinate errors are bounded by $\Psi$, so every convex combination $\bm{\lambda}^\top\hat{\mathbf r}_k$ has estimation error at most $\Psi$. The transition bounds are independent of preferences. The resulting value bound therefore also applies to preferences chosen after observing the data.

\emph{From values to equilibria.} Every unilateral deviation from $\hat{\bm{\pi}}_{\bm{\Lambda}}$ has nonpositive gain in the empirical game. Its true gain differs from its empirical gain by at most $2\eta_T$, accounting for the value errors of both the original and deviating policies. Hence $\hat{\bm{\pi}}_{\bm{\Lambda}}$ is a $2\eta_T$-NE of $\mathcal G_{\bm{\Lambda}}$, and \eqref{eq:pub} gives the WPNE claim.

\emph{Number of episodes.} Each policy's value error is bounded by a Bellman recursion under the true transitions that accumulates the final confidence radii $\Psi+\Phi$ as per-step terms and is clipped at each step at the remaining horizon (Lemma~\ref{lem:value_diff}). Because the radii shrink as visits accumulate, and Phase~1 plans with the bonus $\Phi$ and maximizes over joint actions, this recursion is at most $2\bar V_1^t(s_1)$ for every episode $t$ (Lemma~\ref{lem:optimism-2}). Since $\sum_{t=1}^T\bar V_1^t(s_1)=\tilde O(H^2S\sqrt{AT})$ (Lemma~\ref{lem:bonus_sum}), averaging these comparisons gives $\eta_T=\tilde O(H^2S\sqrt{A/T})$. Requiring $2\eta_T\le\epsilon$ yields the stated episode complexity.
\end{proof}

\begin{remark}
The sample bound is polynomial in the joint action space size $A=\prod_k|\mathcal{A}_k|$, which can grow exponentially with the number of players. Apart from this dependence through $A$, the numbers of players and objectives enter only through logarithmic factors. The episode complexity has the same order as that of Section~\ref{sec:single-solution} for a single fixed profile and that of Multi-VI-Zero~\cite{liu2021sharp} for finitely many reward functions fixed in advance. Phase~2 assumes an exact general-sum Nash solver; with an $\alpha$-approximate solver, the same argument gives an $(\alpha+\epsilon)$-NE. On the same event, the learned model also supports policy evaluation: for every product Markov policy, the WPNG computed on the empirical model, for instance by the linear program in Appendix~\ref{sec:wpne-appendix}, is within $\epsilon$ of its true value (Theorem~\ref{thm:pf-uniform}). Theorem~\ref{thm:pf-main} guarantees that every query returns an approximate equilibrium. It does not imply that the equilibria returned over all queries cover every equilibrium of the true game, since a solver may return only one of several equilibria for the same preference.
\end{remark}
\section{Experiments}
In this section we develop numerical experiments to validate our results.
\subsection{Scalarized Equilibria and Weak Pareto-Nash Gaps}
We first empirically verify our results on the equivalence between Weak Pareto-Nash Equilibria and NEs of linear scalarized games.  We construct three categories of games: (i) four structured normal-form MOMGs with distinct strategic properties (MO-Matching Pennies, MO-Prisoner's Dilemma, MO-Hawk-Dove, and MO-Coordination); (ii) random normal-form MOMGs with varying action spaces ($\mathcal A \in \{2 \times 2, 3 \times 3\}$) and numbers of objectives ($M \in \{2, 3\}$); and (iii) tabular MOMGs with multiple states and time horizons, varying across $H \in \{1,2,3,4\}$, $\mathcal S \in \{2,3,4\}$, $M \in \{2,3\}$, and $\mathcal A \in \{2 \times 2, 3 \times 3\}$. For each game configuration, we instantiate multiple random game instances and, for each instance, sample different preference profiles $\bm\Lambda$ from $\Delta_M$. We then compute approximate scalarized NEs via multiplicative weights (for normal-form games) or value iteration with a multiplicative-weights update at each stage game (for Markov games). For each solution, we compute its scalarized Nash gap by dynamic programming and its WPNG by the linear program in Appendix~\ref{sec:wpne-appendix}.

Figure~\ref{fig:wpng} reports the results. The left panel plots the WPNG of each solution against its scalarized Nash gap under the sampled profile $\bm\Lambda$: the scalarized Nash gap measures how far the solution is from an NE of $\mathcal G_{\boldsymbol{\Lambda}}$, and the WPNG measures how far it is from a WPNE. All points lie on or below the dashed line, i.e., the WPNG of each solution never exceeds its scalarized Nash gap, as predicted by~\Cref{eq:pub}. Hence the residual non-zero WPNGs are attributable to the approximate NE solver: solutions near the origin are close to an NE of $\mathcal G_{\boldsymbol{\Lambda}}$, while solutions with a larger Nash gap, e.g., in MO-Matching Pennies, whose only equilibrium is fully mixed, can also have a larger WPNG. The right panel reports the WPNG of each game type, which is low across most games. When the scalarized games are solved exactly, by backward induction with support enumeration at every stage game, the WPNG is $0$ up to floating-point error. These results hence confirm that every NE of a scalarized game is indeed a WPNE of the original MOMG, verifying our theoretical results.
 
\begin{figure*}[!t]
    \centering
    \includegraphics[width=\textwidth]{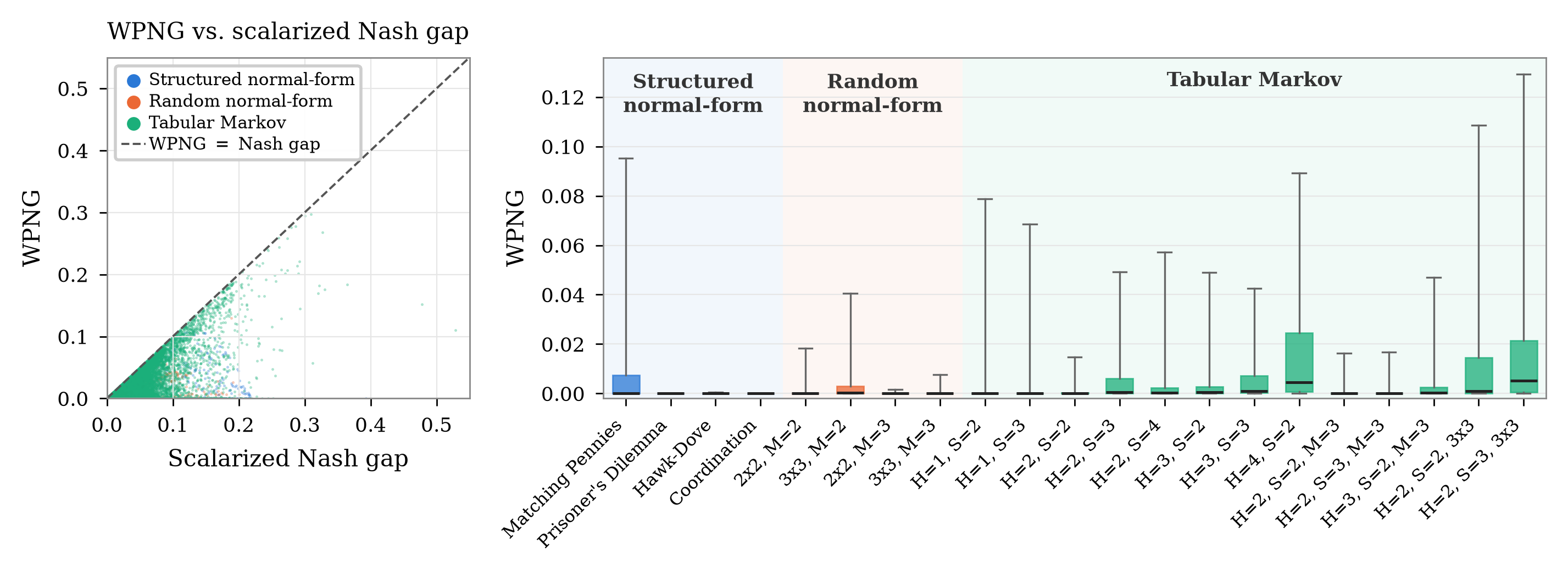}
    \caption{WPNGs of Approximate Scalarized NEs, computed exactly by the linear program in Appendix~\ref{sec:wpne-appendix}. Left: WPNG against the scalarized Nash gap of each solution; the dashed line marks WPNG $=$ Nash gap. Right: each column is a game type. Boxes show interquartile range; black lines indicate medians; whiskers the 2nd and 98th percentiles.}
    \label{fig:wpng}
\end{figure*}

\subsection{Experiments on Multi-Objective Multi-Agent Environments}
\label{sec:exp-envs}

To further evaluate our framework beyond the tabular verifications, we test in four multi-objective multi-agent environments, each with two agents and two objectives and with options that undirected exploration rarely finds (Appendix~\ref{sec:exp-envs-details}).

\textbf{Environments.}
\label{sec:momg exp setup}
\textit{MO-Predator Prey.}
Two predators share capture rewards (Objective~1) but each pays its own movement cost (Objective~2), creating a free-rider tension. Small prey can be caught near the start; a stag worth ten times as much roams a distant den.

\textit{MO-Traffic Control.}
Two agents control the signals of adjacent intersections, trading throughput (Objective~1) against queue fairness (Objective~2). Outflow from one intersection joins the other's queues, coupling the agents; keeping the main road green for $20$ steps joins a high-throughput green wave.

\textit{MO-Resource Race.}
Two agents compete for a red source near their start and blue sources in the far corner (Objectives~1 and~2 count red and blue collections). Depending on $\boldsymbol\Lambda$, they specialize (one on red, the other on blue) or compete head-to-head.

\textit{MO-Negotiation.}
Each round, two agents simultaneously demand low, medium, or high; Objective~1 measures material gain and Objective~2 cooperation quality. Demanding low for $12$ consecutive rounds builds, at the builder's cost, a reputation that turns later low demands into partnerships paying both agents more.

We compare five methods, which share the same tabular learner and differ only in how they explore. \textbf{Two-Phase + EQ (Ours)}: a tabular implementation of \Cref{alg:combined}, whose count-based Error Quantification (EQ) terms play the role of $\Psi,\Phi$ and guide preference-free exploration to support subsequent preference queries. \textbf{Two-Phase (no EQ)}: ablation removing the EQ terms, i.e., exploring at random in Phase~1. \textbf{Fixed $\lambda$ (equal)}: standard MARL trained with a single fixed equal-weight scalarization. \textbf{Independent MORL}: each agent independently trains with random preferences, without an exploration phase. \textbf{Independent MORL + EQ}: Independent MORL with the EQ terms added to its scalarized reward, so that its exploration is driven by the sampled preferences.

Standard multi-objective metrics for single-agent MORL, such as hypervolume and average utility (AUT), can fail to account for strategic stability: a method may achieve high AUT with policies that are not equilibria and would collapse under rational opponents. Therefore, we adopt \emph{strategic metrics} that test the learned policies against unilateral deviations. Specifically, given a returned joint policy $\pi=(\pi_1,\pi_2)$ and a test preference $\boldsymbol\Lambda=(\boldsymbol\lambda_1,\boldsymbol\lambda_2)$, we search for each player $i$ the best deviation $\pi_i'$ under $\boldsymbol\lambda_i$ and measure the following metrics:
(1) \textbf{AUT} ($\uparrow$): the average scalarized utility when both agents follow $\pi$;
(2) \textbf{Nash gap} ($\downarrow$): the larger deviation gain $\boldsymbol\lambda_i^\top \mathbf V_i^{(\pi_i',\pi_{-i})}-\boldsymbol\lambda_i^\top \mathbf V_i^{\pi}$ of the two players, averaged over test preferences, a lower bound on the scalarized Nash gap that \Cref{thm:pf-main} bounds;
(3) \textbf{Worst-case} ($\downarrow$): the largest Nash gap over test preferences, which reflects the uniform guarantee of \Cref{thm:pf-main}.
A method with high AUT but large Nash gap can produce solutions that appear good but are strategically unstable; a method with small Nash gap but low AUT is stable but settles for poor equilibria. Only a method that achieves both high AUT and small Nash gap provides solutions that are both valuable and deployable under rational opponents.

We train each method with $10$ random seeds and evaluate the metrics with $60$ randomly selected test preferences per seed. \Cref{tab:strategic} reports the results across four environments. Our algorithm achieves the lowest Nash gap, the lowest worst-case Nash gap, and the highest AUT in all four environments, and all differences are significant (Mann--Whitney $p\le0.006$). Without EQ, the agents tend to exploit early-discovered rewards and under-explore the states that are critical for other preferences: the baselines without EQ reach them in fewer than $0.05\%$ of the training episodes, compared with $30$ to $48\%$ for our algorithm, so their policies are farther from equilibrium. Adding the EQ terms to preference-driven exploration (Independent MORL + EQ) does not close this gap, and when these states are made easy to reach, the advantage of our algorithm disappears. These results confirm that comprehensive, preference-free exploration is essential for learning policies that remain close to equilibrium under diverse preference profiles, in line with \Cref{thm:pf-main}.

\begin{table*}[!htb]
\centering
\caption{Strategic metrics on MO environments: mean and 95\% confidence interval over $10$ seeds, each evaluated on $60$ test preference profiles.}
\label{tab:strategic}
\small
\begin{tabular}{llccc}
\toprule
Environment & Method & Nash gap ($\downarrow$) & Worst-case ($\downarrow$) & AUT ($\uparrow$) \\
\midrule
MO-Predator Prey   & \textbf{Two-Phase + EQ (Ours)} & $\mathbf{10.5}\pm1.2$ & $\mathbf{43.4}\pm2.9$ & $\mathbf{19.2}\pm1.5$ \\
                   & Two-Phase (no EQ) & $26.0\pm0.9$ & $52.6\pm1.6$ & $3.4\pm0.2$ \\
                   & Fixed $\lambda$ (equal) & $26.3\pm0.9$ & $50.2\pm1.3$ & $4.4\pm0.1$ \\
                   & Indep.\ MORL & $25.5\pm1.1$ & $50.4\pm2.9$ & $3.2\pm0.2$ \\
                   & Indep.\ MORL + EQ & $25.6\pm1.5$ & $53.0\pm2.7$ & $2.1\pm1.1$ \\
\addlinespace
MO-Traffic Control & \textbf{Two-Phase + EQ (Ours)} & $\mathbf{7.1}\pm0.8$ & $\mathbf{13.7}\pm1.3$ & $\mathbf{15.0}\pm0.8$ \\
                   & Two-Phase (no EQ) & $14.2\pm0.6$ & $26.7\pm0.8$ & $7.7\pm0.2$ \\
                   & Fixed $\lambda$ (equal) & $13.7\pm0.4$ & $25.9\pm0.4$ & $8.0\pm0.2$ \\
                   & Indep.\ MORL & $14.0\pm0.4$ & $26.0\pm0.5$ & $8.1\pm0.1$ \\
                   & Indep.\ MORL + EQ & $13.9\pm0.8$ & $26.7\pm1.0$ & $7.8\pm0.3$ \\
\addlinespace
MO-Resource Race   & \textbf{Two-Phase + EQ (Ours)} & $\mathbf{3.1}\pm0.5$ & $\mathbf{14.0}\pm0.5$ & $\mathbf{19.0}\pm0.8$ \\
                   & Two-Phase (no EQ) & $18.4\pm2.0$ & $38.2\pm1.1$ & $9.2\pm0.5$ \\
                   & Fixed $\lambda$ (equal) & $18.9\pm2.3$ & $37.9\pm1.0$ & $9.0\pm0.6$ \\
                   & Indep.\ MORL & $18.1\pm1.7$ & $37.7\pm1.2$ & $9.4\pm0.2$ \\
                   & Indep.\ MORL + EQ & $15.5\pm2.5$ & $30.5\pm1.2$ & $13.2\pm1.2$ \\
\addlinespace
MO-Negotiation     & \textbf{Two-Phase + EQ (Ours)} & $\mathbf{5.8}\pm0.9$ & $\mathbf{16.2}\pm0.4$ & $\mathbf{77.9}\pm0.5$ \\
                   & Two-Phase (no EQ) & $18.8\pm0.6$ & $31.7\pm1.0$ & $61.8\pm0.2$ \\
                   & Fixed $\lambda$ (equal) & $20.6\pm0.7$ & $38.2\pm1.1$ & $59.8\pm0.4$ \\
                   & Indep.\ MORL & $18.5\pm0.9$ & $32.4\pm0.7$ & $61.9\pm0.3$ \\
                   & Indep.\ MORL + EQ & $20.3\pm0.8$ & $37.3\pm3.4$ & $61.2\pm0.6$ \\
\bottomrule
\end{tabular}
\end{table*}

To further enhance the scalability of our approach, we develop a deep-learning extension of our algorithm and conduct more experiments. We defer the discussions to Appendix~\ref{sec:deep exp}.

\section{Conclusion}
In this paper, we introduced MOMGs to model multi-agent systems in which each agent has several objectives. We proposed the PNE and its weak variant WPNE as solution concepts, proved that PNEs always exist, and showed that PNEs and WPNEs are exactly the Nash equilibria of linearly scalarized games with strictly positive and nonnegative weights, respectively. Because the WPNG is bounded above by the Nash gap of every scalarized game, learning an approximate WPNE for a preference profile fixed in advance reduces to learning in one scalarized Markov game. For preferences specified after data collection, our two-phase algorithm learns one model from which approximate WPNEs can be computed for every preference profile. Experiments corroborate the WPNG bound and show that our two-phase algorithm achieves lower estimated Nash gaps and higher average utility than the baselines across four environments. Our guarantees have three main limitations: the sample complexity is polynomial in the size of the joint action space, which grows exponentially with the number of players; both algorithms assume centralized data collection; and both assume an exact solver for general-sum Nash equilibria. Natural next steps are nonlinear preferences, decentralized exploration, guarantees under function approximation for large or continuous spaces, and correlated variants of the WPNE that need no Nash solver.

\newpage
\bibliography{ref}

@string{aaai="Proc. Conference on Artificial Intelligence (AAAI)"}

@string{colt="Proc. Annual Conference on Learning Theory (CoLT)"}

@string{icml="Proc. International Conference on Machine Learning (ICML)"}

@string{nips="Proc. Advances in Neural Information Processing Systems (NIPS)"}

@string{nipsnew="Proc. Advances in Neural Information Processing Systems (NeurIPS)"}

@inproceedings{reymond2022pareto,
author = {Reymond, Mathieu and Bargiacchi, Eugenio and Now\'{e}, Ann},
title = {Pareto Conditioned Networks},
year = {2022},
isbn = {9781450392136},
publisher = {International Foundation for Autonomous Agents and Multiagent Systems},
address = {Richland, SC},
booktitle = {Proceedings of the 21st International Conference on Autonomous Agents and Multiagent Systems},
pages = {1110–1118},
numpages = {9},
location = {Virtual Event, New Zealand},
series = {AAMAS '22}
}

@article{yu2022surprising,
  title={The surprising effectiveness of ppo in cooperative multi-agent games},
  author={Yu, Chao and Velu, Akash and Vinitsky, Eugene and Gao, Jiaxuan and Wang, Yu and Bayen, Alexandre and Wu, Yi},
  journal={Advances in neural information processing systems},
  volume={35},
  pages={24611--24624},
  year={2022}
}

@article{felten2024momaland,
  title={Momaland: A set of benchmarks for multi-objective multi-agent reinforcement learning},
  author={Felten, Florian and Ucak, Umut and Azmani, Hicham and Peng, Gao and R{\"o}pke, Willem and Baier, Hendrik and Mannion, Patrick and Roijers, Diederik M and Terry, Jordan K and Talbi, El-Ghazali and others},
  journal={arXiv preprint arXiv:2407.16312},
  year={2024}
}

@inproceedings{kallstrom2019tunable,
  title={Tunable dynamics in agent-based simulation using multi-objective reinforcement learning},
  author={K{\"a}llstr{\"o}m, Johan and Heintz, Fredrik},
  booktitle={Adaptive and Learning Agents Workshop (ALA-19) at AAMAS, Montreal, Canada, May 13-14, 2019},
  pages={1--7},
  year={2019}
}

@article{matignon2012independent,
  title={Independent reinforcement learners in cooperative markov games: a survey regarding coordination problems},
  author={Matignon, Laetitia and Laurent, Guillaume J and Le Fort-Piat, Nadine},
  journal={The Knowledge Engineering Review},
  volume={27},
  number={1},
  pages={1--31},
  year={2012},
  publisher={Cambridge University Press}
}

@article{mi2023taxai,
  title={Taxai: A dynamic economic simulator and benchmark for multi-agent reinforcement learning},
  author={Mi, Qirui and Xia, Siyu and Song, Yan and Zhang, Haifeng and Zhu, Shenghao and Wang, Jun},
  journal={arXiv preprint arXiv:2309.16307},
  year={2023}
}

@article{huang2024multi,
  title={A multi-agent reinforcement learning framework for optimizing financial trading strategies based on TimesNet},
  author={Huang, Yuling and Zhou, Chujin and Cui, Kai and Lu, Xiaoping},
  journal={Expert Systems with Applications},
  volume={237},
  pages={121502},
  year={2024},
  publisher={Elsevier}
}

@inproceedings{abrol2024deep,
  title={A Deep Reinforcement Learning Approach for Adaptive Traffic Routing in Next-gen Networks},
  author={Abrol, Akshita and Mohan, Purnima Murali and Truong-Huu, Tram},
  booktitle={ICC 2024-IEEE International Conference on Communications},
  pages={465--471},
  year={2024},
  organization={IEEE}
}

@article{shabestary2022adaptive,
  title={Adaptive traffic signal control with deep reinforcement learning and high dimensional sensory inputs: Case study and comprehensive sensitivity analyses},
  author={Shabestary, Soheil Mohamad Alizadeh and Abdulhai, Baher},
  journal={IEEE Transactions on Intelligent Transportation Systems},
  volume={23},
  number={11},
  pages={20021--20035},
  year={2022},
  publisher={IEEE}
}

@article{li2024find,
  title={How to Find the Exact Pareto Front for Multi-Objective MDPs?},
  author={Li, Yining and Ju, Peizhong and Shroff, Ness B},
  journal={arXiv preprint arXiv:2410.15557},
  year={2024}
}

@inproceedings{gupta2017cooperative,
  title={Cooperative multi-agent control using deep reinforcement learning},
  author={Gupta, Jayesh K and Egorov, Maxim and Kochenderfer, Mykel},
  booktitle={International Conference on Autonomous Agents and Multiagent Systems},
  pages={66--83},
  year={2017}
}

@incollection{littman1994markov,
  title={Markov games as a framework for multi-agent reinforcement learning},
  author={Littman, Michael L},
  booktitle={Machine learning proceedings 1994},
  pages={157--163},
  year={1994},
  publisher={Elsevier}
}

@article{hu2003nash,
  title={Nash Q-learning for general-sum stochastic games},
  author={Hu, Junling and Wellman, Michael P},
  journal={Journal of machine learning research},
  volume={4},
  number={Nov},
  pages={1039--1069},
  year={2003}
}

@article{shapley1953stochastic,
  title={Stochastic games},
  author={Shapley, Lloyd S},
  journal={Proceedings of the national academy of sciences},
  volume={39},
  number={10},
  pages={1095--1100},
  year={1953},
  publisher={National Academy of Sciences}
}

@article{jin2021v,
  title={V-Learning--A Simple, Efficient, Decentralized Algorithm for Multiagent RL},
  author={Jin, Chi and Liu, Qinghua and Wang, Yuanhao and Yu, Tiancheng},
  journal={arXiv preprint arXiv:2110.14555},
  year={2021}
}

@article{zhang2021multi,
  title={Multi-agent reinforcement learning: A selective overview of theories and algorithms},
  author={Zhang, Kaiqing and Yang, Zhuoran and Ba{\c{s}}ar, Tamer},
  journal={Handbook of reinforcement learning and control},
  pages={321--384},
  year={2021},
  publisher={Springer}
}

@article{hernandez2019survey,
  title={A survey and critique of multiagent deep reinforcement learning},
  author={Hernandez-Leal, Pablo and Kartal, Bilal and Taylor, Matthew E},
  journal={Autonomous Agents and Multi-Agent Systems},
  volume={33},
  number={6},
  pages={750--797},
  year={2019},
  publisher={Springer}
}

@inproceedings{bai2020provable,
  title={Provable self-play algorithms for competitive reinforcement learning},
  author={Bai, Yu and Jin, Chi},
  booktitle=icml,
  pages={551--560},
  year={2020},
  organization={PMLR}
}

@inproceedings{xie2020learning,
  title={Learning zero-sum simultaneous-move markov games using function approximation and correlated equilibrium},
  author={Xie, Qiaomin and Chen, Yudong and Wang, Zhaoran and Yang, Zhuoran},
  booktitle=colt,
  pages={3674--3682},
  year={2020},
  organization={PMLR}
}

@book{fudenberg1991game,
  title={Game theory},
  author={Fudenberg, Drew and Tirole, Jean},
  year={1991},
  publisher={MIT press}
}

@book{filar2012competitive,
  title={Competitive Markov decision processes},
  author={Filar, Jerzy and Vrieze, Koos},
  year={2012},
  publisher={Springer Science \& Business Media}
}

@book{wooldridge2009introduction,
  title={An introduction to multiagent systems},
  author={Wooldridge, Michael},
  year={2009},
  publisher={John wiley \& sons}
}

@article{chen2022decentralized,
  title={Decentralized Policy Gradient for Nash Equilibria Learning of General-sum Stochastic Games},
  author={Chen, Yan and Li, Tao},
  journal={arXiv preprint arXiv:2210.07651},
  year={2022}
}

@Article{wu2021offline,
  author    = {Wu, Runzhe and Zhang, Yufeng and Yang, Zhuoran and Wang, Zhaoran},
  title     = {Offline Constrained Multi-Objective Reinforcement Learning via Pessimistic Dual Value Iteration},
  volume    = {34},
  booktitle = nipsnew,
  journal   = {Advances in Neural Information Processing Systems},
  year      = {2021},
}

@InProceedings{liu2021sharp,
  author       = {Liu, Qinghua and Yu, Tiancheng and Bai, Yu and Jin, Chi},
  booktitle    = icml,
  title        = {A sharp analysis of model-based reinforcement learning with self-play},
  organization = {PMLR},
  pages        = {7001--7010},
  year         = {2021},
}

@InProceedings{lowe2017multi,
  author    = {Lowe, Ryan and Wu, Yi and Tamar, Aviv and Harb, Jean and Abbeel, Pieter and Mordatch, Igor},
  booktitle = nips,
  title     = {Multi-agent actor-critic for mixed cooperative-competitive environments},
  pages     = {6379--6390},
  year      = {2017},
}

@book{shoham2008multiagent,
  title={Multiagent systems: Algorithmic, game-theoretic, and logical foundations},
  author={Shoham, Yoav and Leyton-Brown, Kevin},
  year={2008},
  publisher={Cambridge University Press}
}

@inproceedings{zhang2018fully,
  title={Fully decentralized multi-agent reinforcement learning with networked agents},
  author={Zhang, Kaiqing and Yang, Zhuoran and Liu, Han and Zhang, Tong and Basar, Tamer},
  booktitle=icml,
  pages={5872--5881},
  year={2018},
  organization={PMLR}
}

@Book{puterman2014markov,
  author     = {Puterman, Martin L},
  title      = {Markov decision processes: discrete stochastic dynamic programming},
  edition    = {1},
  isbn       = {978-0-471-61977-2 978-0-470-31688-7},
  language   = {en},
  publisher  = {John Wiley \& Sons},
  series     = {Wiley {Series} in {Probability} and {Statistics}},
  urldate    = {2021-12-15},
  month      = apr,
  shorttitle = {Markov {Decision} {Processes}},
  year       = {2014},
}

@inproceedings{cui2023breaking,
  title={Breaking the curse of multiagents in a large state space: Rl in markov games with independent linear function approximation},
  author={Cui, Qiwen and Zhang, Kaiqing and Du, Simon},
  booktitle={The Thirty Sixth Annual Conference on Learning Theory},
  pages={2651--2652},
  year={2023},
  organization={PMLR}
}

@article{qiu2024traversing,
  title={Traversing pareto optimal policies: Provably efficient multi-objective reinforcement learning},
  author={Qiu, Shuang and Zhang, Dake and Yang, Rui and Lyu, Boxiang and Zhang, Tong},
  journal={arXiv preprint arXiv:2407.17466},
  year={2024}
}

@article{zhao1991equilibria,
  title={The equilibria of a multiple objective game},
  author={Zhao, Jingang},
  journal={International Journal of Game Theory},
  volume={20},
  number={2},
  pages={171--182},
  year={1991}
}

@inproceedings{sridhar2012pareto,
  title={Pareto optimal allocation in multi-agent coalitional games with non-linear payoffs},
  author={Sridhar, Usha and Mandyam, Sridhar},
  booktitle={2012 IEEE/ACM International Conference on Advances in Social Networks Analysis and Mining},
  pages={1301--1308},
  year={2012},
  organization={IEEE}
}

@article{lozovanu2005multiobjective,
  title={Multiobjective games and determining pareto-nash equilibria},
  author={Lozovanu, Dmitrii and Solomon, Dumitru and Zelikovsky, Alexander},
  journal={Buletinul Academiei de {\c{S}}tiin{\c{t}}e a Republicii Moldova. Matematica},
  volume={49},
  number={3},
  pages={115--122},
  year={2005}
}

@article{rodriguez2024analytical,
  title={An analytical study of utility functions in multi-objective reinforcement learning},
  author={Rodr{\'\i}guez Soto, Manel and Rodriguez-Aguilar, Juan A and Lopez-Sanchez, Maite},
  journal={Advances in Neural Information Processing Systems},
  volume={37},
  pages={77726--77747},
  year={2024}
}

@article{ropke2022nash,
  title={On nash equilibria in normal-form games with vectorial payoffs},
  author={R{\"o}pke, Willem and Roijers, Diederik M and Now{\'e}, Ann and R{\u{a}}dulescu, Roxana},
  journal={Autonomous Agents and Multi-Agent Systems},
  volume={36},
  number={2},
  pages={53},
  year={2022},
  publisher={Springer}
}

@article{ruadulescu2020multi,
  title={Multi-objective multi-agent decision making: a utility-based analysis and survey},
  author={R{\u{a}}dulescu, Roxana and Mannion, Patrick and Roijers, Diederik M and Now{\'e}, Ann},
  journal={Autonomous Agents and Multi-Agent Systems},
  volume={34},
  number={1},
  pages={10},
  year={2020},
  publisher={Springer}
}

@article{borm1988pareto,
  title={Pareto equilibria in multiobjective games},
  author={Borm, PEM and Tijs, SH and Van Den Aarssen, JCM},
  journal={Methods of Operations Research},
  volume={60},
  pages={302--312},
  year={1988},
  publisher={Verlag Anton Hain Meisenheim CmbH}
}

@article{ropke2023bridging,
  title={Bridging the gap between single and multi objective games},
  author={R{\"o}pke, Willem and Groenland, Carla and R{\u{a}}dulescu, Roxana and Now{\'e}, Ann and Roijers, Diederik M},
  journal={arXiv preprint arXiv:2301.05755},
  year={2023}
}

@article{lu2019multi,
  title={Multi-objective generalized linear bandits},
  author={Lu, Shiyin and Wang, Guanghui and Hu, Yao and Zhang, Lijun},
  journal={arXiv preprint arXiv:1905.12879},
  year={2019}
}

@inproceedings{turgay2018multi,
  title={Multi-objective contextual bandit problem with similarity information},
  author={Turgay, Eralp and Oner, Doruk and Tekin, Cem},
  booktitle={International Conference on Artificial Intelligence and Statistics},
  pages={1673--1681},
  year={2018},
  organization={PMLR}
}

@article{chang2014partially,
  title={Partially Observed, Multi-objective Markov Games},
  author={Chang, Yanling and Erera, Alan L and White III, Chelsea C},
  journal={arXiv preprint arXiv:1404.4388},
  year={2014}
}

@inproceedings{van2013scalarized,
  title={Scalarized multi-objective reinforcement learning: Novel design techniques},
  author={Van Moffaert, Kristof and Drugan, Madalina M and Now{\'e}, Ann},
  booktitle={2013 IEEE symposium on adaptive dynamic programming and reinforcement learning (ADPRL)},
  pages={191--199},
  year={2013},
  organization={IEEE}
}

@inproceedings{wiering2014model,
  title={Model-based multi-objective reinforcement learning},
  author={Wiering, Marco A and Withagen, Maikel and Drugan, M{\u{a}}d{\u{a}}lina M},
  booktitle={2014 IEEE symposium on adaptive dynamic programming and reinforcement learning (ADPRL)},
  pages={1--6},
  year={2014},
  organization={IEEE}
}

@article{eriksson2008adaptive,
  title={Adaptive foresight: navigating the complex landscape of policy strategies},
  author={Eriksson, E Anders and Weber, K Matthias},
  journal={Technological Forecasting and Social Change},
  volume={75},
  number={4},
  pages={462--482},
  year={2008},
  publisher={Elsevier}
}

@article{osika2024navigating,
  title={Navigating Trade-offs: Policy Summarization for Multi-Objective Reinforcement Learning},
  author={Osika, Zuzanna and Zatarain-Salazar, Jazmin and Oliehoek, Frans A and Murukannaiah, Pradeep K},
  journal={arXiv preprint arXiv:2411.04784},
  year={2024}
}

@article{chapman2023bridging,
  title={Bridging adaptive management and reinforcement learning for more robust decisions},
  author={Chapman, Melissa and Xu, Lily and Lapeyrolerie, Marcus and Boettiger, Carl},
  journal={Philosophical Transactions of the Royal Society B},
  volume={378},
  number={1881},
  pages={20220195},
  year={2023},
  publisher={The Royal Society}
}

@article{addy2024machine,
  title={Machine learning in financial markets: A critical review of algorithmic trading and risk management},
  author={Addy, Wilhelmina Afua and Ajayi-Nifise, Adeola Olusola and Bello, Binaebi Gloria and Tula, Sunday Tubokirifuruar and Odeyemi, Olubusola and Falaiye, Titilola},
  journal={International Journal of Science and Research Archive},
  volume={11},
  number={1},
  pages={1853--1862},
  year={2024}
}

@article{olanrewaju2025artificial,
  title={Artificial Intelligence in Financial Markets: Optimizing Risk Management, Portfolio Allocation, and Algorithmic Trading},
  author={Olanrewaju, Ayobami Gabriel},
  journal={International Journal of Research Publication and Reviews},
  volume={6},
  pages={8855--8870},
  year={2025}
}

@inproceedings{bhardwaj2024risk,
  title={Risk Assessment and Management in Stock Trading using Artificial Intelligence},
  author={Bhardwaj, Alok and Ranjan, Onima and Biswas, Susmi and Gupta, Lucky and Chanti, Yerrolla and Sharma, Meenakshi},
  booktitle={2024 3rd International Conference on Sentiment Analysis and Deep Learning (ICSADL)},
  pages={138--145},
  year={2024},
  organization={IEEE}
}

@incollection{deb2016multi,
  title={Multi-objective optimization},
  author={Deb, Kalyanmoy and Sindhya, Karthik and Hakanen, Jussi},
  booktitle={Decision sciences},
  pages={161--200},
  year={2016},
  publisher={CRC Press}
}

@inproceedings{chatterjee2006markov,
  title={Markov decision processes with multiple objectives},
  author={Chatterjee, Krishnendu and Majumdar, Rupak and Henzinger, Thomas A},
  booktitle={Annual symposium on theoretical aspects of computer science},
  pages={325--336},
  year={2006},
  organization={Springer}
}

@article{wakuta1998solution,
  title={Solution procedures for multi-objective Markov decision processes},
  author={Wakuta, K and Togawa, K},
  journal={Optimization},
  volume={43},
  number={1},
  pages={29--46},
  year={1998},
  publisher={Taylor \& Francis}
}

@inproceedings{wiering2007computing,
  title={Computing optimal stationary policies for multi-objective markov decision processes},
  author={Wiering, Marco A and De Jong, Edwin D},
  booktitle={2007 IEEE international symposium on approximate dynamic programming and reinforcement learning},
  pages={158--165},
  year={2007},
  organization={IEEE}
}

@article{busoniu2008comprehensive,
  title={A comprehensive survey of multiagent reinforcement learning},
  author={Busoniu, Lucian and Babuska, Robert and De Schutter, Bart},
  journal={IEEE Transactions on Systems, Man, and Cybernetics, Part C (Applications and Reviews)},
  volume={38},
  number={2},
  pages={156--172},
  year={2008},
  publisher={IEEE}
}

@article{oroojlooy2023review,
  title={A review of cooperative multi-agent deep reinforcement learning},
  author={Oroojlooy, Afshin and Hajinezhad, Davood},
  journal={Applied Intelligence},
  volume={53},
  number={11},
  pages={13677--13722},
  year={2023},
  publisher={Springer}
}

@inproceedings{littman1996generalized,
  title={A generalized reinforcement-learning model: Convergence and applications},
  author={Littman, Michael L and Szepesv{\'a}ri, Csaba},
  booktitle={ICML},
  volume={96},
  pages={310--318},
  year={1996}
}

@inproceedings{littman2001friend,
  title={Friend-or-foe Q-learning in general-sum games},
  author={Littman, Michael L and others},
  booktitle={ICML},
  volume={1},
  number={2001},
  pages={322--328},
  year={2001}
}

@article{fink1964equilibrium,
  title={Equilibrium in a stochastic $ n $-person game},
  author={Fink, Arlington M},
  journal={Journal of science of the hiroshima university, series ai (mathematics)},
  volume={28},
  number={1},
  pages={89--93},
  year={1964},
  publisher={Hiroshima University, Mathematics Program}
}

@article{daskalakis2013complexity,
  title={On the complexity of approximating a Nash equilibrium},
  author={Daskalakis, Constantinos},
  journal={ACM Transactions on Algorithms (TALG)},
  volume={9},
  number={3},
  pages={1--35},
  year={2013},
  publisher={ACM New York, NY, USA}
}

@article{hansen2013strategy,
  title={Strategy iteration is strongly polynomial for 2-player turn-based stochastic games with a constant discount factor},
  author={Hansen, Thomas Dueholm and Miltersen, Peter Bro and Zwick, Uri},
  journal={Journal of the ACM (JACM)},
  volume={60},
  number={1},
  pages={1--16},
  year={2013},
  publisher={ACM New York, NY, USA}
}

@article{mao2023provably,
  title={Provably efficient reinforcement learning in decentralized general-sum markov games},
  author={Mao, Weichao and Ba{\c{s}}ar, Tamer},
  journal={Dynamic Games and Applications},
  volume={13},
  number={1},
  pages={165--186},
  year={2023},
  publisher={Springer}
}

@article{song2021can,
  title={When can we learn general-sum Markov games with a large number of players sample-efficiently?},
  author={Song, Ziang and Mei, Song and Bai, Yu},
  journal={arXiv preprint arXiv:2110.04184},
  year={2021}
}

@article{feng2023improving,
  title={Improving sample efficiency of model-free algorithms for zero-sum Markov games},
  author={Feng, Songtao and Yin, Ming and Wang, Yu-Xiang and Yang, Jing and Liang, Yingbin},
  journal={arXiv preprint arXiv:2308.08858},
  year={2023}
}

@article{li2024provable,
  title={Provable memory efficient self-play algorithm for model-free reinforcement learning},
  author={Li, Na and Jiao, Yuchen and Shan, Hangguan and Yan, Shefeng},
  journal={arXiv preprint arXiv:2512.00351},
  year={2025}
}

@article{blackwell1956analog,
  title={An analog of the minimax theorem for vector payoffs},
  author={Blackwell, David},
journal={Pacific Journal of Mathematics}, 
volume={6},
number={1},
  year={1956}
}

@article{shapley1959equilibrium,
  title={Equilibrium points in games with vector payoffs},
  author={Shapley, Lloyd S and Rigby, Fred D},
  journal={Naval Research Logistics Quarterly},
  volume={6},
  number={1},
  pages={57--61},
  year={1959},
  publisher={Wiley Online Library}
}

@article{krieger2003pareto,
  title={On Pareto equilibria in vector-valued extensive form games},
  author={Krieger, Thomas},
  journal={Mathematical Methods of Operations Research},
  volume={58},
  number={3},
  pages={449--458},
  year={2003},
  publisher={Springer}
}

@article{voorneveld1999axiomatizations,
  title={Axiomatizations of Pareto equilibria in multicriteria games},
  author={Voorneveld, Mark and Vermeulen, Dries and Borm, Peter},
  journal={Games and economic behavior},
  volume={28},
  number={1},
  pages={146--154},
  year={1999},
  publisher={Elsevier}
}

@article{park2019decision,
  title={Decision Making and Games with Vector Outcomes},
  author={Park, Jaeok},
  journal={The BE Journal of Theoretical Economics},
  volume={20},
  number={1},
  pages={20180170},
  year={2019},
  publisher={De Gruyter}
}

@article{salonen1992axiomatic,
  title={An axiomatic analysis of the Nash equilibrium concept},
  author={Salonen, Hannu},
  journal={Theory and decision},
  volume={33},
  number={2},
  pages={177--189},
  year={1992},
  publisher={Springer}
}

@article{zapata2019maxmin,
  title={A maxmin approach for the equilibria of vector-valued games},
  author={Zapata, A and M{\'a}rmol, AM and Monroy, L and Caraballo, MA},
  journal={Group Decision and Negotiation},
  volume={28},
  number={2},
  pages={415--432},
  year={2019},
  publisher={Springer}
}

@article{yu1998study,
  title={The study of Pareto equilibria for multiobjective games by fixed point and Ky Fan minimax inequality methods},
  author={Yu, J and Yuan, GX-Z},
  journal={Computers \& Mathematics with Applications},
  volume={35},
  number={9},
  pages={17--24},
  year={1998},
  publisher={Elsevier}
}

@inproceedings{agarwal2022multi,
	title={Multi-objective reinforcement learning with non-linear scalarization},
	author={Agarwal, Mridul and Aggarwal, Vaneet and Lan, Tian},
	booktitle={Proceedings of the 21st International Conference on Autonomous Agents and Multiagent Systems},
	pages={9--17},
	year={2022}
}

@article{das1997closer,
	title={A closer look at drawbacks of minimizing weighted sums of objectives for Pareto set generation in multicriteria optimization problems},
	author={Das, Indraneel and Dennis, John E},
	journal={Structural optimization},
	volume={14},
	pages={63--69},
	year={1997},
	publisher={Springer}
}

@inproceedings{pirotta2015multi,
	title={Multi-objective reinforcement learning with continuous pareto frontier approximation},
	author={Pirotta, Matteo and Parisi, Simone and Restelli, Marcello},
	booktitle={Proceedings of the AAAI conference on artificial intelligence},
	volume={29},
	number={1},
	year={2015}
}

@inproceedings{barrett2008learning,
	title={Learning all optimal policies with multiple criteria},
	author={Barrett, Leon and Narayanan, Srini},
	booktitle={Proceedings of the 25th international conference on Machine learning},
	pages={41--47},
	year={2008}
}

@article{roijers2013survey,
	title={A survey of multi-objective sequential decision-making},
	author={Roijers, Diederik M and Vamplew, Peter and Whiteson, Shimon and Dazeley, Richard},
	journal={Journal of Artificial Intelligence Research},
	volume={48},
	pages={67--113},
	year={2013}
}

@article{li2020deep,
	title={Deep reinforcement learning for multiobjective optimization},
	author={Li, Kaiwen and Zhang, Tao and Wang, Rui},
	journal={IEEE transactions on cybernetics},
	volume={51},
	number={6},
	pages={3103--3114},
	year={2020},
	publisher={IEEE}
}

@article{klamroth2007constrained,
	title={Constrained optimization using multiple objective programming},
	author={Klamroth, Kathrin and J{\o}rgen, Tind},
	journal={Journal of Global Optimization},
	volume={37},
	pages={325--355},
	year={2007},
	publisher={Springer}
}

@book{ehrgott2005multicriteria,
	title={Multicriteria optimization},
	author={Ehrgott, Matthias},
	volume={491},
	year={2005},
	publisher={Springer Science \& Business Media}
}

@article{kasimbeyli2019comparison,
	title={Comparison of some scalarization methods in multiobjective optimization: comparison of scalarization methods},
	author={Kasimbeyli, Refail and Ozturk, Zehra Kamisli and Kasimbeyli, Nergiz and Yalcin, Gulcin Dinc and Erdem, Banu Icmen},
	journal={Bulletin of the Malaysian Mathematical Sciences Society},
	volume={42},
	pages={1875--1905},
	year={2019},
	publisher={Springer}
}

@inproceedings{
jiang2023multiobjective,
title={Multi-Objective Online Learning},
author={Jiyan Jiang and Wenpeng Zhang and Shiji Zhou and Lihong Gu and Xiaodong Zeng and Wenwu Zhu},
booktitle={The Eleventh International Conference on Learning Representations },
year={2023},
url={https://openreview.net/forum?id=dKkMnCWfVmm}
}

@article{gunantara2018review,
  title={A review of multi-objective optimization: Methods and its applications},
  author={Gunantara, Nyoman},
  journal={Cogent Engineering},
  volume={5},
  number={1},
  pages={1502242},
  year={2018},
  publisher={Taylor \& Francis}
}

@article{wu2021accommodating,
	title={Accommodating picky customers: Regret bound and exploration complexity for multi-objective reinforcement learning},
	author={Wu, Jingfeng and Braverman, Vladimir and Yang, Lin},
	journal={Advances in Neural Information Processing Systems},
	volume={34},
	pages={13112--13124},
	year={2021}
}

@book{caramia2008multi,
  title={Multi-objective management in freight logistics: Increasing capacity, service level and safety with optimization algorithms},
  author={Caramia, Massimiliano and Dell{\'{}} Olmo, Paolo},
  year={2008},
  publisher={Springer}
}

@article{yang2019generalized,
  title={A generalized algorithm for multi-objective reinforcement learning and policy adaptation},
  author={Yang, Runzhe and Sun, Xingyuan and Narasimhan, Karthik},
  journal={Advances in neural information processing systems},
  volume={32},
  year={2019}
}

@article{xiao2024direction,
	title={Direction-oriented multi-objective learning: Simple and provable stochastic algorithms},
	author={Xiao, Peiyao and Ban, Hao and Ji, Kaiyi},
	journal={Advances in Neural Information Processing Systems},
	volume={36},
	year={2024}
}

@inproceedings{bowman1976relationship,
	title={On the relationship of the Tchebycheff norm and the efficient frontier of multiple-criteria objectives},
	author={Bowman Jr, V Joseph},
	booktitle={Multiple Criteria Decision Making: Proceedings of a Conference Jouy-en-Josas, France May 21--23, 1975},
	pages={76--86},
	year={1976},
	organization={Springer}
}

@article{choo1983proper,
	title={Proper efficiency in nonconvex multicriteria programming},
	author={Choo, Eng Ung and Atkins, Derek R},
	journal={Mathematics of Operations Research},
	volume={8},
	number={3},
	pages={467--470},
	year={1983},
	publisher={INFORMS}
}

@book{miettinen1999nonlinear,
	title={Nonlinear multiobjective optimization},
	author={Miettinen, Kaisa},
	volume={12},
	year={1999},
	publisher={Springer Science \& Business Media}
}

@article{lin2024smooth,
	title={Smooth Tchebycheff Scalarization for Multi-Objective Optimization},
	author={Lin, Xi and Zhang, Xiaoyuan and Yang, Zhiyuan and Liu, Fei and Wang, Zhenkun and Zhang, Qingfu},
	journal={arXiv preprint arXiv:2402.19078},
	year={2024}
}

@article{hu2024revisiting,
	title={Revisiting scalarization in multi-task learning: A theoretical perspective},
	author={Hu, Yuzheng and Xian, Ruicheng and Wu, Qilong and Fan, Qiuling and Yin, Lang and Zhao, Han},
	journal={Advances in Neural Information Processing Systems},
	volume={36},
	year={2024}
}

@article{chen2024three,
	title={Three-way trade-off in multi-objective learning: Optimization, generalization and conflict-avoidance},
	author={Chen, Lisha and Fernando, Heshan and Ying, Yiming and Chen, Tianyi},
	journal={Advances in Neural Information Processing Systems},
	volume={36},
	year={2024}
}

@inproceedings{drugan2013designing,
  title={Designing multi-objective multi-armed bandits algorithms: A study},
  author={Drugan, Madalina M and Nowe, Ann},
  booktitle={The 2013 international joint conference on neural networks (IJCNN)},
  pages={1--8},
  year={2013},
  organization={IEEE}
}

@inproceedings{chen2023preference,
  title={Preference-controlled multi-objective reinforcement learning for conditional text generation},
  author={Chen, Wenqing and Tian, Jidong and Fan, Caoyun and Li, Yitian and He, Hao and Jin, Yaohui},
  booktitle={Proceedings of the AAAI Conference on Artificial Intelligence},
  volume={37},
  pages={12662--12672},
  year={2023}
}

@article{sener2018multi,
  title={Multi-task learning as multi-objective optimization},
  author={Sener, Ozan and Koltun, Vladlen},
  journal={Advances in neural information processing systems},
  volume={31},
  year={2018}
}

@inproceedings{mahapatra2020multi,
  title={Multi-task learning with user preferences: Gradient descent with controlled ascent in {Pareto} optimization},
  author={Mahapatra, Debabrata and Rajan, Vaibhav},
  booktitle={International Conference on Machine Learning},
  pages={6597--6607},
  year={2020},
  organization={PMLR}
}

@article{liu2021conflict,
  title={Conflict-averse gradient descent for multi-task learning},
  author={Liu, Bo and Liu, Xingchao and Jin, Xiaojie and Stone, Peter and Liu, Qiang},
  journal={Advances in Neural Information Processing Systems},
  volume={34},
  pages={18878--18890},
  year={2021}
}

@inproceedings{mahapatra2023multi,
  title={Multi-label learning to rank through multi-objective optimization},
  author={Mahapatra, Debabrata and Dong, Chaosheng and Chen, Yetian and Momma, Michinari},
  booktitle={Proceedings of the 29th ACM SIGKDD Conference on Knowledge Discovery and Data Mining},
  pages={4605--4616},
  year={2023}
}

@article{liu2021profiling,
  title={Profiling {Pareto} front with multi-objective stein variational gradient descent},
  author={Liu, Xingchao and Tong, Xin and Liu, Qiang},
  journal={Advances in Neural Information Processing Systems},
  volume={34},
  pages={14721--14733},
  year={2021}
}

@inproceedings{yu2021provably,
  title={Provably efficient algorithms for multi-objective competitive {RL}},
  author={Yu, Tiancheng and Tian, Yi and Zhang, Jingzhao and Sra, Suvrit},
  booktitle={International Conference on Machine Learning},
  pages={12167--12176},
  year={2021},
  organization={PMLR}
}

@article{zhou2022anchor,
	title={Anchor-changing regularized natural policy gradient for multi-objective reinforcement learning},
	author={Zhou, Ruida and Liu, Tao and Kalathil, Dileep and Kumar, PR and Tian, Chao},
	journal={Advances in Neural Information Processing Systems},
	volume={35},
	pages={13584--13596},
	year={2022}
}

@article{tekin2018multi,
  title={Multi-objective contextual multi-armed bandit with a dominant objective},
  author={Tekin, Cem and Tur{\u{g}}ay, Eralp},
  journal={IEEE Transactions on Signal Processing},
  volume={66},
  number={14},
  pages={3799--3813},
  year={2018},
  publisher={IEEE}
}

@inproceedings{busa2017multi,
  title={Multi-objective bandits: Optimizing the generalized Gini index},
  author={Busa-Fekete, R{\'o}bert and Sz{\"o}r{\'e}nyi, Bal{\'a}zs and Weng, Paul and Mannor, Shie},
  booktitle={International Conference on Machine Learning},
  pages={625--634},
  year={2017},
  organization={PMLR}
}

@inproceedings{yahyaa2014annealing,
  title={Annealing-pareto multi-objective multi-armed bandit algorithm},
  author={Yahyaa, Saba Q and Drugan, Madalina M and Manderick, Bernard},
  booktitle={2014 IEEE Symposium on Adaptive Dynamic Programming and Reinforcement Learning (ADPRL)},
  pages={1--8},
  year={2014},
  organization={IEEE}
}

@inproceedings{yahyaa2014scalarized,
  title={The scalarized multi-objective multi-armed bandit problem: An empirical study of its exploration vs. exploitation tradeoff},
  author={Yahyaa, Saba Q and Drugan, Madalina M and Manderick, Bernard},
  booktitle={2014 International Joint Conference on Neural Networks (IJCNN)},
  pages={2290--2297},
  year={2014},
  organization={IEEE}
}

@inproceedings{riquelme2015performance,
  title={Performance metrics in multi-objective optimization},
  author={Riquelme, Nery and Von L{\"u}cken, Christian and Baran, Benjamin},
  booktitle={2015 Latin American computing conference (CLEI)},
  pages={1--11},
  year={2015},
  organization={IEEE}
}

@article{giagkiozis2015methods,
  title={Methods for multi-objective optimization: An analysis},
  author={Giagkiozis, Ioannis and Fleming, Peter J},
  journal={Information Sciences},
  volume={293},
  pages={338--350},
  year={2015},
  publisher={Elsevier}
}

@inproceedings{fernando2022mitigating,
  title={Mitigating gradient bias in multi-objective learning: A provably convergent approach},
  author={Fernando, Heshan Devaka and Shen, Han and Liu, Miao and Chaudhury, Subhajit and Murugesan, Keerthiram and Chen, Tianyi},
  booktitle={The Eleventh International Conference on Learning Representations},
  year={2022}
}

@inproceedings{natarajan2005dynamic,
	title={Dynamic preferences in multi-criteria reinforcement learning},
	author={Natarajan, Sriraam and Tadepalli, Prasad},
	booktitle={Proceedings of the 22nd international conference on Machine learning},
	pages={601--608},
	year={2005}
}

@article{wang2013hypervolume,
	title={Hypervolume indicator and dominance reward based multi-objective Monte-Carlo Tree Search},
	author={Wang, Weijia and Sebag, Michele},
	journal={Machine learning},
	volume={92},
	pages={403--429},
	year={2013},
	publisher={Springer}
}

@inproceedings{van2013hypervolume,
	title={Hypervolume-based multi-objective reinforcement learning},
	author={Van Moffaert, Kristof and Drugan, Madalina M and Now{\'e}, Ann},
	booktitle={Evolutionary Multi-Criterion Optimization: 7th International Conference, EMO 2013, Sheffield, UK, March 19-22, 2013. Proceedings 7},
	pages={352--366},
	year={2013},
	organization={Springer}
}

@article{van2014multi,
  title={Multi-objective reinforcement learning using sets of pareto dominating policies},
  author={Van Moffaert, Kristof and Now{\'e}, Ann},
  journal={The Journal of Machine Learning Research},
  volume={15},
  number={1},
  pages={3483--3512},
  year={2014},
  publisher={JMLR. org}
}

@article{geoffrion1968proper,
	title={Proper efficiency and the theory of vector maximization},
	author={Geoffrion, Arthur M},
	journal={Journal of mathematical analysis and applications},
	volume={22},
	number={3},
	pages={618--630},
	year={1968},
	publisher={Elsevier}
}

@inproceedings{xu2020prediction,
  title={Prediction-guided multi-objective reinforcement learning for continuous robot control},
  author={Xu, Jie and Tian, Yunsheng and Ma, Pingchuan and Rus, Daniela and Sueda, Shinjiro and Matusik, Wojciech},
  booktitle={International conference on machine learning},
  pages={10607--10616},
  year={2020},
  organization={PMLR}
}

@inproceedings{chen2019meta,
  title={Meta-learning for multi-objective reinforcement learning},
  author={Chen, Xi and Ghadirzadeh, Ali and Bj{\"o}rkman, M{\aa}rten and Jensfelt, Patric},
  booktitle={2019 IEEE/RSJ International Conference on Intelligent Robots and Systems (IROS)},
  pages={977--983},
  year={2019},
  organization={IEEE}
}

@article{hayes2022practical,
	title={A practical guide to multi-objective reinforcement learning and planning},
	author={Hayes, Conor F and R{\u{a}}dulescu, Roxana and Bargiacchi, Eugenio and K{\"a}llstr{\"o}m, Johan and Macfarlane, Matthew and Reymond, Mathieu and Verstraeten, Timothy and Zintgraf, Luisa M and Dazeley, Richard and Heintz, Fredrik and others},
	journal={Autonomous Agents and Multi-Agent Systems},
	volume={36},
	number={1},
	pages={26},
	year={2022},
	publisher={Springer}
}

@article{zhu2023scaling,
  title={Scaling pareto-efficient decision making via offline multi-objective rl},
  author={Zhu, Baiting and Dang, Meihua and Grover, Aditya},
  journal={arXiv preprint arXiv:2305.00567},
  year={2023}
}

@article{EvansSteuer1973,
author = {Evans, J. P. and Steuer, R. E.},
title = {A revised simplex method for linear multiple objective programs},
year = {1973},
issue_date = {December  1973},
publisher = {Springer-Verlag},
address = {Berlin, Heidelberg},
volume = {5},
number = {1},
issn = {0025-5610},
url = {https://doi.org/10.1007/BF01580111},
doi = {10.1007/BF01580111},
journal = {Math. Program.},
month = dec,
pages = {54–72},
numpages = {19}
}

@ARTICLE{MifraniNoll2025,
       author = {{Mifrani}, Anas and {Noll}, Dominikus},
        title = "{Linear programming for finite-horizon vector-valued Markov decision processes}",
      journal = {arXiv e-prints},
         year = 2025,
        month = feb,
          eid = {arXiv:2502.13697},
        pages = {arXiv:2502.13697},
          doi = {10.48550/arXiv.2502.13697},
archivePrefix = {arXiv},
       eprint = {2502.13697},
 primaryClass = {math.OC},
       adsurl = {https://ui.adsabs.harvard.edu/abs/2025arXiv250213697M}
}

@article{benson1978existence,
  title={Existence of efficient solutions for vector maximization problems},
  author={Benson, Harold P.},
  journal={Journal of Optimization Theory and Applications},
  volume={26},
  number={4},
  pages={569--580},
  year={1978},
  doi={10.1007/BF00933152}
}
\bibliographystyle{abbrv} 

\newpage
\appendix
\onecolumn
\section{Related Works}\label{sec:related works}
\textbf{Markov Games.} Markov games (MGs), or stochastic games, introduced by \cite{shapley1953stochastic,fink1964equilibrium}, form the standard foundation for multi-agent reinforcement learning (MARL), particularly in equilibrium learning. Comprehensive surveys such as \cite{busoniu2008comprehensive, oroojlooy2023review, zhang2021multi} offer thorough coverage of the field’s evolution. Early work in MARL focused on asymptotic convergence guarantees \cite{littman2001friend, littman1996generalized}, whereas recent research emphasizes finite-sample analyses to establish non-asymptotic guarantees, especially for learning Nash equilibria (NE)—a central solution concept. The existence of NE in general-sum MGs was shown by \cite{fink1964equilibrium}, and the algorithmic foundation was laid by the seminal work of \cite{littman1994markov}. Classical algorithms such as Nash-Q \cite{hu2003nash} and FF-Q \cite{littman2001friend} were proposed to compute NE and its variants. However, computing NE in general-sum multi-player settings remains  computationally difficult: even one-stage games already exhibit PPAD-hard NE computation problems for the relevant approximation regimes \cite{daskalakis2013complexity}. 
 In the two-player zero-sum settings, backward induction solves a finite matrix game at each state and stage, and each matrix game can be solved by linear programming. For discounted turn-based two-player zeor-sum games, \cite{hansen2013strategy} additionally establish a strongly polynomial  strategy-iteration
bound. To address the computational intractability in general-sum MGs, attention has shifted to weaker notions like CE and CCE, with polynomial-time algorithms such as V-learning \cite{jin2021v, mao2023provably, song2021can} and Nash value iteration \cite{liu2021sharp} enabling efficient computation. Furthermore, significant progress in finite-sample analysis—spanning both model-based and model-free algorithms—has been achieved in the two-player zero-sum setting, as evidenced by \cite{bai2020provable, xie2020learning, cui2023breaking, chen2022decentralized, liu2021sharp, feng2023improving, li2024provable}, advancing the theoretical understanding of equilibrium learning in standard MARL without robustness considerations.

\textbf{Multi-Objective Reinforcement Learning.} Multi-objective reinforcement learning (MORL) is based on multi-objective optimization, e.g., \cite{choo1983proper,geoffrion1968proper,ehrgott2005multicriteria,bowman1976relationship,miettinen1999nonlinear,caramia2008multi,gunantara2018review, deb2016multi, giagkiozis2015methods, riquelme2015performance,das1997closer,liu2021profiling, liu2021conflict, chen2023preference,mahapatra2023multi, sener2018multi,klamroth2007constrained,kasimbeyli2019comparison,fernando2022mitigating,hu2024revisiting,chen2024three,mahapatra2020multi,xiao2024direction,lin2024smooth}. This literature explores various scalarization methods for finding Pareto optimal solutions.

Multi-objective optimization is then later extended  to the online learning setting, including online convex optimization and bandit problems~\cite{drugan2013designing, yahyaa2014scalarized, turgay2018multi, lu2019multi, tekin2018multi, busa2017multi, yahyaa2014annealing, jiang2023multiobjective}. Extensive studies on MORL are also developed \cite{roijers2013survey,puterman2014markov,agarwal2022multi,van2013hypervolume,natarajan2005dynamic,wang2013hypervolume,barrett2008learning,pirotta2015multi,van2013scalarized,xu2020prediction,hayes2022practical,van2014multi,chen2019meta,yang2019generalized, wiering2014model,zhu2023scaling,wu2021offline, yu2021provably, wu2021accommodating,zhou2022anchor,li2020deep,qiu2024traversing,li2024find}, which have studied different scalarization methods. However, we consider multi-agent RL, and results of these studies cannot be directly extended.

\textbf{Normal-Form Games with Vectorial Payoffs.}
Game theory has developed extensive studies on normal-form games with vector or multi-criteria payoffs, where the notion of PNE is proposed and studied \cite{blackwell1956analog,
shapley1959equilibrium,
krieger2003pareto,
voorneveld1999axiomatizations,
park2019decision,
salonen1992axiomatic,
zapata2019maxmin,
yu1998study,lozovanu2005multiobjective,sridhar2012pareto}. However, all of these works are focusing on normal-form games, which do not extend to Markov games with sequential decisions.

\section{Details of the Experiments in Section~\ref{sec:exp-envs}}\label{sec:exp-envs-details}

\subsection{Environments}\label{sec:exp-envs-env}
All environments have two agents and two objectives, and all rewards are summed over an episode without discounting. In each environment, one region or event pays off only for some preference profiles and is rarely reached without directed exploration: under uniformly random actions, it is reached in $0\%$ (MO-Resource Race), $0.06\%$ (MO-Predator Prey), $0.02\%$ (MO-Traffic Control), and $0\%$ (MO-Negotiation) of $5{,}000$ episodes. In each environment, a single agent can reach it without help from the other agent.

\textbf{MO-Predator Prey.} Two predators share the rewards of their captures but pay individually for their movement, so they must decide who spends effort hunting, which creates a free-rider tension. A small prey is available next to the start, and a valuable one far away. The grid has $12\times12$ cells, and episodes last $40$ steps. The predators start at cells $(0,0)$ and $(0,1)$.
\begin{itemize}
\item \textit{Hare burrow.} The burrow is at cell $(4,1)$. A predator on the burrow catches a hare if one is there, and a new hare appears two steps after each catch. Each catch gives both predators $0.5$ on Objective~1.
\item \textit{Stag.} The stag starts at a uniformly random cell of the den $\{9,10,11\}^2$. After the predators move, it takes a uniformly random step (four moves or staying) inside the den. A capture occurs when a predator occupies the stag's cell before or after the stag's step. Both predators then receive $5$ on Objective~1, and a new stag appears at a random cell of the den.
\item \textit{Energy.} Each predator receives $-0.5$ on Objective~2 for every step in which it moves.
\item \textit{Observation.} A predator observes its own cell and, when it is in or next to the den, the stag's cell.
\end{itemize}
The burrow thus offers a small immediate reward, while the stag is worth ten times as much.

\textbf{MO-Traffic Control.} As in the example of Section~1, two agents control the signals of adjacent intersections on a corridor. Traffic served at one intersection enters the other, which couples their decisions. Episodes last $50$ steps. Each intersection has a corridor queue and a side-street queue, initialized uniformly in $\{0,\dots,3\}$; at every step, each queue receives one vehicle with probability $1/2$. An agent gives green either to the corridor or to the side street, and a green approach serves up to $2$ queued vehicles. An intersection that has given green to the corridor for $20$ consecutive steps is in the corridor's green wave and serves $4$ additional corridor vehicles per step while it keeps the corridor green. Half of the corridor vehicles served at one intersection (rounded down) join the corridor queue of the other intersection; queues are capped at $15$. Objective~1 is the number of vehicles served divided by $4$, and Objective~2 is $-0.02\,|q^{\mathrm{corr}}-q^{\mathrm{side}}|$ per step. An agent observes which of its two queues is longer and how many consecutive steps it has given green to the corridor.

\textbf{MO-Resource Race.} Two agents compete for red and blue resources. Depending on the preference profile, they specialize (one on red, the other on blue) or compete for the same color. The grid has $12\times12$ cells, and episodes last $40$ steps. The agents start at cells $(0,0)$ and $(0,1)$. One red source is at $(2,2)$ and three blue sources are at $(11,10)$, $(11,11)$, and $(10,11)$. An agent on a ready source harvests it: it receives $2$ on the source's objective (Objective~1 for red, Objective~2 for blue), and the other agent receives $-0.3$ on the same objective; if both agents harvest the same source together, each receives $1$. A harvested source is ready again two steps later. An agent observes its own cell.

\textbf{MO-Negotiation.} Two agents repeatedly divide a resource. Demanding little builds a reputation that benefits both agents but costs its builder, which creates a free-rider tension. Episodes last $30$ rounds. In each round, the agents simultaneously demand low, medium, or high ($d\in\{0,1,2\}$). If both demand high, each receives $0.3(d+1)$ on Objective~1 and $-1$ on Objective~2. Otherwise, each receives $d+1$ on Objective~1 and $2-0.5\,|d_1-d_2|$ on Objective~2. An agent that has demanded low in $12$ consecutive rounds has a reputation, which it keeps as long as it demands low. In every round in which an agent with a reputation demands low, both agents receive $2.5$ on Objective~1 and $4$ on Objective~2 instead. An agent observes its own reputation level, that is, its number of consecutive low demands (capped at $12$).

\subsection{Learner and Methods}\label{sec:exp-envs-learner}
All methods use the same tabular learner and the same hyperparameters, and differ only in how they explore. Each agent keeps a table of action values for each of five preference bins; bin $b$ covers the weight interval $[b/5,(b+1)/5)$ on Objective~1 and uses the weight vector $(\tfrac{2b+1}{10},1-\tfrac{2b+1}{10})$. At the end of every episode, each transition updates the entry of every bin with a one-step Q-learning target, processed backward through the episode. Every transition thus serves every preference. The learning rate is $\max\{0.01,\,0.3(1-e/8000)\}$ in episode $e$. In MO-Predator Prey and MO-Resource Race, the state is the observation and the discount factor is $0.95$. In MO-Traffic Control and MO-Negotiation, the state also includes the time step and returns are undiscounted (finite horizon); there, the observation is small, so including the time step keeps the table small. Training lasts $8{,}000$ episodes. In each episode, each agent draws its weight vector from $\mathrm{Dirichlet}(1.5\cdot\mathbf 1)$ and acts with the bin that contains it.

\textbf{Two-Phase + EQ (Ours).} Phase~1 consists of the first $2{,}000$ episodes. In Phase~1, each agent ignores the rewards. It acts greedily (with probability $0.05$, uniformly at random) with respect to an exploration table that is trained with the count-based Error Quantification (EQ) term $\beta/\sqrt{N(s,a)}$, $\beta=1$, which plays the role of the confidence radii $\Psi,\Phi$ of \Cref{alg:combined}. This table uses the observation without the time step, a discount factor of $0.95$, a learning rate of $0.3$, and the optimistic initial value $\beta/(1-0.95)=20$. The bonus enters the target of this table, so it propagates to the states that lead to rarely visited states. The counts $N(s,a)$ do not depend on preferences. At the end of Phase~1, the stored Phase~1 episodes are replayed three times (learning rate $0.3$). In Phase~2, the agents act $\varepsilon$-greedily with respect to the preference-conditioned values, with $\varepsilon$ decreasing linearly from $0.3$ to $0.03$. The bonus never enters the preference-conditioned values, which define the evaluated policy.

\textbf{Two-Phase (no EQ).} As above, except that the agents take uniformly random actions in Phase~1.

\textbf{Independent MORL.} No Phase~1. The agents act $\varepsilon$-greedily with respect to the preference-conditioned values, with $\varepsilon$ decreasing linearly from $0.8$ to $0.03$ over training.

\textbf{Independent MORL + EQ.} As Independent MORL, but the agents act $\varepsilon$-greedily with respect to an optimistic table. For each bin, this table is trained with the scalarized reward plus the EQ term $\beta/\sqrt{N(s,a)}$. It is initialized at $(r_{\max}+\beta)/(1-0.95)$, where $r_{\max}$ is an upper bound on the one-step reward ($6$, $1.5$, $2$, and $4$ in the four environments). We evaluate two policies: the greedy policy of the optimistic table, and the greedy policy of the preference-conditioned values learned from the same transitions. We report the one with the lower Nash gap in each environment, that is, the better one for the baseline (the preference-conditioned values in MO-Predator Prey and MO-Resource Race, and the optimistic table in the other two).

\textbf{Fixed $\lambda$ (equal).} A single table for the weight vector $(0.5,0.5)$, trained with $\varepsilon$ decreasing from $0.8$ to $0.03$ and used for every test profile.

\subsection{Evaluation}\label{sec:exp-envs-eval}
Each method is trained with $10$ seeds. Each trained run is evaluated on $60$ test profiles; each weight vector of a profile is drawn from $\mathrm{Dirichlet}(2\cdot\mathbf 1)$. For a test profile $\bm\Lambda$, each agent follows its greedy policy for its own weight vector, which gives the joint policy $\bm\pi^{\bm\Lambda}$. AUT averages $\bm\lambda_i^\top V_i^{\bm\pi^{\bm\Lambda}}$ over the two agents and the test profiles, with each value estimated from $20$ episodes.

\textbf{Nash gap.} For each player $i$, we search for a profitable deviation against the fixed policy of the other player. The candidates are as follows:
\begin{itemize}
\item A best response trained for $600$ episodes with the two-phase learner above: $150$ bonus-driven exploration episodes, replay, and $450$ $\varepsilon$-greedy episodes. It is evaluated over $50$ episodes.
\item The scripted strategies of the environment, each evaluated over $10$ episodes:
  \begin{itemize}
  \item MO-Predator Prey: walk to the den and chase the stag; walk to the burrow and stay; or never move.
  \item MO-Traffic Control: always give green to the corridor; or always serve the longer queue.
  \item MO-Resource Race: walk to the red source and stay; walk to the blue sources and alternate between two of them; or stay.
  \item MO-Negotiation: always demand low, always medium, or always high.
  \end{itemize}
\end{itemize}
The gain of player $i$ is $\bm\lambda_i^\top V_i$ under its best candidate minus $\bm\lambda_i^\top V_i^{\bm\pi^{\bm\Lambda}}$. The Nash gap of the profile is the larger of the two players' gains, clipped at zero. It is a lower bound on $\mathrm{NashGap}_{\mathcal G_{\bm\Lambda}}(\bm\pi^{\bm\Lambda})$, because the search may miss better deviations. \Cref{tab:strategic} reports its mean over test profiles. \Cref{tab:exp-envs-extra} also reports its maximum over the $60$ test profiles (worst case).

\textbf{Effect of a deviation on the other agent.} BR-AUT is the value of the other agent under the deviation that is best for the deviator, averaged as AUT, and $\Delta=\mathrm{AUT}-\text{BR-AUT}$. These two quantities do not measure the distance to an equilibrium. They describe how a unilateral deviation affects the agent that does not deviate. If $\Delta<0$, the deviation helps it. For example, in MO-Predator Prey a deviating player that starts hunting gives the other player captures for free, and in MO-Resource Race a deviating player that leaves for the blue sources leaves the red source to the other player. Both effects are large for the baselines, whose policies forgo these options.

\subsection{Additional Results}\label{sec:exp-envs-results}
\Cref{tab:exp-envs-extra} reports the worst-case Nash gap over the $60$ test profiles, BR-AUT, and $\Delta$.

\textbf{Exploration.} In the four environments, our method reaches the den, the green wave, the blue sources, or a partnership round in $40\%$, $42\%$, $30\%$, and $48\%$ of its training episodes, and Independent MORL + EQ in $37\%$, $2\%$, $14\%$, and $0.4\%$. The other three baselines reach them in fewer than $0.05\%$ of the episodes.

\textbf{Worst case.} Our method also has the smallest worst-case Nash gap in all four environments. In MO-Predator Prey, the worst cases arise when the two independently learning predators settle on a convention in which only one of them hunts, regardless of their preferences; Phase~2 of \Cref{alg:combined} instead computes an equilibrium for each queried profile.

\textbf{Effect of a deviation on the other agent.} For the baselines, $\Delta$ is strongly negative in MO-Predator Prey and MO-Negotiation. Their policies forgo the hard-to-reach option, so a deviating player that takes it also benefits the other player: captures are shared, and a reputation turns rounds into partnerships. For our method, $|\Delta|$ is much smaller in these two environments, and in the other two it is about $12\%$ of its AUT or less.

\begin{table*}[t]
    \centering
    \small
    \caption{Further results for \Cref{tab:strategic} (mean and 95\% confidence interval over $10$ seeds). Worst-case Nash gap: maximum over the $60$ test profiles. BR-AUT and $\Delta=\mathrm{AUT}-\text{BR-AUT}$: utility of the agent that does not deviate.}
    \label{tab:exp-envs-extra}
\begin{tabular}{llccc}
\toprule
Environment & Method & Worst-case Nash gap & BR-AUT & $\Delta$ \\
\midrule
MO-Predator Prey   & \textbf{Two-Phase + EQ (Ours)} & $43.4\pm2.9$ & $26.02\pm1.24$ & $-6.80\pm0.99$ \\
                   & Two-Phase (no EQ) & $52.6\pm1.6$ & $25.21\pm1.15$ & $-21.80\pm0.97$ \\
                   & Fixed $\lambda$ (equal) & $50.2\pm1.3$ & $25.60\pm1.43$ & $-21.25\pm1.31$ \\
                   & Indep.\ MORL & $50.4\pm2.9$ & $25.37\pm0.95$ & $-22.13\pm0.78$ \\
                   & Indep.\ MORL + EQ & $53.0\pm2.7$ & $23.80\pm0.72$ & $-21.71\pm1.15$ \\
\addlinespace
MO-Traffic Control & \textbf{Two-Phase + EQ (Ours)} & $13.7\pm1.3$ & $13.15\pm1.25$ & $1.80\pm0.52$ \\
                   & Two-Phase (no EQ) & $26.7\pm0.8$ & $7.48\pm0.28$ & $0.24\pm0.10$ \\
                   & Fixed $\lambda$ (equal) & $25.9\pm0.4$ & $7.30\pm0.24$ & $0.67\pm0.09$ \\
                   & Indep.\ MORL & $26.0\pm0.5$ & $7.55\pm0.19$ & $0.54\pm0.10$ \\
                   & Indep.\ MORL + EQ & $26.7\pm1.0$ & $7.84\pm1.17$ & $0.00\pm0.99$ \\
\addlinespace
MO-Resource Race   & \textbf{Two-Phase + EQ (Ours)} & $14.0\pm0.5$ & $20.81\pm0.89$ & $-1.77\pm0.40$ \\
                   & Two-Phase (no EQ) & $38.2\pm1.1$ & $13.08\pm1.78$ & $-3.86\pm1.35$ \\
                   & Fixed $\lambda$ (equal) & $37.9\pm1.0$ & $11.38\pm2.67$ & $-2.38\pm2.09$ \\
                   & Indep.\ MORL & $37.7\pm1.2$ & $13.11\pm1.27$ & $-3.74\pm1.35$ \\
                   & Indep.\ MORL + EQ & $30.5\pm1.2$ & $14.56\pm1.51$ & $-1.34\pm1.00$ \\
\addlinespace
MO-Negotiation     & \textbf{Two-Phase + EQ (Ours)} & $16.2\pm0.4$ & $74.75\pm0.67$ & $3.18\pm0.27$ \\
                   & Two-Phase (no EQ) & $31.7\pm1.0$ & $74.47\pm0.52$ & $-12.67\pm0.53$ \\
                   & Fixed $\lambda$ (equal) & $38.2\pm1.1$ & $75.42\pm0.99$ & $-15.58\pm1.27$ \\
                   & Indep.\ MORL & $32.4\pm0.7$ & $73.59\pm0.68$ & $-11.65\pm0.82$ \\
                   & Indep.\ MORL + EQ & $37.3\pm3.4$ & $73.61\pm1.06$ & $-12.43\pm1.58$ \\
\bottomrule
\end{tabular}
\end{table*}

\textbf{Control: easy variants.} To check that the gains come from exploration, we repeat the experiment ($5$ seeds) with the region or event made easy to reach:
\begin{itemize}
\item the blue sources at $(4,3)$, $(4,4)$, and $(3,4)$;
\item the stag's den at the cells $\{0,1,2\}\times\{3,4,5\}$;
\item the green wave after $3$ steps;
\item the reputation after $2$ rounds.
\end{itemize}
Uniformly random actions then reach it in $52$ to $100\%$ of episodes, and every method reaches it in $74$ to $100\%$ of its training episodes. \Cref{tab:exp-envs-control} shows that our method then performs like the ablation and Independent MORL (all Mann--Whitney tests on Nash gap and AUT give $p\ge0.09$, except that our Nash gap is lower than that of Independent MORL in MO-Traffic Control). Once these options are easy to reach, all methods find them, and the advantage of preference-free exploration vanishes, which confirms that the gains in \Cref{tab:strategic} come from exploration.

\begin{table*}[t]
    \centering
    \small
    \caption{Control experiment: the environments of \Cref{tab:strategic} with the region or event made easy to reach (mean and 95\% confidence interval over $5$ seeds).}
    \label{tab:exp-envs-control}
\begin{tabular}{llcc}
\toprule
Environment (easy variant) & Method & Nash gap ($\downarrow$) & AUT ($\uparrow$) \\
\midrule
MO-Predator Prey   & \textbf{Two-Phase + EQ (Ours)} & $6.78\pm1.73$ & $45.79\pm1.59$ \\
                   & Two-Phase (no EQ) & $6.71\pm1.48$ & $46.12\pm2.53$ \\
                   & Fixed $\lambda$ (equal) & $1.66\pm0.82$ & $54.92\pm2.76$ \\
                   & Indep.\ MORL & $4.11\pm2.02$ & $49.58\pm3.43$ \\
                   & Indep.\ MORL + EQ & $3.62\pm0.44$ & $50.26\pm4.18$ \\
\addlinespace
MO-Traffic Control & \textbf{Two-Phase + EQ (Ours)} & $0.12\pm0.01$ & $34.58\pm1.60$ \\
                   & Two-Phase (no EQ) & $0.18\pm0.14$ & $35.23\pm0.52$ \\
                   & Fixed $\lambda$ (equal) & $6.97\pm0.38$ & $26.77\pm0.74$ \\
                   & Indep.\ MORL & $1.80\pm0.74$ & $34.11\pm1.28$ \\
                   & Indep.\ MORL + EQ & $0.14\pm0.03$ & $35.34\pm1.62$ \\
\addlinespace
MO-Resource Race   & \textbf{Two-Phase + EQ (Ours)} & $4.81\pm3.98$ & $26.78\pm1.56$ \\
                   & Two-Phase (no EQ) & $2.18\pm0.77$ & $28.55\pm0.92$ \\
                   & Fixed $\lambda$ (equal) & $13.96\pm5.12$ & $19.71\pm2.22$ \\
                   & Indep.\ MORL & $5.53\pm4.81$ & $26.16\pm2.90$ \\
                   & Indep.\ MORL + EQ & $9.86\pm5.55$ & $25.26\pm2.44$ \\
\addlinespace
MO-Negotiation     & \textbf{Two-Phase + EQ (Ours)} & $0.71\pm0.15$ & $93.59\pm1.09$ \\
                   & Two-Phase (no EQ) & $0.70\pm0.21$ & $94.04\pm0.89$ \\
                   & Fixed $\lambda$ (equal) & $0.52\pm0.35$ & $93.66\pm1.05$ \\
                   & Indep.\ MORL & $0.56\pm0.20$ & $94.11\pm0.75$ \\
                   & Indep.\ MORL + EQ & $1.30\pm0.29$ & $93.93\pm0.82$ \\
\bottomrule
\end{tabular}
\end{table*}
\section{Scalable Deep Multi-Agent Exploration via Neural Approximation}\label{sec:deep exp}

While Algorithm \ref{alg:combined} offers rigorous theoretical guarantees in tabular settings, its direct application to complex, continuous tasks is hindered by the curse of dimensionality. To extend our framework to high-dimensional state spaces, we propose Algorithm \ref{alg:deep}: Deep Multi-Objective Multi-Player Learning with Neural Uncertainty Quantifier (DML-NUQ). 
This practical realization approximates the theoretical principles of the tabular phase using deep reinforcement learning. Specifically, we introduce a scalable architecture that integrates three key components: (1) neural function approximation for efficient policy and value representation, (2) a distillation-based uncertainty quantifier to emulate count-based exploration bonuses in continuous domains, and (3) the Multi-Agent PPO (MAPPO) framework \cite{yu2022surprising} for robust policy optimization. The details of these components are elaborated below. 

We first initialize the conditioned policy $\pi_\theta(a|o, \boldsymbol\lambda)$ as a neural network parameterized by $\theta$, where $o$ denotes the agent's local observation. To accelerate training for homogeneous agents, we employ parameter sharing. Additionally, we initialize a centralized critic $V_\phi(s,\boldsymbol\lambda)$ parameterized by $\phi$, which outputs a scalar value estimate under preference $\boldsymbol\lambda$. While Algorithm \ref{alg:combined} relies on exact visitation counts $N(s, {a})$ to construct the uncertainty bonuses $\Phi$ and $\Psi$, such tabulation is intractable in high-dimensional continuous state spaces. To extend our theoretical framework to complex environments, we introduce a neural uncertainty quantifier that serves as a differentiable surrogate for the count-based bonus. Specifically, we employ a random network distillation paradigm to estimate the epistemic uncertainty of the agents. We construct two neural networks mapping the state space $\mathcal{S}$ to a latent embedding space $\mathbb{R}^k$: a fixed, randomly initialized target network $f: \mathcal{S} \to \mathbb{R}^k$, and a trainable predictor network $f_\psi: \mathcal{S} \to \mathbb{R}^k$. The predictor is optimized to minimize the distillation error against the fixed target. The prediction error serves as a heuristic state-novelty signal: $r_{\text{int}}(s) \triangleq || f_\psi(s) - f(s) ||^2$.

Since the predictor struggles to match the random target on rarely visited states, a high $r_{\text{int}}$ correlates with low empirical visitation density. We utilize this error as the intrinsic reward to approximate the exploration bonus in Algorithm \ref{alg:combined}, effectively driving the agents to explore the state manifold. For each episode, we sample a preference profile and use the conditioned policy to collect data from the environment and save them to the replay buffer. Once we have enough data, we adopt MAPPO \cite{yu2022surprising} framework to train the networks over these data. 

\begin{algorithm}[!htb]
\caption{Deep Multi-Objective Multi-Player Learning with Neural Uncertainty Quantifier (DML-NUQ)}
\label{alg:deep}
\begin{algorithmic}[1]
    \STATE \textbf{Input:} Learning rates $\eta_\theta, \eta_\phi, \eta_\psi$, exploration coefficient $\beta$, number of episodes $K$.
    \STATE \textbf{Initialize:} Conditioned policy network $\pi_\theta(a|o, \boldsymbol\lambda)$ and value network $ V_\phi(s, \boldsymbol\lambda)$, fixed target $f$ and trainable predictor $f_\psi$, replay buffer $\mathcal{B}$.
    
    \FOR{episode $k = 1, \dots, K$}
        \STATE Sample preference vector $\boldsymbol\lambda_k \sim \text{Dirichlet}({1}_M)$.
        \STATE Reset environment, observe state $s_0$ and observations $o_{i,0}$ for each agent $i$.
        
        \FOR{step $t = 0, \dots, T-1$}
            \STATE Each agent $i$ executes $a_i \sim \pi_\theta(\cdot | {o}_{i,t}, \boldsymbol\lambda_k)$. Get next state $s_{t+1}$, observation ${o}_{t+1}$, and vector reward $\mathbf{r}_t$.
            
            \STATE Compute intrinsic reward $r^{\text{int}}_t \gets || \hat{f}_\psi(s_{t+1}) - f(s_{t+1}) ||^2$
            
            \STATE Store tuple $({o}_t, s_t, \ba_t, \mathbf{r}_t, r^{\text{int}}_t, \boldsymbol\lambda_k)$ in $\mathcal{B}$.
        \ENDFOR
        
        \STATE Compute scalarized rewards: $R_t \gets \boldsymbol\lambda_k^\top \mathbf{r}_t + \beta \cdot r^{\text{int}}_t$.
        \STATE Compute advantages $\hat{A}_t$ and returns $\hat{G}_t$ using $V_\phi(s, \boldsymbol\lambda_k)$ and $R_t$.
        
        \STATE Update $\theta$ via maximizing the PPO-Clip objective.
        \STATE Update $\phi$ by minimizing $\mathcal{L}_v = ( V_\phi(s_t, \boldsymbol\lambda_k) - \hat{G}_t)^2$.
        \STATE Update predictor $\psi$ by minimizing $\mathcal{L}_{p} = || f_\psi(s_t) - f(s_t) ||^2$.
    \ENDFOR
    \STATE \textbf{Return:} Policy $\pi_{\theta}(\cdot|\cdot, \cdot)$
\end{algorithmic}
\end{algorithm}

\subsection{Experimental Results}

We benchmark our proposed Algorithm~\ref{alg:deep} against two baselines: MOMAPPO~\cite{felten2024momaland}, a multi-objective multi-agent policy gradient method, and PCN (Pareto Conditioned Networks)~\cite{reymond2022pareto}, a multi-objective RL algorithm that trains a conditioned network to approximate the full Pareto front by conditioning on preference vectors. PCN's implementation and hyperparameters follow the official Momaland \cite{felten2024momaland} codebase directly, ensuring a fair and reproducible comparison. We inherited the base training hyperparameters directly from the standard MOMAPPO framework, using identical experimental configurations for both methods, including the neural network architecture (number of layers, hidden units), learning rates, and the total number of training time steps. We list the common hyperparameters used in Table~\ref{tab:common_params}, and the environment-specific configurations along with details for the NUQ module in our proposed algorithm in Table~\ref{tab:specific_params}.

We evaluate two environments with distinct interaction structures. \textbf{MO-MultiWalker} (MO-MW) is a cooperative setting where agents share a common goal (jointly carrying a package) and receive a shared team reward. \textbf{MO-Item Gathering} (MO-IG) is a mixed/general-sum setting: agents receive individualized multi-objective reward vectors based on their own item collection, and items are removed upon pickup, creating implicit spatial competition over finite shared resources. This is the general-sum, agent-specific preference setting that our theory targets. In both environments, we report the team utility (sum of all agents' rewards) as the primary performance measure for comparability across methods.

\textbf{(1) Multi-Objective Item Gathering (Discrete Environment):} Adapted from \cite{kallstrom2019tunable}, this environment is a multi-agent resource gathering problem. The agents inhabit a grid world populated with items of various colors and each color corresponds to a distinct objective dimension. The agents are tasked with navigating the grid to collect items, aiming to maximize the cumulative collection count. 

\textbf{(2) Multi-Objective MultiWalker (Continuous Environment):} This is a multi-objective extension of the standard MultiWalker environment from the PettingZoo library, originally introduced by \cite{gupta2017cooperative}. In this task, a team of bipedal walkers must transport a large package across a terrain. The agents aim to optimize two objectives simultaneously: maximizing the distance traveled by the package, and minimizing the fluctuation and tilt of the package to keep it steady. 

\begin{figure}
    \centering
    \includegraphics[width=0.35\linewidth]{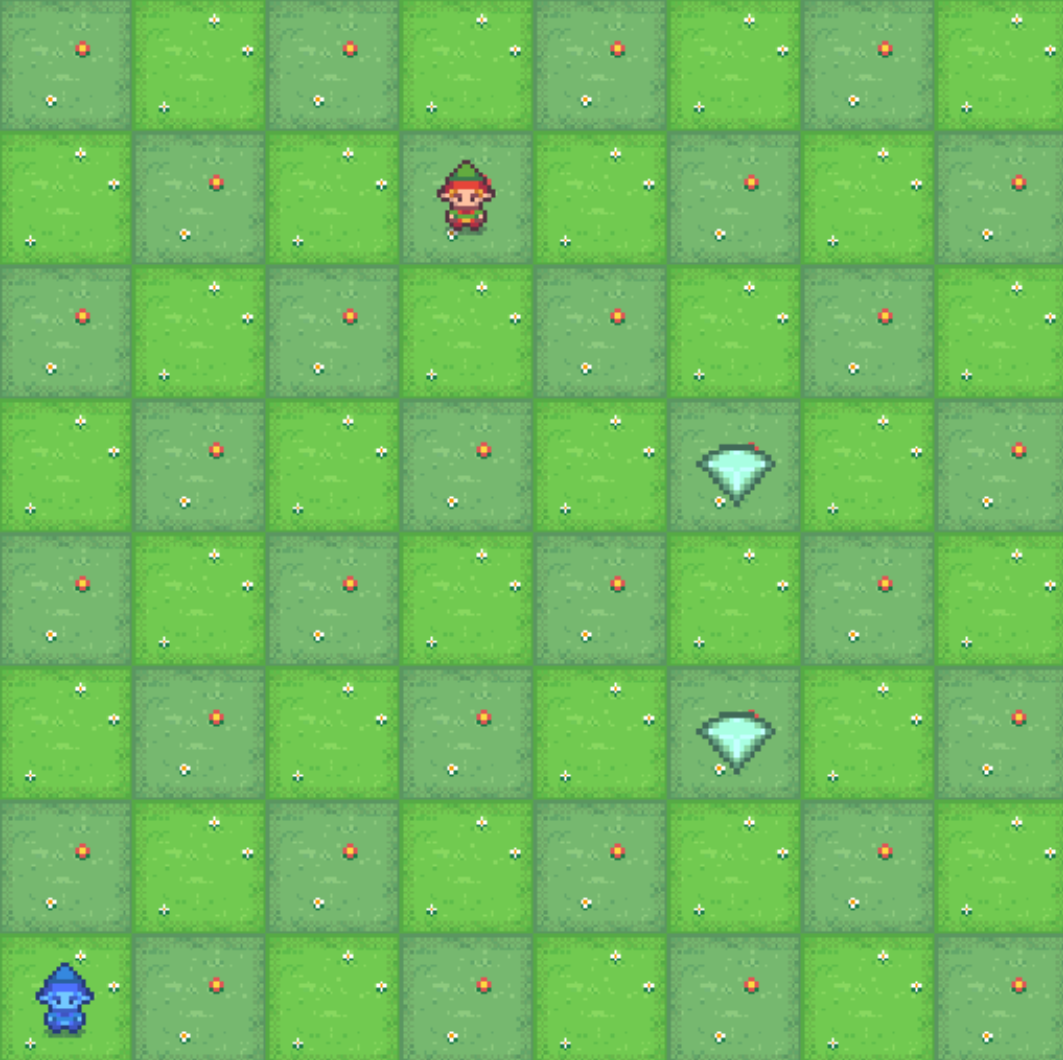}
    \includegraphics[width=0.38\linewidth]{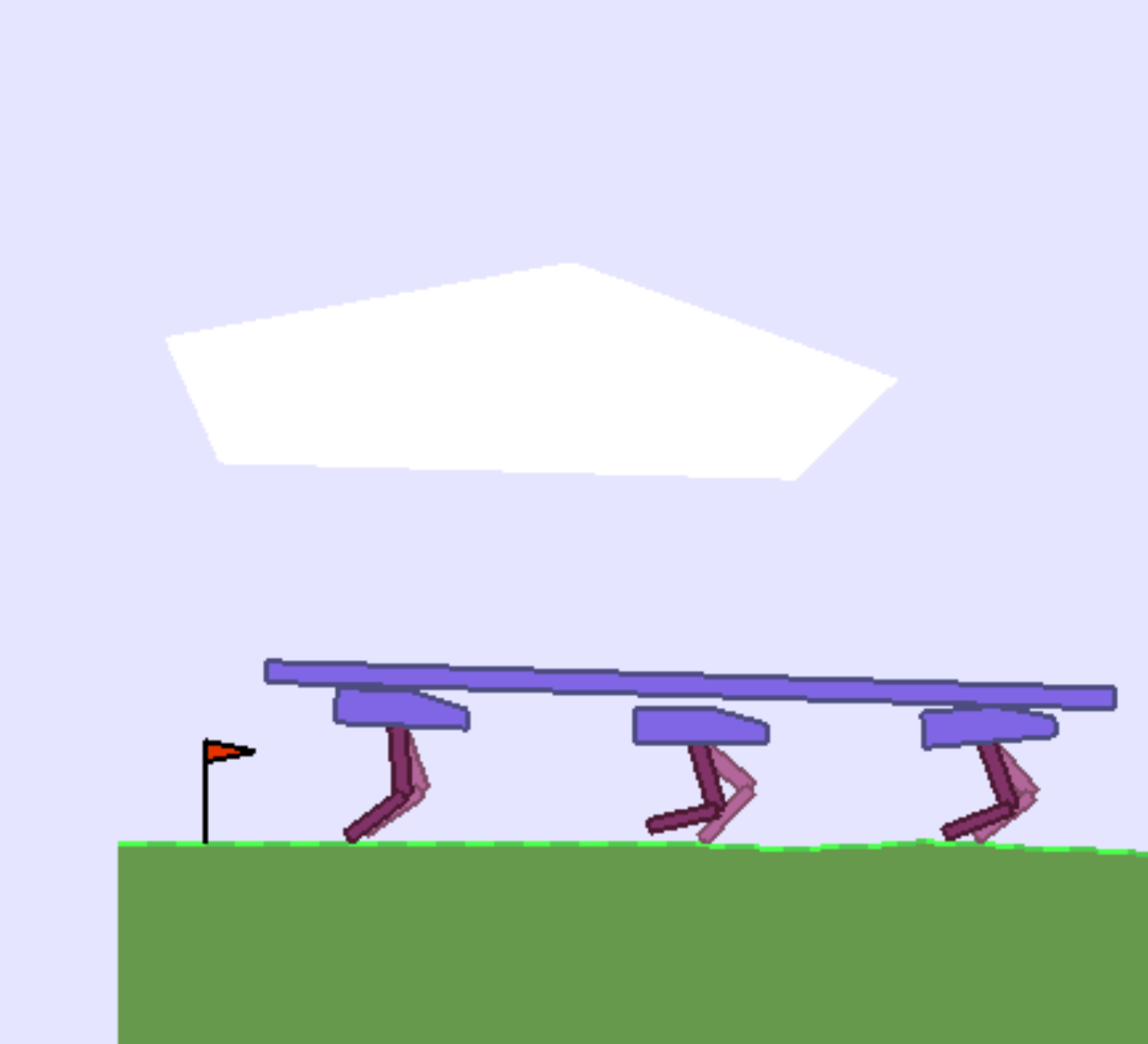}
    \caption{The experimental testbeds. (left) Multi-Objective Item Gathering (MOIG), requiring agents to navigate and select resources based on preference profiles. (right) Multi-Objective MultiWalker (MOMW), challenging agents to maintain dynamic balance and coordination under continuous physical constraints.}
    \label{fig:env}
\end{figure}

We report two complementary metrics. \textbf{Average Utility (AUT)} measures the average utility obtained by the trained agent on randomly sampled preferences and is a useful proxy to Adaptation Error when we do not have access to the optimal policy~\cite{yang2019generalized}. \textbf{Hypervolume (HV)} measures the volume of the objective space dominated by the learned Pareto front relative to a reference point, and is a standard metric in MORL for assessing Pareto front quality. Together, AUT captures average utility under sampled preferences while HV captures the coverage and quality of the full Pareto front. We compare the two metrics for 100 test episodes with uniformly sampled preferences over 5 random seeds. Full results are summarized in Table~\ref{tab:main_comparison}. DML-NUQ achieves the highest AUT and HV in all four settings, confirming superiority both in average utility and Pareto front coverage.

\begin{table}[h]
    \centering\small
    \caption{Performance comparison using Average Utility (AUT) and Hypervolume (HV). The arrow ($\uparrow$) indicates that higher values are preferred. \textbf{DML-NUQ (Ours)} achieves the highest AUT and HV across all four settings.}
    \label{tab:main_comparison}
    \begin{tabular}{lcccc}
        \toprule
        \textbf{Method} & \textbf{MO-IG AUT} ($\uparrow$) & \textbf{MO-IG HV} ($\uparrow$) & \textbf{MO-MW AUT} ($\uparrow$) & \textbf{MO-MW HV} ($\uparrow$) \\
        \midrule
        MOMAPPO & $1.65 \pm 0.25$ & $2.59 \pm 2.14$ & $-98.86 \pm 69.97$ & $4.52 \times 10^4 \pm 2.91 \times 10^4$ \\
        PCN & $1.60 \pm 0.22$ & $4.11 \pm 1.65$ & $-41.75 \pm 36.27$ & $6.80 \times 10^4 \pm 1.90 \times 10^4$ \\
        \textbf{DML-NUQ} & $\mathbf{1.82 \pm 0.19}$ & $\mathbf{4.60 \pm 1.85}$ & $\mathbf{-32.53 \pm 35.92}$ & $\mathbf{7.26 \times 10^4 \pm 1.89 \times 10^4}$ \\
        \bottomrule
    \end{tabular}
\end{table}

\paragraph{Ablation study: contribution of NUQ}
To isolate the contribution of the NUQ component, we compare DML-NUQ against two variants while keeping all other components identical (network architecture, preference conditioning, training schedule, and hyperparameters): (1)~\emph{zero intrinsic reward}, which removes the NUQ component entirely; and (2)~\emph{random intrinsic reward}, which replaces NUQ with a random bonus of the same magnitude drawn from a normal distribution. Results are reported in Table~\ref{tab:ablation}.

Removing NUQ causes severe degradation across all metrics (MO-MW AUT: $-32.53 \to -96.86$; HV: $7.26\times10^4 \to 4.52\times10^4$). Replacing NUQ with random intrinsic reward partially recovers performance but remains below DML-NUQ in all four metrics, confirming that the benefit of NUQ is not attributable to any exploration bonus, but specifically to its uncertainty-driven signal. 
We note that conditioning and preference sampling are components shared identically with the MOMAPPO and PCN baselines; ablating them would test a fundamentally different algorithmic paradigm rather than isolating the contributions of DML-NUQ specifically.

\begin{table}[h]
    \centering\footnotesize
    \caption{Ablation study on the NUQ component. Removing or randomizing NUQ leads to consistent degradation across both environments and both metrics.}
    \label{tab:ablation}
    \begin{tabular}{lcccc}
        \toprule
        \textbf{Method} & \textbf{MO-IG AUT} ($\uparrow$) & \textbf{MO-IG HV} ($\uparrow$) & \textbf{MO-MW AUT} ($\uparrow$) & \textbf{MO-MW HV} ($\uparrow$) \\
        \midrule
        w/ zero int. reward   & $1.65 \pm 0.25$ & $2.59 \pm 2.14$ & $-96.86 \pm 69.97$ & $4.52 \times 10^4 \pm 2.91 \times 10^4$ \\
        w/ random int. reward & $1.74 \pm 0.26$ & $3.62 \pm 2.80$ & $-43.30 \pm 63.86$ & $6.92 \times 10^4 \pm 2.83 \times 10^4$ \\
        \textbf{DML-NUQ (Ours)}    & $\mathbf{1.82 \pm 0.19}$ & $\mathbf{4.60 \pm 1.85}$ & $\mathbf{-32.53 \pm 35.92}$ & $\mathbf{7.26 \times 10^4 \pm 1.89 \times 10^4}$ \\
        \bottomrule
    \end{tabular}
\end{table}

\subsection{Hardware and Software Setup}
Experiments were performed on a high-performance server equipped with 64 physical CPU cores (128 logical cores) and 756 GB of system RAM. To accelerate training, we utilized 8 $\times$ NVIDIA RTX 6000 Ada Generation GPUs, each with 48 GB of VRAM. The system runs on Linux (Kernel 5.15.0) with CUDA 13.0. All algorithms were implemented in Python 3.10 using the JAX framework.

\begin{table}[h]
    \centering
    \caption{Common Hyperparameters for all experiments.}
    \label{tab:common_params}
    \begin{tabular}{ll}
        \toprule
        \textbf{Hyperparameter} & \textbf{Value} \\
        \midrule
        Total Timesteps & $10^7$ \\
        Random Seed & 42 \\
        \midrule
        \textit{Training Strategy} & \\
        Weight Generation Method & OLS \\
        Number of Weights & 10 \\
        Timesteps per Weight & $10^6$ \\
        \midrule
        \textit{PPO \& Optimization} & \\
        Optimizer & Adam ($\epsilon=10^{-5}$) \\
        Learning Rate Schedule & Linear Decay \\
        Discount Factor ($\gamma$) & 0.99 \\
        GAE Parameter ($\lambda$) & 0.99 \\
        PPO Clip Range ($\epsilon$) & 0.2 \\
        Value Function Coeff. & 0.8 \\
        Max Gradient Norm & 0.5 \\
        Update Epochs & 2 \\
        Rollout Length (Steps per Epoch) & 128 \\
        \midrule
        \textit{Network Architecture} & \\
        Actor/Critic Hidden Layers & 2 \\
        Hidden Units per Layer & 256 \\
        Activation Function & Tanh \\
        \bottomrule
    \end{tabular}
\end{table}

\begin{table*}[h]
    \centering
    \caption{Environment-specific hyperparameters and NUQ module configuration.}
    \label{tab:specific_params}
    \begin{tabular}{lcc}
        \toprule
        \textbf{Parameter} & \textbf{MO-Item Gathering (MO-IG)} & \textbf{MO-MultiWalker (MO-MW)} \\
        \midrule
        Action Space & Discrete & Continuous \\
        Learning Rate & $2.5 \times 10^{-4}$ & $1.0 \times 10^{-3}$ \\
        Entropy Coefficient & $0.01$ & $0.0$ \\
        Reference Point & $[0.0, 0.0, 0.0]$ & $[-300.0, -300.0]$\\
        \midrule
        \midrule
        \textbf{NUQ Module (Ours)} & \multicolumn{2}{c}{\textbf{Configuration}} \\
        \midrule
        Intrinsic Reward Coeff. ($\beta$) & \multicolumn{2}{c}{$0.1$} \\
        Predictor Learning Rate & \multicolumn{2}{c}{$1.0 \times 10^{-4}$} \\
        Feature Dimension & \multicolumn{2}{c}{$128$} \\
        Network Architecture & \multicolumn{2}{c}{MLP (2 layers, 128 units, ReLU)} \\
        Initialization & \multicolumn{2}{c}{Orthogonal} \\
        Intrinsic Objective & \multicolumn{2}{c}{MSE (Predictor vs. Target)} \\
        \bottomrule
    \end{tabular}
\end{table*}

\section{Proofs for \Cref{sec:momg}}

This appendix proves the results of \Cref{sec:momg}. \Cref{sec:pne-appendix} shows that every finite MOMG has a PNE (\Cref{thm:appendix-pne-existence}, which restates \Cref{thm:pne-exist}) and that the PNEs are exactly the Nash equilibria of the scalarized games with strictly positive weights (\Cref{thm:pne-scalarization}, which restates \Cref{thm:scalar}). The second result rests on two lemmas: with the policies of the other players fixed, the set of value vectors that a player can reach is a polytope (\Cref{lem:achievable-polytope}), and every Pareto optimal point of a polytope has a strictly positive supporting weight (\Cref{lem:positive-support}). \Cref{sec:wpne-appendix} proves the corresponding results for WPNE (\Cref{thm:wpne}), expresses PNG and WPNG through the scalarized Nash gap, and gives the linear programs that compute WPNG and test for a PNE when the model is known. \Cref{sec:utility-appendix} proves the results on utility-based equilibria stated in \Cref{sec:utility}.

\paragraph{Notation.}
As in \Cref{subsec:finite SOMG}, player $k$ uses a time-dependent randomized Markov policy $\pi_k=(\pi_{k,1},\ldots,\pi_{k,H})\in\Pi_k:=\prod_{h=1}^H\prod_{s\in\mathcal S}\Delta(\mathcal A_k)$, and joint policies are products $\pi=(\pi_1,\ldots,\pi_N)$. We write $\mathbf V^\pi_{k,1}(s_1)\in\mathbb R^M$ for the value vector of player $k$ from the initial state, and
\[
\mathbf d_k(\pi'_k;\pi):=\mathbf V^{(\pi'_k,\pi_{-k})}_{k,1}(s_1)-\mathbf V^{\pi}_{k,1}(s_1)
\]
for the gain vector of a unilateral deviation $\pi'_k\in\Pi_k$, as in \Cref{sec:wpne}. For $\mathbf u,\mathbf v\in\mathbb R^M$, we write $\mathbf u\succeq\mathbf v$ if $u_i\ge v_i$ for all $i\in[M]$. Thus $\mathbf u$ Pareto dominates $\mathbf v$ if $\mathbf u\succeq\mathbf v$ and $\mathbf u\neq\mathbf v$, and $\mathbf u$ strictly dominates $\mathbf v$ if $u_i>v_i$ for all $i\in[M]$. For a preference profile $\bm\Lambda=(\bm\lambda_1,\ldots,\bm\lambda_N)\in(\Delta_M)^N$, the scalarized game $\mathcal G_{\bm\Lambda}$ has the states, actions, and transitions of $\mathcal G$, and player $k$ receives the reward $\bm\lambda_k^\top\mathbf r_{k,h}$. We write $\mathrm{NE}(\mathcal G_{\bm\Lambda})$ for its set of Nash equilibria in product time-dependent randomized Markov policies.

\subsection{Existence and Scalarization of PNE}\label{sec:pne-appendix}

\begin{theorem}[Existence of PNE]\label{thm:appendix-pne-existence}
Every MOMG with finitely many players, states, and actions and a finite horizon has a PNE in time-dependent randomized Markov policies.
\end{theorem}

\begin{proof}
Fix a preference profile $\bm\Lambda=(\bm\lambda_1,\ldots,\bm\lambda_N)\in(\Delta_M^\circ)^N$. Write $r^\Lambda_{k,h}(s,\bm a)=\bm\lambda_k^\top\mathbf r_{k,h}(s,\bm a)$ for the scalarized reward and $U^\pi_{k,h}(s)=\bm\lambda_k^\top\mathbf V^\pi_{k,h}(s)$ for the scalarized value. We first construct a Nash equilibrium $\pi^*$ of $\mathcal G_{\bm\Lambda}$ by backward induction and then show that it is a PNE of $\mathcal G$.

Set $W_{k,H+1}\equiv0$. For $h=H,\ldots,1$ and each state $s$, consider the finite normal-form game in which player $k$ receives
\[
Q_{k,h}(s,\bm a)
=r^\Lambda_{k,h}(s,\bm a)
+\sum_{s'\in\mathcal S}P_h(s'\mid s,\bm a)\,W_{k,h+1}(s').
\]
By Nash's theorem~\cite{fudenberg1991game}, this game has a mixed Nash equilibrium $\mu_h(\bm a\mid s)=\prod_{k=1}^N\mu_{k,h}(a_k\mid s)$. Set
\[
\pi^*_{k,h}(\cdot\mid s)=\mu_{k,h}(\cdot\mid s),
\qquad
W_{k,h}(s)=\sum_{\bm a\in\mathcal A}\mu_h(\bm a\mid s)\,Q_{k,h}(s,\bm a).
\]
This defines a time-dependent randomized Markov profile $\pi^*$, and by construction $W_{k,h}=U^{\pi^*}_{k,h}$.

Fix a player $k$ and a deviation $\rho_k\in\Pi_k$. We show by backward induction on $h$ that $U^{(\rho_k,\pi^*_{-k})}_{k,h}\le W_{k,h}$. The claim holds at $h=H+1$. If it holds at $h+1$, then for every state $s$,
\[
U^{(\rho_k,\pi^*_{-k})}_{k,h}(s)
\le \sum_{\bm a\in\mathcal A}\rho_{k,h}(a_k\mid s)\prod_{\ell\neq k}\pi^*_{\ell,h}(a_\ell\mid s)\,Q_{k,h}(s,\bm a)
\le W_{k,h}(s),
\]
where the first inequality applies the induction hypothesis to the continuation value and the second is the Nash condition of the stage game at $(h,s)$. At $h=1$, this gives
\begin{equation}\label{eq:ne_condition}
\bm\lambda_k^\top\mathbf V^{(\rho_k,\pi^*_{-k})}_{k,1}(s_1)
\le \bm\lambda_k^\top\mathbf V^{\pi^*}_{k,1}(s_1),
\end{equation}
so $\pi^*\in\mathrm{NE}(\mathcal G_{\bm\Lambda})$.

Suppose that some deviation $\rho_k$ reached a value vector that Pareto dominates $\mathbf V^{\pi^*}_{k,1}(s_1)$. Its gain vector would be nonnegative with at least one positive coordinate, and since every coordinate of $\bm\lambda_k$ is positive, the left-hand side of~\eqref{eq:ne_condition} would exceed the right-hand side. This contradiction shows that $\pi^*$ is a PNE.
\end{proof}

\begin{lemma}[Achievable value set under fixed opponents is a polytope]
\label{lem:achievable-polytope}
Fix a player $k$ and time-dependent randomized Markov policies $\pi_\ell\in\Pi_\ell$ for all players $\ell\neq k$, and write $\pi_{-k}=(\pi_\ell)_{\ell\neq k}$. Then the set
\[
\mathcal V_k(\pi_{-k})
:=\bigl\{\mathbf V^{(\pi_k,\pi_{-k})}_{k,1}(s_1):\pi_k\in\Pi_k\bigr\}\subseteq\mathbb R^M
\]
of value vectors that player $k$ can reach by deviating is a compact convex polytope.
\end{lemma}

\begin{proof}
With $\pi_{-k}$ fixed, player $k$ faces a finite-horizon multi-objective MDP with transition kernel and vector reward
\[
P^{\pi_{-k}}_{k,h}(s'\mid s,a_k)
:=\sum_{a_{-k}\in\mathcal A_{-k}}\pi_{-k,h}(a_{-k}\mid s)\,P_h\bigl(s'\mid s,(a_k,a_{-k})\bigr),
\qquad
\mathbf r^{\pi_{-k}}_{k,h}(s,a_k)
:=\sum_{a_{-k}\in\mathcal A_{-k}}\pi_{-k,h}(a_{-k}\mid s)\,\mathbf r_{k,h}\bigl(s,(a_k,a_{-k})\bigr),
\]
where $\pi_{-k,h}(a_{-k}\mid s):=\prod_{\ell\neq k}\pi_{\ell,h}(a_\ell\mid s)$. For every $\pi_k\in\Pi_k$, the states and the actions of player $k$ have the same joint distribution under $(\pi_k,\pi_{-k})$ in $\mathcal G$ as under $\pi_k$ in this MDP, and the conditional expected reward given $(s_h,a_{k,h})$ is $\mathbf r^{\pi_{-k}}_{k,h}(s_h,a_{k,h})$ in both. Hence $\mathbf V^{(\pi_k,\pi_{-k})}_{k,1}(s_1)$ equals the value of $\pi_k$ in the MDP.

The state--action occupancy measure of $\pi_k$ in this MDP is
\[
\theta_h(s,a_k):=\Pr\nolimits_{\pi_k,P^{\pi_{-k}}_k}\bigl(s_h=s,\ a_{k,h}=a_k\bigm| s_1\bigr).
\]
Let $\Theta_k(\pi_{-k})$ be the set of arrays $\theta=(\theta_h(s,a_k))_{h,s,a_k}$ that satisfy
\begin{align*}
&\theta_h(s,a_k)\ge0, && h\in[H],\ s\in\mathcal S,\ a_k\in\mathcal A_k,\\
&\textstyle\sum_{a_k\in\mathcal A_k}\theta_1(s,a_k)=\mathbf 1\{s=s_1\}, && s\in\mathcal S,\\
&\textstyle\sum_{a_k\in\mathcal A_k}\theta_{h+1}(s',a_k)=\sum_{s\in\mathcal S}\sum_{a_k\in\mathcal A_k}\theta_h(s,a_k)\,P^{\pi_{-k}}_{k,h}(s'\mid s,a_k), && h\in[H-1],\ s'\in\mathcal S.
\end{align*}
These are finitely many linear constraints. By induction on $h$, they imply $\sum_{s,a_k}\theta_h(s,a_k)=1$ for every $h$, so $\Theta_k(\pi_{-k})$ is a bounded polyhedron, hence a compact convex polytope.

Every policy $\pi_k\in\Pi_k$ induces an occupancy measure in $\Theta_k(\pi_{-k})$. Conversely, every $\theta\in\Theta_k(\pi_{-k})$ is the occupancy measure of the policy
\[
\pi^\theta_{k,h}(a_k\mid s)=\frac{\theta_h(s,a_k)}{\sum_{b_k\in\mathcal A_k}\theta_h(s,b_k)},
\]
with an arbitrary distribution in $\Delta(\mathcal A_k)$ where the denominator is zero; this follows from the flow constraints by induction on $h$. Finally, the value of $\pi_k$ is linear in its occupancy measure $\theta$:
\[
\mathbf V^{(\pi_k,\pi_{-k})}_{k,1}(s_1)=L_k(\theta),
\qquad
L_k(\theta):=\sum_{h=1}^H\sum_{s\in\mathcal S}\sum_{a_k\in\mathcal A_k}\theta_h(s,a_k)\,\mathbf r^{\pi_{-k}}_{k,h}(s,a_k).
\]
Therefore $\mathcal V_k(\pi_{-k})=L_k\bigl(\Theta_k(\pi_{-k})\bigr)$ is the image of a compact convex polytope under a linear map, and hence is itself a compact convex polytope.
\end{proof}

The results below use three properties of $\mathcal V_k(\pi_{-k})$ and $\Theta_k(\pi_{-k})$: the polytope structure (\Cref{thm:pne-scalarization}, through \Cref{lem:positive-support}), convexity (\Cref{thm:wpne-scalarization} and \Cref{par:gap-minimax}), and compactness (\Cref{par:gap-minimax} and the linear programs at the end of \Cref{sec:wpne-appendix}).

\begin{lemma}[Positive supporting weight for a Pareto optimal point of a polytope]
\label{lem:positive-support}
Let $C\subset\mathbb R^M$ be a compact convex polytope, and let $\mathbf v^\star\in C$ be Pareto optimal in $C$, that is, no $\mathbf v\in C$ satisfies $\mathbf v\succeq\mathbf v^\star$ and $\mathbf v\neq\mathbf v^\star$. Then there exists $\bm\lambda\in\Delta_M^\circ$ such that $\bm\lambda^\top\mathbf v^\star\ge\bm\lambda^\top\mathbf v$ for all $\mathbf v\in C$.
\end{lemma}

This is a classical fact from multiple-objective linear programming~\cite{EvansSteuer1973}; we include a proof for completeness. The polytope structure is needed: in the unit disk, the point $(1,0)$ is Pareto optimal, but its only supporting weight is $(1,0)$, which is not strictly positive.

\begin{proof}
Since $C$ is a compact convex polytope, it is the convex hull of finitely many points, $C=\operatorname{conv}\{\mathbf c^1,\ldots,\mathbf c^L\}$. Let $\mathbf d^\ell:=\mathbf c^\ell-\mathbf v^\star$ for $\ell\in[L]$, and let
\[
D:=\operatorname{cone}\{\mathbf d^1,\ldots,\mathbf d^L\}=\Bigl\{\textstyle\sum_{\ell=1}^L\alpha_\ell\mathbf d^\ell:\alpha_\ell\ge0\Bigr\}.
\]

\emph{Step 1: $D\cap\mathbb R_+^M=\{\mathbf 0\}$.} Suppose that $\mathbf z=\sum_\ell\alpha_\ell\mathbf d^\ell\in\mathbb R_+^M\setminus\{\mathbf 0\}$ with all $\alpha_\ell\ge0$. Then $\alpha:=\sum_\ell\alpha_\ell>0$, and $\bar{\mathbf v}:=\sum_\ell(\alpha_\ell/\alpha)\,\mathbf c^\ell\in C$ satisfies $\bar{\mathbf v}-\mathbf v^\star=\mathbf z/\alpha\in\mathbb R_+^M\setminus\{\mathbf 0\}$. Thus $\bar{\mathbf v}$ Pareto dominates $\mathbf v^\star$, a contradiction.

\emph{Step 2: a strictly positive weight.} Let $A\in\mathbb R^{M\times L}$ be the matrix with columns $\mathbf d^1,\ldots,\mathbf d^L$. We claim that the system
\[
A^\top\bm\lambda\le\mathbf 0,\qquad \bm\lambda\ge\mathbf 1
\]
(componentwise) has a solution. By Farkas' lemma, a system $G\mathbf x\le\mathbf g$ has no solution if and only if there is $\mathbf y\ge\mathbf 0$ with $G^\top\mathbf y=\mathbf 0$ and $\mathbf g^\top\mathbf y<0$. With $G=\begin{bmatrix}A^\top\\-I_M\end{bmatrix}$ and $\mathbf g=\begin{bmatrix}\mathbf 0\\-\mathbf 1\end{bmatrix}$, infeasibility would therefore give $\bm\alpha\in\mathbb R_+^L$ and $\bm\beta\in\mathbb R_+^M$ with $A\bm\alpha=\bm\beta$ and $\mathbf 1^\top\bm\beta>0$. Then $\bm\beta=\sum_\ell\alpha_\ell\mathbf d^\ell$ is a nonzero element of $D\cap\mathbb R_+^M$, which contradicts Step~1. Hence some $\bm\lambda$ satisfies $\lambda_i\ge1$ for all $i$ and $\bm\lambda^\top(\mathbf c^\ell-\mathbf v^\star)\le0$ for all $\ell$; in particular, every component of $\bm\lambda$ is strictly positive.

Every $\mathbf v\in C$ is a convex combination $\sum_\ell\rho_\ell\mathbf c^\ell$, so $\bm\lambda^\top(\mathbf v-\mathbf v^\star)=\sum_\ell\rho_\ell\,\bm\lambda^\top(\mathbf c^\ell-\mathbf v^\star)\le0$. Since every component of $\bm\lambda$ is strictly positive, $\bar{\bm\lambda}:=\bm\lambda/\sum_i\lambda_i$ lies in $\Delta_M^\circ$, and normalization preserves the inequality $\bar{\bm\lambda}^\top\mathbf v^\star\ge\bar{\bm\lambda}^\top\mathbf v$ for all $\mathbf v\in C$.
\end{proof}

\begin{theorem}[Linear scalarization and Pareto--Nash equilibria]\label{thm:pne-scalarization}
Let $\mathcal G$ be a finite-horizon MOMG with finitely many players, states, and actions. For a preference profile $\bm\Lambda\in(\Delta_M)^N$, let $\mathcal G_{\bm\Lambda}$ be the scalarized game in which player $k$ receives the reward $\bm\lambda_k^\top\mathbf r_{k,h}$, and let $\mathrm{NE}(\mathcal G_{\bm\Lambda})$ be its set of Nash equilibria in product time-dependent randomized Markov policies. Then
\[
\mathrm{PNE}(\mathcal G)=\bigcup_{\bm\Lambda\in(\Delta_M^\circ)^N}\mathrm{NE}(\mathcal G_{\bm\Lambda}).
\]
\end{theorem}

\begin{proof}
\emph{Inclusion $\supseteq$.} Let $\bm\Lambda\in(\Delta_M^\circ)^N$ and $\pi^\star\in\mathrm{NE}(\mathcal G_{\bm\Lambda})$. The last paragraph of the proof of \Cref{thm:appendix-pne-existence} uses only the Nash condition~\eqref{eq:ne_condition} and the positivity of the weights, so it applies to $\pi^\star$: a deviation whose value vector Pareto dominates $\mathbf V^{\pi^\star}_{k,1}(s_1)$ would strictly increase the scalarized value of player $k$. Hence $\pi^\star$ is a PNE.

\emph{Inclusion $\subseteq$.} Let $\pi^\star\in\mathrm{PNE}(\mathcal G)$ and fix a player $k$. By \Cref{lem:achievable-polytope}, the set $C_k:=\mathcal V_k(\pi^\star_{-k})$ is a compact convex polytope, and it contains $\mathbf v^\star_k:=\mathbf V^{\pi^\star}_{k,1}(s_1)$. Since $\pi^\star$ is a PNE, no deviation of player $k$ reaches a value vector that Pareto dominates $\mathbf v^\star_k$; that is, $\mathbf v^\star_k$ is Pareto optimal in $C_k$. By \Cref{lem:positive-support}, there is $\bm\lambda_k\in\Delta_M^\circ$ with $\bm\lambda_k^\top\mathbf v^\star_k\ge\bm\lambda_k^\top\mathbf v$ for all $\mathbf v\in C_k$, that is,
\[
\bm\lambda_k^\top\mathbf V^{\pi^\star}_{k,1}(s_1)\ge\bm\lambda_k^\top\mathbf V^{(\pi_k,\pi^\star_{-k})}_{k,1}(s_1)
\qquad\text{for all }\pi_k\in\Pi_k.
\]
Thus $\pi^\star_k$ is a best response to $\pi^\star_{-k}$ when the reward of player $k$ is scalarized by $\bm\lambda_k$. Choosing such a $\bm\lambda_k$ for every player gives a profile $\bm\Lambda=(\bm\lambda_1,\ldots,\bm\lambda_N)\in(\Delta_M^\circ)^N$ with $\pi^\star\in\mathrm{NE}(\mathcal G_{\bm\Lambda})$.
\end{proof}

\begin{proposition}\label{prop:png-appendix}
If $\pi$ is a PNE, then $\mathrm{PNG}(\pi)=0$. The converse does not hold.
\end{proposition}

\begin{proof}
Recall that $\mathrm{PNG}(\pi)=\max_{k\in[N]}\sup_{\pi'_k\in\Pi_k}\inf_{\bm\lambda\in\Delta_M^\circ}\bm\lambda^\top\mathbf d_k(\pi'_k;\pi)$. For every $\mathbf d\in\mathbb R^M$,
\[
\inf_{\bm\lambda\in\Delta_M^\circ}\bm\lambda^\top\mathbf d
=\min_{\bm\lambda\in\Delta_M}\bm\lambda^\top\mathbf d
=\min_{i\in[M]}d_i,
\]
because $\Delta_M^\circ$ is dense in $\Delta_M$ and a linear function attains its minimum over $\Delta_M$ at a vertex. If $\pi$ is a PNE, no gain vector $\mathbf d_k(\pi'_k;\pi)$ is nonnegative with a positive coordinate, so each of them has a nonpositive coordinate and $\mathrm{PNG}(\pi)\le0$. The deviation $\pi'_k=\pi_k$ has zero gain, so $\mathrm{PNG}(\pi)=0$.

For the converse, consider one player with two objectives, one state, one step, and two actions with rewards $\mathbf r(a_1)=(0,0)$ and $\mathbf r(a_2)=(1,0)$. The policy $\pi$ that plays $a_1$ is not a PNE, since $a_2$ Pareto dominates it. A deviation that plays $a_2$ with probability $p\in[0,1]$ has gain $(p,0)$, so
\[
\mathrm{PNG}(\pi)=\sup_{p\in[0,1]}\inf_{\bm\lambda\in\Delta_2^\circ}\lambda_1p=0.
\]
Adding players with a single action, who affect neither rewards nor transitions, gives a multi-player example.
\end{proof}

\subsection{WPNE and the Pareto--Nash Gaps}\label{sec:wpne-appendix}

This subsection proves \Cref{thm:wpne}, expresses PNG and WPNG through the scalarized Nash gap, and gives the linear programs mentioned in \Cref{sec:wpne}.

\begin{theorem}\label{thm:wpne-scalarization}
Let $\mathcal G$ and $\mathrm{NE}(\mathcal G_{\bm\Lambda})$ be as in \Cref{thm:pne-scalarization}. The WPNEs of $\mathcal G$ are exactly the Nash equilibria of the scalarized games with nonnegative weights:
\[
\mathrm{WPNE}(\mathcal G)=\bigcup_{\bm\Lambda\in(\Delta_M)^N}\mathrm{NE}(\mathcal G_{\bm\Lambda}).
\]
\end{theorem}

\begin{proof}
\emph{Inclusion $\supseteq$.} Let $\bm\Lambda\in(\Delta_M)^N$ and $\pi\in\mathrm{NE}(\mathcal G_{\bm\Lambda})$. Suppose that some player $k$ has a deviation $\pi'_k$ whose gain vector $\mathbf d=\mathbf d_k(\pi'_k;\pi)$ satisfies $d_i>0$ for all $i\in[M]$. Since $\bm\lambda_k\ge\mathbf 0$ and $\bm\lambda_k\neq\mathbf 0$, we have $\bm\lambda_k^\top\mathbf d>0$, which contradicts the Nash condition for player $k$ in $\mathcal G_{\bm\Lambda}$. Hence $\pi$ is a WPNE.

\emph{Inclusion $\subseteq$.} Let $\pi$ be a WPNE and fix a player $k$. By \Cref{lem:achievable-polytope}, the set $\mathcal V_k(\pi_{-k})$ is convex, and it contains $\mathbf v^*:=\mathbf V^\pi_{k,1}(s_1)$. Since $\pi$ is a WPNE, no point of $\mathcal V_k(\pi_{-k})$ strictly dominates $\mathbf v^*$, so $\mathcal V_k(\pi_{-k})$ is disjoint from the convex set $B:=\{\mathbf u\in\mathbb R^M:u_i>v^*_i\text{ for all }i\}$. By the separating hyperplane theorem, there are a nonzero $\bm\mu\in\mathbb R^M$ and $c\in\mathbb R$ with
\[
\bm\mu^\top\mathbf v\le c\le\bm\mu^\top\mathbf u
\qquad\text{for all }\mathbf v\in\mathcal V_k(\pi_{-k})\text{ and }\mathbf u\in B.
\]
Since $\mathbf v^*\in\mathcal V_k(\pi_{-k})$ and $\mathbf v^*$ lies in the closure of $B$, we have $c=\bm\mu^\top\mathbf v^*$. For every $i$ and $t>0$, the point $\mathbf u=\mathbf v^*+t\mathbf e_i+\mathbf 1$ lies in $B$, so $c\le\bm\mu^\top\mathbf u$ gives $t\mu_i+\mathbf 1^\top\bm\mu\ge0$; letting $t\to\infty$ shows that $\mu_i\ge0$. Thus $\bm\mu\ge\mathbf 0$ and $\bm\mu\neq\mathbf 0$, so $\mathbf 1^\top\bm\mu>0$, and $\bm\lambda_k:=\bm\mu/(\mathbf 1^\top\bm\mu)\in\Delta_M$ satisfies $\bm\lambda_k^\top\mathbf v^*\ge\bm\lambda_k^\top\mathbf v$ for all $\mathbf v\in\mathcal V_k(\pi_{-k})$, that is, $\pi_k$ is a best response to $\pi_{-k}$ when the reward of player $k$ is scalarized by $\bm\lambda_k$. Choosing such a $\bm\lambda_k$ for every player gives a profile $\bm\Lambda\in(\Delta_M)^N$ with $\pi\in\mathrm{NE}(\mathcal G_{\bm\Lambda})$.
\end{proof}

The inclusion $\subseteq$ uses only the convexity of $\mathcal V_k(\pi_{-k})$ and yields weights that may vanish. Strictly positive weights, as in \Cref{thm:pne-scalarization}, require the polytope structure (\Cref{lem:positive-support}).

\begin{theorem}\label{thm:wpne-zero-gap}
For every product policy $\pi$, $\mathrm{WPNG}(\pi)=0$ if and only if $\pi$ is a WPNE.
\end{theorem}

\begin{proof}
Since $\min_{\bm\lambda\in\Delta_M}\bm\lambda^\top\mathbf d=\min_i d_i$ for every $\mathbf d\in\mathbb R^M$,
\[
\mathrm{WPNG}(\pi)=\max_{k\in[N]}\sup_{\pi'_k\in\Pi_k}\min_{i\in[M]}d_{k,i}(\pi'_k;\pi).
\]
The deviation $\pi'_k=\pi_k$ has zero gain, so $\mathrm{WPNG}(\pi)\ge0$. If $\pi$ is a WPNE, then for every player $k$ and every deviation $\pi'_k$, some coordinate of $\mathbf d_k(\pi'_k;\pi)$ is nonpositive, so $\min_i d_{k,i}(\pi'_k;\pi)\le0$ and $\mathrm{WPNG}(\pi)\le0$. If $\pi$ is not a WPNE, some player $k$ has a deviation $\pi'_k$ with $d_{k,i}(\pi'_k;\pi)>0$ for all $i$, so $\mathrm{WPNG}(\pi)\ge\min_i d_{k,i}(\pi'_k;\pi)>0$.
\end{proof}

\begin{proposition}[PNG and WPNG through the scalarized Nash gap]\label{par:gap-minimax}
For every product policy $\pi$,
\[
\mathrm{PNG}(\pi)=\mathrm{WPNG}(\pi)
=\min_{\bm\Lambda\in(\Delta_M)^N}\mathrm{NashGap}_{\mathcal G_{\bm\Lambda}}(\pi)
=\inf_{\bm\Lambda\in(\Delta_M^\circ)^N}\mathrm{NashGap}_{\mathcal G_{\bm\Lambda}}(\pi).
\]
Consequently, $\pi$ is a WPNE if and only if this minimum is zero, and $\pi$ is a PNE if and only if $\mathrm{NashGap}_{\mathcal G_{\bm\Lambda}}(\pi)=0$ for some $\bm\Lambda\in(\Delta_M^\circ)^N$.
\end{proposition}

\begin{proof}
Fix $\pi$ and a player $k$. Let $D_k:=\mathcal V_k(\pi_{-k})-\mathbf V^\pi_{k,1}(s_1)$ be the set of gain vectors of player $k$, and let $R_k(\bm\lambda):=\max_{\mathbf d\in D_k}\bm\lambda^\top\mathbf d$ be the scalarized Nash gap of player $k$ under the weight $\bm\lambda$. By \Cref{lem:achievable-polytope}, $D_k$ is nonempty, compact, and convex; it contains $\mathbf 0$ and lies in $[-H,H]^M$. As in the proof of \Cref{prop:png-appendix}, $\inf_{\bm\lambda\in\Delta_M^\circ}\bm\lambda^\top\mathbf d=\min_{\bm\lambda\in\Delta_M}\bm\lambda^\top\mathbf d$ for every $\mathbf d$, so PNG and WPNG coincide, and the term of player $k$ in both is
\[
\gamma_k:=\max_{\mathbf d\in D_k}\min_{\bm\lambda\in\Delta_M}\bm\lambda^\top\mathbf d
=\min_{\bm\lambda\in\Delta_M}\max_{\mathbf d\in D_k}\bm\lambda^\top\mathbf d
=\min_{\bm\lambda\in\Delta_M}R_k(\bm\lambda).
\]
The middle equality is the minimax theorem for a bilinear function on a product of two compact convex sets. The function $R_k$ is $H$-Lipschitz with respect to $\|\cdot\|_1$, so it attains its minimum on $\Delta_M$.

Since $\mathrm{NashGap}_{\mathcal G_{\bm\Lambda}}(\pi)=\max_kR_k(\bm\lambda_k)$ and each $\bm\lambda_k$ can be chosen separately,
\[
\min_{\bm\Lambda\in(\Delta_M)^N}\mathrm{NashGap}_{\mathcal G_{\bm\Lambda}}(\pi)
=\max_{k\in[N]}\min_{\bm\lambda\in\Delta_M}R_k(\bm\lambda)
=\max_{k\in[N]}\gamma_k
=\mathrm{WPNG}(\pi).
\]
For the infimum over interior profiles, let $\bm\lambda^*_k$ minimize $R_k$ and set $\bm\lambda^\eta_k=(1-\eta)\bm\lambda^*_k+\eta\mathbf 1/M$ for $\eta\in(0,1)$. The profile $\bm\Lambda^\eta=(\bm\lambda^\eta_1,\ldots,\bm\lambda^\eta_N)$ lies in $(\Delta_M^\circ)^N$, and $\|\bm\lambda^\eta_k-\bm\lambda^*_k\|_1\le2\eta$, so $\mathrm{NashGap}_{\mathcal G_{\bm\Lambda^\eta}}(\pi)\le\mathrm{WPNG}(\pi)+2H\eta$. Letting $\eta\downarrow0$ shows that the infimum over interior profiles is at most the minimum over all profiles; the reverse inequality holds because the infimum is taken over a subset.

The two consequences follow from \Cref{thm:wpne-zero-gap} and \Cref{thm:pne-scalarization}: $\pi$ is a WPNE if and only if $\mathrm{WPNG}(\pi)=0$, and $\pi$ is a PNE if and only if $\pi\in\mathrm{NE}(\mathcal G_{\bm\Lambda})$, that is, $\mathrm{NashGap}_{\mathcal G_{\bm\Lambda}}(\pi)=0$, for some $\bm\Lambda\in(\Delta_M^\circ)^N$.
\end{proof}

In the example of \Cref{prop:png-appendix}, the policy that plays $a_1$ has scalarized Nash gap $\lambda_1$ under the weight $(\lambda_1,\lambda_2)$. This gap tends to zero as $\lambda_1\to0$ but is positive at every interior weight, so the infimum over interior profiles is zero but not attained. The policy is therefore a WPNE but not a PNE.

\paragraph{Computing WPNG with a known model.}
Fix a product policy $\pi$ and a player $k$, and let $\Theta_k(\pi_{-k})$ and $L_k$ be the occupancy polytope and the linear value map of \Cref{lem:achievable-polytope}. Write
\[
g_{k,i}(\theta):=L_{k,i}(\theta)-V^{\pi,i}_{k,1}(s_1)
\]
for the gain in objective $i$. The term of player $k$ in $\mathrm{WPNG}(\pi)$, that is, the largest amount by which player $k$ can improve all of its objectives simultaneously, is the optimal value $t_k^*$ of the linear program
\begin{equation}\label{eq:wpng-lp}
\begin{aligned}
\max_{\theta,t}\quad & t \\
\text{s.t.}\quad & \theta\in\Theta_k(\pi_{-k}),\\
& t\le g_{k,i}(\theta),\qquad i\in[M].
\end{aligned}
\end{equation}
Indeed, for a fixed $\theta$, the largest feasible $t$ is $\min_i g_{k,i}(\theta)$, and by \Cref{lem:achievable-polytope} the vectors $L_k(\theta)$ with $\theta\in\Theta_k(\pi_{-k})$ are exactly the value vectors that player $k$ can reach. Hence $\mathrm{WPNG}(\pi)=\max_{k\in[N]}t_k^*$, and computing WPNG requires one linear program per player. The program is feasible with value at least zero, since the occupancy measure of $\pi_k$ gives $t=0$. Its optimum is attained because $\Theta_k(\pi_{-k})$ is compact, and \Cref{lem:achievable-polytope} turns an optimal occupancy measure into a Markov deviation that attains $t_k^*$.

\paragraph{Testing for a PNE with a known model.}
A Pareto improvement may leave some objectives unchanged. To detect such a deviation, we maximize the total gain subject to no objective becoming worse~\cite{benson1978existence}:
\begin{equation}\label{eq:pne-test-lp}
\begin{aligned}
\max_{\theta}\quad & \sum_{i=1}^M g_{k,i}(\theta) \\
\text{s.t.}\quad & \theta\in\Theta_k(\pi_{-k}),\\
& g_{k,i}(\theta)\ge0,\qquad i\in[M].
\end{aligned}
\end{equation}
The occupancy measure of $\pi_k$ is feasible with value zero. On the feasible set all gains are nonnegative, so the optimal value is positive exactly when some deviation improves at least one objective without worsening any other. Therefore $\pi$ is a PNE if and only if the optimal value of~\eqref{eq:pne-test-lp} is zero for every player.

\begin{remark}[Boundary preferences]\label{rem:boundary-preferences}
By \Cref{thm:pne-scalarization,thm:wpne-scalarization}, a WPNE that is not a PNE is a Nash equilibrium of some scalarized game, but only of scalarized games whose preference profile has at least one zero weight:
\[
\mathrm{WPNE}(\mathcal G)\setminus\mathrm{PNE}(\mathcal G)
\subseteq
\bigcup_{\bm\Lambda\in(\Delta_M)^N\setminus(\Delta_M^\circ)^N}
\mathrm{NE}(\mathcal G_{\bm\Lambda}).
\]
The inclusion can be strict, because a profile on the boundary may also support a PNE. In particular, if preference profiles are drawn from a distribution that puts no mass on the boundary, such as independent Dirichlet distributions with positive parameters, then with probability one every exact Nash equilibrium of the sampled scalarized game is a PNE.
\end{remark}

\subsection{Utility-Based Equilibria}\label{sec:utility-appendix}

Fix utilities $u_k:\mathbb R_+^M\to\mathbb R$. With the stagewise ESR convention of \Cref{sec:utility}, the two criteria are
\[
J_k^{\mathrm{ESR}}(\pi)
=\mathbb E_{\pi,P}\!\left[\sum_{h=1}^H u_k\bigl(\mathbf r_{k,h}(s_h,\bm a_h)\bigr)\right],
\qquad
J_k^{\mathrm{SER}}(\pi)
=u_k\bigl(\mathbf V^\pi_{k,1}(s_1)\bigr).
\]
For either criterion $J$, a Nash equilibrium is a product of time-dependent randomized Markov policies $\pi$ with $J_k(\pi)\ge J_k(\pi'_k,\pi_{-k})$ for every $k\in[N]$ and $\pi'_k\in\Pi_k$. For nonlinear utilities, stagewise ESR differs from trajectory-level ESR, $\mathbb E_{\pi,P}[u_k(\sum_h\mathbf r_{k,h}(s_h,\bm a_h))]$. The reduction to a scalar Markov game below applies to the stagewise criterion; the two criteria coincide in the one-step examples used here.

\begin{proposition}\label{prop:utility-existence}
An ESR-NE exists in every finite MOMG. An SER-NE need not exist, even in a one-state, one-step two-player game with rewards in $[0,1]^2$ and continuous utilities that are strictly increasing in each coordinate on $\mathbb R_+^2$.
\end{proposition}

\begin{proof}
For ESR, define the scalar rewards $\bar r_{k,h}(s,\bm a)=u_k(\mathbf r_{k,h}(s,\bm a))$. The backward-induction construction in the proof of \Cref{thm:appendix-pne-existence}, applied to these rewards, gives a Nash equilibrium in time-dependent randomized Markov policies, which is an ESR-NE of the original game.

For SER, consider the following normalized version of the two-player example of~\cite{ropke2022nash}. Both players have actions $A$ and $B$, and each entry lists the reward vectors of the row and the column player:
\[
\begin{array}{c|cc}
 & A & B\\ \hline
A & ((1,0),(1/2,0)) & ((1/2,0),(0,1))\\
B & ((0,1/2),(1,0)) & ((0,1),(0,1/2))
\end{array}
\]
Both players use $u(\mathbf v)=v_1^2+v_2^2$, which is continuous, strictly convex, and strictly increasing in each coordinate on $\mathbb R_+^2$. Let $p$ and $q$ be the probabilities with which the row and the column player choose $A$. Against $q$, the two actions of the row player give the expected reward vectors
\[
\mathbf x_A(q)=\bigl((1+q)/2,\,0\bigr),\qquad \mathbf x_B(q)=\bigl(0,\,1-q/2\bigr),
\]
which are distinct for every $q\in[0,1]$. For $0<p<1$, strict convexity therefore gives
\[
u\bigl(p\mathbf x_A+(1-p)\mathbf x_B\bigr)
<p\,u(\mathbf x_A)+(1-p)\,u(\mathbf x_B)
\le\max\{u(\mathbf x_A),u(\mathbf x_B)\},
\]
so every best response of the row player is pure. The two actions of the column player give $(1-p/2,0)$ and $(0,(1+p)/2)$, which are also distinct, so its best responses are pure as well. An SER-NE would therefore be a pure profile. In each pure profile, however, one player has utility $1/4$ and obtains utility $1$ by switching its action, along the cycle
\[
(A,A)\longrightarrow(A,B)\longrightarrow(B,B)
\longrightarrow(B,A)\longrightarrow(A,A).
\]
Hence no SER-NE exists.
\end{proof}

\begin{proposition}\label{prop:esr-pareto}
An ESR-NE need not be Pareto optimal, even in a one-state, one-step MOMDP with rewards in $[0,1]^2$ and a concave utility that is strictly increasing in each coordinate.
\end{proposition}

\begin{proof}
First consider three actions with rewards
\[
\mathbf r(a_1)=(0.3,0.3),\qquad
\mathbf r(a_2)=(0.1,1),\qquad
\mathbf r(a_3)=(1,0.1),
\]
and the utility $u(\mathbf v)=\min\{v_1,v_2\}$. An action distribution $x=(x_1,x_2,x_3)$ has ESR value
\[
J^{\mathrm{ESR}}(x)=0.3x_1+0.1x_2+0.1x_3=0.1+0.2x_1,
\]
so $a_1$ is the unique optimal policy. Yet the equal mixture of $a_2$ and $a_3$ has expected reward $(0.55,0.55)$, which strictly dominates $(0.3,0.3)$.

This utility is not strictly increasing, but the same happens with one that is. Take
\[
\mathbf r(a_1)=(0.4,0.4),\qquad
\mathbf r(a_2)=(0,1),\qquad
\mathbf r(a_3)=(1,0),
\]
and $u(\mathbf v)=\sqrt{v_1}+\sqrt{v_2}$ on $\mathbb R_+^2$, which is concave and strictly increasing in each coordinate. The ESR values of the three actions are $2\sqrt{0.4}$, $1$, and $1$. Since the ESR value is linear in the action probabilities and $2\sqrt{0.4}>1$, playing $a_1$ is again the unique optimal policy. The equal mixture of $a_2$ and $a_3$ has expected reward $(0.5,0.5)$, which strictly dominates $(0.4,0.4)$.
\end{proof}

\section{Optimistic Nash Value Iteration}\label{sec:nash}

This appendix specifies ONVI-MG (Algorithm~\ref{alg:onvi-mg}) and proves its guarantee (\Cref{thm:NASH-reg}), which is summarized in \Cref{sec:single-solution}. The preference profile $\bm\Lambda=(\bm\lambda_1,\ldots,\bm\lambda_N)\in(\Delta_M)^N$ is fixed before data collection, so by~\eqref{eq:pub} it suffices to learn an approximate Nash equilibrium of the scalarized game $\mathcal G_{\bm\Lambda}$. ONVI-MG follows Nash value iteration~\cite{liu2021sharp}: in every episode, it computes optimistic scalarized values from the empirical model, selects a Nash equilibrium of each optimistic stage game, and executes the resulting policy. It also computes pessimistic values of the same policy. The difference between the two at $s_1$ upper bounds the Nash gap of the executed policy and serves as a certificate for choosing the returned policy.

The proof proceeds in five steps. \Cref{lemma1} gives confidence bounds for the reward and transition estimates that hold at all episodes. \Cref{lem:onvi-optimism} shows that the optimistic values upper bound the value of every unilateral deviation from the executed policy. \Cref{lemma3} bounds the Nash gap of each executed policy, and its certificate, by the expected sum of exploration bonuses along its trajectory. \Cref{lemma4} bounds the sum of the bonuses along the realized trajectories, and \Cref{thm:onvi-cumulative} relates the expected and realized sums by a martingale argument. \Cref{thm:NASH-reg} follows because the returned policy has the smallest certificate, which is at most the average certificate.

\paragraph{Notation.}
Write
\[
r_{k,h}^{\Lambda}(s,a)=\bm\lambda_k^\top\mathbf r_{k,h}(s,a),
\qquad U_h^{k,\pi}(s)=\bm\lambda_k^\top\mathbf V_{k,h}^{\pi}(s),
\qquad U_k^\pi=U_1^{k,\pi}.
\]
All policies are products of time-dependent randomized Markov policies, and $a\in\mathcal A$ denotes a joint action. Let $S=|\mathcal S|$ and $A=|\mathcal A|=\prod_k|\mathcal A_k|$.

\paragraph{Bounded-observation model.}
Each episode starts at $s_1$ and contains $H$ transitions. At $(h,s,a)$, the learner observes the reward vector $\mathbf R_{k,h}\in[0,1]^M$ of every player and the next state. Conditional on $(h,s,a)$, these observations follow a fixed joint distribution with reward means $\mathbf r_{k,h}(s,a)$ and next-state marginal $P_h(\cdot\mid s,a)$, independently of the earlier history. The coordinates and the players' rewards within one observation may be dependent. The players share these observations, or equivalently a centralized learner collects them.

\paragraph{Empirical model and confidence radii.}
Let $N_h^t(s,a)$ count the visits to $(h,s,a)$ in the first $t$ episodes, and let $\hat{\mathbf r}_{k,h}^t$ and $\hat P_h^t$ be the corresponding empirical reward means and transition frequencies. At unvisited triples, set $\hat{\mathbf r}_{k,h}^t=0$ and $\hat P_h^t(\cdot\mid s,a)=\mathrm{Uniform}(\mathcal S)$. Episode $t$ is planned using only the estimates after $t-1$ episodes. For $T\ge1$ and $\delta\in(0,1)$, define
\begin{equation}\label{eq:onvi-bonuses}
L=\log\frac{16NMSAH(T+1)}{\delta},\qquad
\Psi_h^t(s,a)=\sqrt{\frac{2L}{N_h^{t-1}(s,a)\vee1}},\qquad
\Phi_h^t(s,a)=H\sqrt{\frac{2SL}{N_h^{t-1}(s,a)\vee1}},\qquad
b_h^t=\Psi_h^t+\Phi_h^t.
\end{equation}
By \Cref{lemma1}, $\Psi_h^t$ and $\Phi_h^t$ are confidence radii for the reward and the transition estimates, and $b_h^t$ is the exploration bonus. The analysis of Algorithm~\ref{alg:combined} uses the same radii clipped at $1$ and $H$, respectively; see~\eqref{eq:pf-radii}. We use the shorthand $P_h f(s,a)=\sum_{s'}P_h(s'\mid s,a)f(s')$, similarly for $\hat P_h^t$, and write $\hat r_{k,h}^{\Lambda,t}=\bm\lambda_k^\top\hat{\mathbf r}_{k,h}^t$.

\begin{definition}[Nash regret]\label{def:nash regret}
For the policies $\pi^1,\ldots,\pi^T$ executed during learning, define
\[
g_{k,t}=\max_{\pi'_k\in\Pi_k}U_k^{(\pi'_k,\pi_{-k}^t)}(s_1)-U_k^{\pi^t}(s_1),
\qquad g_t=\max_{k\in[N]}g_{k,t},\qquad
\mathrm{Regret}(T)=\sum_{t=1}^T g_t.
\]
\end{definition}
Thus $g_t=\mathrm{NashGap}_{\mathcal G_{\bm\Lambda}}(\pi^t)$ is the scalarized Nash gap of the policy executed in episode $t$. By~\eqref{eq:pub}, every policy $\pi$ satisfies
\begin{equation}\label{eq:onvi-gap-reduction}
\mathrm{WPNG}(\pi)\le
\mathrm{NashGap}_{\mathcal G_{\bm\Lambda}}(\pi)
=\max_{k\in[N]}\max_{\pi'_k\in\Pi_k}
\bigl\{U_k^{(\pi'_k,\pi_{-k})}(s_1)-U_k^\pi(s_1)\bigr\},
\end{equation}
so an $\epsilon$-Nash equilibrium of $\mathcal G_{\bm\Lambda}$ is an $\epsilon$-WPNE of the MOMG.

\paragraph{Planning and policy selection.}
ONVI-MG extends Nash value iteration~\cite{liu2021sharp} to unknown stochastic rewards. At each $(h,s)$, it computes a mixed Nash equilibrium of the finite game whose payoffs are the optimistic action values $Q_h^{k,t}(s,\cdot)$. As in Nash value iteration, we assume an exact Nash equilibrium oracle for general-sum normal-form games, with a fixed measurable selection rule. With terminal values $U_{H+1}^{k,t}=\underline U_{H+1}^{k,t}=0$, the optimistic and pessimistic recursions are
\begin{align}
Q_h^{k,t}(s,a)
&=\min\bigl\{H,\hat r_{k,h}^{\Lambda,t-1}(s,a)+b_h^t(s,a)+\hat P_h^{t-1}U_{h+1}^{k,t}(s,a)\bigr\},\label{eq:onvi-upper}\\
\underline Q_h^{k,t}(s,a)
&=\max\bigl\{0,\hat r_{k,h}^{\Lambda,t-1}(s,a)-b_h^t(s,a)+\hat P_h^{t-1}\underline U_{h+1}^{k,t}(s,a)\bigr\},\label{eq:onvi-lower}\\
U_h^{k,t}(s)&=\mathbb E_{a\sim\pi_h^t(\cdot\mid s)}Q_h^{k,t}(s,a),\qquad
\underline U_h^{k,t}(s)=\mathbb E_{a\sim\pi_h^t(\cdot\mid s)}\underline Q_h^{k,t}(s,a).\label{eq:onvi-values}
\end{align}
Only the optimistic values enter the Nash oracle, so they alone determine the executed policy $\pi^t$. The pessimistic values evaluate $\pi^t$ and require no additional samples. The certificate of episode $t$ and the returned policy are
\begin{equation}\label{eq:onvi-certificate}
C_t=\max_{k\in[N]}\bigl\{U_1^{k,t}(s_1)-\underline U_1^{k,t}(s_1)\bigr\},
\qquad t_*\in\arg\min_{t\in[T]}C_t,\qquad \pi^{\mathrm{out}}=\pi^{t_*},
\end{equation}
with ties broken by the smallest index. The returned joint policy is that of a single episode, so all players use the same index $t_*$.

\begin{algorithm}[!htb]
\caption{ONVI-MG: Optimistic Nash Value Iteration with Certified Policy Selection.}
\label{alg:onvi-mg}
\begin{algorithmic}[1]
\STATE \textbf{Input:} Fixed preferences $\boldsymbol\Lambda\in(\Delta_M)^N$, episodes $T$, confidence $\delta$.
\STATE \textbf{Initialize:} $N_h^0(s,a)=0$, $\hat{\mathbf r}_{k,h}^0(s,a)=0$, and $\hat P_h^0(\cdot\mid s,a)=\mathrm{Uniform}(\mathcal S)$ for all $k,h,s,a$.
\STATE Set $C_{\mathrm{best}}=+\infty$ and $L$ as in~\eqref{eq:onvi-bonuses}.
\FOR{$t=1,\ldots,T$}
  \STATE Set $U_{H+1}^{k,t}=\underline U_{H+1}^{k,t}=0$ for all $k$.
  \FOR{$h=H,\ldots,1$}
    \FOR{each $s\in\mathcal S$}
      \STATE Compute $\Psi_h^t(s,\cdot),\Phi_h^t(s,\cdot)$ from $N_h^{t-1}$ using~\eqref{eq:onvi-bonuses}.
      \STATE Compute $Q_h^{k,t}(s,a)$ and $\underline Q_h^{k,t}(s,a)$ for all $k,a$ using~\eqref{eq:onvi-upper}--\eqref{eq:onvi-lower}.
      \STATE Choose a product mixed Nash equilibrium $\pi_h^t(\cdot\mid s)$ of the game $\{Q_h^{k,t}(s,\cdot)\}_{k=1}^N$.
      \STATE Compute $U_h^{k,t}(s)$ and $\underline U_h^{k,t}(s)$ for all $k$ using~\eqref{eq:onvi-values}.
    \ENDFOR
  \ENDFOR
  \STATE Compute $C_t$ using~\eqref{eq:onvi-certificate}.
  \IF{$C_t<C_{\mathrm{best}}$}
    \STATE Set $C_{\mathrm{best}}\leftarrow C_t$ and $\pi^{\mathrm{out}}\leftarrow\pi^t$.
  \ENDIF
  \STATE Execute $\pi^t$ from $s_1$ and observe $\{(s_h^t,a_h^t,\{\mathbf R_{k,h}^t\}_{k=1}^N,s_{h+1}^t)\}_{h=1}^H$.
  \STATE Update all visit counts, reward means, and transition frequencies to $N_h^t$, $\hat{\mathbf r}_{k,h}^t$, and $\hat P_h^t$.
\ENDFOR
\STATE \textbf{Return:} $\pi^{\mathrm{out}}$ and $C_{\mathrm{best}}$.
\end{algorithmic}
\end{algorithm}

\begin{theorem}[Regret and policy guarantee]\label{thm:NASH-reg}
Under the observation model and the oracle assumption above, Algorithm~\ref{alg:onvi-mg} satisfies, with probability at least $1-\delta$,
\[
\mathrm{Regret}(T)\le10\sqrt2\,H^2S\sqrt{ATL},\qquad
\mathrm{WPNG}(\pi^{\mathrm{out}})\le
\mathrm{NashGap}_{\mathcal G_{\bm\Lambda}}(\pi^{\mathrm{out}})
\le C_{\mathrm{best}}\le15\sqrt2\,H^2S\sqrt{\frac{AL}{T}}.
\]
In particular, $T\ge450H^4S^2AL/\epsilon^2$ suffices for an $\epsilon$-WPNE with probability at least $1-\delta$: this is $\tilde O(H^4S^2A/\epsilon^2)$ episodes, or $\tilde O(H^5S^2A/\epsilon^2)$ environment transitions.
\end{theorem}

The regret bound controls the average Nash gap $\mathrm{Regret}(T)/T$ of the executed policies. To return a single policy with a guarantee, the algorithm uses the certificates: $C_t$ upper bounds the Nash gap of $\pi^t$ (\Cref{lemma3}), so the policy with the smallest certificate inherits the bound on the average certificate.

\subsection{Proof of \Cref{thm:NASH-reg}}

Let $\mathcal F_t$ contain the observed history through episode $t$. Since the preference profile and the selection rule of the oracle are fixed, the planning quantities of episode $t$ are measurable with respect to $\mathcal F_{t-1}$. We write $\mathbb E_t$ for the conditional expectation given $\mathcal F_{t-1}$, with trajectories generated by $\pi^t$ and the true transition kernel $P$.

\begin{lemma}[Uniform model concentration]\label{lemma1}
There is an event $\mathcal E_{\mathrm{mod}}$ with probability at least $1-\delta/2$ on which, simultaneously for all $t\in[T]$, $h,s,a,k$, and functions $f:\mathcal S\to[0,H]$,
\begin{align}
\bigl|\hat r_{k,h}^{\Lambda,t-1}(s,a)-r_{k,h}^{\Lambda}(s,a)\bigr|&\le\Psi_h^t(s,a),\label{eq:reward_bound}\\
\bigl|(\hat P_h^{t-1}-P_h)f(s,a)\bigr|&\le\Phi_h^t(s,a).\label{eq:transition_bound}
\end{align}
Because~\eqref{eq:transition_bound} holds simultaneously for all such $f$, it applies in particular to value functions computed from the same data, such as $U_{h+1}^{k,t}$.
\end{lemma}

\begin{proof}
The observations at a fixed triple $(h,s,a)$ are i.i.d.\ regardless of when the triple is visited, so for every fixed $n\in[T]$ the first $n$ of them form an i.i.d.\ sample. We prove the bounds below for each fixed $n$ and take a union bound over $n\in[T]$; the bounds then hold at the random visit counts.

For a reward coordinate and a sample of size $n$, Hoeffding's inequality gives
\[
\Pr\bigl\{|\hat r_n-r|>\sqrt{2L/n}\bigr\}\le2e^{-4L}.
\]
For transitions, put $\mathbf x=\hat P_n-P$. Since $\|\mathbf x\|_1=\max_{\mathbf z\in\{-1,1\}^S}\mathbf z^\top\mathbf x$, one-sided Hoeffding applied to the samples $z(s')\in[-1,1]$ and a union bound over the $2^S$ sign vectors give
\[
\Pr\bigl\{\|\hat P_n-P\|_1>2\sqrt{2SL/n}\bigr\}\le2^S e^{-4SL}.
\]
Both kernels are probability distributions, so $\mathbf 1^\top\mathbf x=0$, and for every $f:\mathcal S\to[0,H]$,
\begin{equation}\label{eq:holder}
|\mathbf x^\top f|=\bigl|\mathbf x^\top(f-\tfrac H2\mathbf 1)\bigr|\le\tfrac H2\|\mathbf x\|_1.
\end{equation}
Since~\eqref{eq:holder} holds for every such $f$, one event controls all of them. A union bound over triples, sample sizes, and reward coordinates gives total failure probability at most
\[
2NMSAHT\,e^{-4L}+SAHT\,2^S e^{-4SL}\le\delta/2.
\]
Indeed, $L>1$ implies $2e^{-4L}\le e^{-L}$ and $2^Se^{-4SL}\le e^{-L}$, and each term is at most $\delta/16$. At unvisited triples, the reward error is at most $1\le\Psi_h^t$ and the error in the expected next-state value is at most $H\le\Phi_h^t$. Finally, for $\bm\lambda_k\in\Delta_M$, the error of the scalarized reward is at most the largest coordinate error. Substituting $n=N_h^{t-1}(s,a)\le T$ proves both claims.
\end{proof}

The next lemma is where the general-sum structure enters. Optimism is needed, and holds, only for unilateral deviations from the executed policy $\pi^t$, because the Nash gap compares $\pi^t$ with such deviations.

\begin{lemma}[Unilateral optimism]\label{lem:onvi-optimism}
On $\mathcal E_{\mathrm{mod}}$, for every $t,k,h,s,a$ and $\pi'_k\in\Pi_k$,
\[
Q_h^{k,(\pi'_k,\pi_{-k}^t)}(s,a)\le Q_h^{k,t}(s,a),\qquad
U_h^{k,(\pi'_k,\pi_{-k}^t)}(s)\le U_h^{k,t}(s),
\]
where $Q_h^{k,\pi}=r_{k,h}^{\Lambda}+P_hU_{h+1}^{k,\pi}$ is the true scalarized action-value function.
\end{lemma}

\begin{proof}
Fix $t,k,\pi'_k$ and put $\nu=(\pi'_k,\pi_{-k}^t)$. We argue by backward induction on $h$; the terminal values are zero. Assume the claim holds at $h+1$. The Bellman equation and \Cref{lemma1} yield
\[
Q_h^{k,\nu}
=r_{k,h}^{\Lambda}+P_h U_{h+1}^{k,\nu}
\le r_{k,h}^{\Lambda}+P_h U_{h+1}^{k,t}
\le\hat r_{k,h}^{\Lambda,t-1}+b_h^t+\hat P_h^{t-1}U_{h+1}^{k,t}.
\]
The true action value also satisfies $Q_h^{k,\nu}\le H-h+1\le H$, so it is at most the minimum in~\eqref{eq:onvi-upper}. By the Nash condition of the stage game at $(h,s)$,
\[
U_h^{k,\nu}(s)
\le\mathbb E_{a\sim\pi'_{k,h}(\cdot\mid s)\times\pi_{-k,h}^t(\cdot\mid s)}Q_h^{k,t}(s,a)
\le\mathbb E_{a\sim\pi_h^t(\cdot\mid s)}Q_h^{k,t}(s,a)
=U_h^{k,t}(s).
\]
This completes the induction, in which the opponents of player $k$ are held fixed at their policies in $\pi^t$.
\end{proof}

\begin{lemma}[Single-episode regret and certificate]\label{lemma3}
Let $\widetilde b_h^t=\min\{H,b_h^t\}$ and $\mu_t=\mathbb E_t[\sum_{h=1}^H\widetilde b_h^t(s_h^t,a_h^t)]$. On $\mathcal E_{\mathrm{mod}}$, for every episode $t$,
\[
g_t\le2\mu_t,\qquad g_t\le C_t\le3\mu_t.
\]
\end{lemma}

\begin{proof}
First, backward induction shows that, for every realization of the data,
\[
0\le\underline Q_h^{k,t}\le Q_h^{k,t}\le H,\qquad
0\le\underline U_h^{k,t}\le U_h^{k,t}\le H,\qquad
\underline Q_h^{k,t},\underline U_h^{k,t}\le H-h+1.
\]
To see the order, assume $\underline U_{h+1}^{k,t}\le U_{h+1}^{k,t}$. Then the unclipped lower expression in~\eqref{eq:onvi-lower} is at most the unclipped upper expression in~\eqref{eq:onvi-upper}, which is nonnegative. Also, $\hat r_{k,h}^{\Lambda,t-1}\le1$ and $\underline U_{h+1}^{k,t}\le H-h$ imply $\underline Q_h^{k,t}\le H-h+1$. Clipping the lower expression at zero and the upper expression at $H$ therefore preserves their order, and averaging under $\pi_h^t$ gives the bounds on the values.

On $\mathcal E_{\mathrm{mod}}$, a second backward induction gives $\underline U_h^{k,t}\le U_h^{k,\pi^t}$. Indeed, the unclipped lower expression satisfies
\[
\hat r_{k,h}^{\Lambda,t-1}-b_h^t+\hat P_h^{t-1}\underline U_{h+1}^{k,t}
\le r_{k,h}^{\Lambda}+P_h\underline U_{h+1}^{k,t}
\le r_{k,h}^{\Lambda}+P_h U_{h+1}^{k,\pi^t}=Q_h^{k,\pi^t},
\]
and since $Q_h^{k,\pi^t}\ge0$, clipping at zero preserves this comparison. Together with \Cref{lem:onvi-optimism}, this proves $g_t\le C_t$.

For the regret bound, set $D_h^{k,t}=U_h^{k,t}-U_h^{k,\pi^t}$ and $d_h^{k,t}=Q_h^{k,t}-Q_h^{k,\pi^t}$. \Cref{lem:onvi-optimism} with $\pi'_k=\pi_k^t$ implies $D_h^{k,t},d_h^{k,t}\ge0$, and both are at most $H$. Subtracting the Bellman equations and using \Cref{lemma1} gives
\begin{align*}
d_h^{k,t}
&\le(\hat r_{k,h}^{\Lambda,t-1}-r_{k,h}^{\Lambda})+b_h^t
 +(\hat P_h^{t-1}-P_h)U_{h+1}^{k,t}+P_hD_{h+1}^{k,t}\\
&\le2b_h^t+P_hD_{h+1}^{k,t}.
\end{align*}
Combining this inequality with $d_h^{k,t}\le H$ and $P_hD_{h+1}^{k,t}\ge0$ gives
\[
d_h^{k,t}\le\min\{H,2b_h^t+P_hD_{h+1}^{k,t}\}
\le2\widetilde b_h^t+P_hD_{h+1}^{k,t}.
\]
Averaging under $\pi_h^t$ and unrolling under $P$ yields $D_1^{k,t}(s_1)\le2\mu_t$. The value of the best deviation is at most $U_1^{k,t}(s_1)$ by \Cref{lem:onvi-optimism}, so $g_{k,t}\le D_1^{k,t}(s_1)\le2\mu_t$ for every player, and hence $g_t\le2\mu_t$.

For the certificate, define $W_h^{k,t}=U_h^{k,t}-\underline U_h^{k,t}$ and $w_h^{k,t}=Q_h^{k,t}-\underline Q_h^{k,t}$. The order established at the beginning of the proof holds for every realization of the data, without $\mathcal E_{\mathrm{mod}}$, and gives $0\le W_h^{k,t},w_h^{k,t}\le H$. Removing both clippings and applying the uniform transition bound to $W_{h+1}^{k,t}$ yields
\[
w_h^{k,t}\le2b_h^t+\hat P_h^{t-1}W_{h+1}^{k,t}
\le2\Psi_h^t+3\Phi_h^t+P_hW_{h+1}^{k,t}.
\]
Using $w_h^{k,t}\le H$ and $P_hW_{h+1}^{k,t}\ge0$, we obtain
\[
w_h^{k,t}\le\min\{H,2\Psi_h^t+3\Phi_h^t\}+P_hW_{h+1}^{k,t}
\le3\widetilde b_h^t+P_hW_{h+1}^{k,t}.
\]
Averaging and unrolling under $P$ gives $W_1^{k,t}(s_1)\le3\mu_t$. The right-hand side does not depend on the player, so taking the maximum over players proves $C_t\le3\mu_t$.
\end{proof}

\begin{lemma}[Total bonus bound]\label{lemma4}
For every realized sequence of $T$ episodes,
\[
\sum_{t=1}^T\sum_{h=1}^H b_h^t(s_h^t,a_h^t)
\le B_T:=2H\bigl(\sqrt{2L}+H\sqrt{2SL}\bigr)\sqrt{SAT}.
\]
\end{lemma}

\begin{proof}
If a fixed $(h,s,a)$ is visited $n$ times, the counts preceding these visits are $0,1,\ldots,n-1$. The sum of the reciprocal square roots of $\max\{\cdot,1\}$ over these counts is zero for $n=0$ and one for $n=1$. For $n\ge2$,
\[
\sum_{j=1}^n\frac1{\sqrt{(j-1)\vee1}}
=2+\sum_{j=2}^{n-1}j^{-1/2}
\le2+\int_1^{n-1}x^{-1/2}\,dx
=2\sqrt{n-1}\le2\sqrt n.
\]
Since $\sum_{h,s,a}N_h^T(s,a)=HT$, the Cauchy--Schwarz inequality gives
\[
\sum_{h,s,a}\sqrt{N_h^T(s,a)}\le\sqrt{HSA\cdot HT}=H\sqrt{SAT}.
\]
Multiplying by $2(\sqrt{2L}+H\sqrt{2SL})$ proves the claim.
\end{proof}

\begin{theorem}[Cumulative regret and certificate bounds]\label{thm:onvi-cumulative}
With probability at least $1-\delta$,
\[
\mathrm{Regret}(T)\le2K_T\sqrt T,\qquad
\sum_{t=1}^T C_t\le3K_T\sqrt T,
\]
where
\[
K_T=2H\bigl(\sqrt{2L}+H\sqrt{2SL}\bigr)\sqrt{SA}
+H^2\sqrt{2\log(2/\delta)}
\le5\sqrt2\,H^2S\sqrt{AL}.
\]
\end{theorem}

\begin{proof}
\Cref{lemma3} bounds each episode by the expected bonus $\mu_t$, whereas \Cref{lemma4} bounds the bonuses collected along the realized trajectories. We relate the two by a martingale argument. Define
\[
X_t=\sum_{h=1}^H\widetilde b_h^t(s_h^t,a_h^t),\qquad
\mu_t=\mathbb E[X_t\mid\mathcal F_{t-1}],\qquad Z_t=\mu_t-X_t.
\]
The sequence $(Z_t)_{t=1}^T$ is a martingale-difference sequence adapted to $(\mathcal F_t)$: $X_t$ is $\mathcal F_t$-measurable and $\mathbb E[Z_t\mid\mathcal F_{t-1}]=0$. Since $0\le X_t,\mu_t\le H^2$ on every sample path, $|Z_t|\le H^2$, and the Azuma--Hoeffding inequality applies without conditioning on $\mathcal E_{\mathrm{mod}}$. It gives an event $\mathcal E_{\mathrm{sam}}$ of probability at least $1-\delta/2$ on which
\begin{equation}\label{eq:onvi-sampling}
\sum_{t=1}^T\mu_t
\le\sum_{t=1}^T X_t+H^2\sqrt{2T\log(2/\delta)}
\le B_T+H^2\sqrt{2T\log(2/\delta)}=K_T\sqrt T.
\end{equation}
The second inequality uses $\widetilde b_h^t\le b_h^t$ and \Cref{lemma4}.

\Cref{lemma4} holds for every realization of the data and adds no failure probability, so the intersection $\mathcal E_{\mathrm{mod}}\cap\mathcal E_{\mathrm{sam}}$ has probability at least $1-\delta$. On this event, \Cref{lemma3} implies
\[
\mathrm{Regret}(T)=\sum_{t=1}^T\max_k g_{k,t}\le2\sum_{t=1}^T\mu_t\le2K_T\sqrt T,
\qquad
\sum_{t=1}^T C_t\le3\sum_{t=1}^T\mu_t\le3K_T\sqrt T.
\]
Since \Cref{lemma3} bounds $g_{k,t}$ by $2\mu_t$ for every player $k$, the regret bound carries no factor $N$. Finally, $H,S,A\ge1$ and $\log(2/\delta)\le L$ imply the stated upper bound on $K_T$.
\end{proof}

\begin{proof}[Proof of \Cref{thm:NASH-reg}]
Work on $\mathcal E_{\mathrm{mod}}\cap\mathcal E_{\mathrm{sam}}$. The regret bound follows from \Cref{thm:onvi-cumulative}. By~\eqref{eq:onvi-gap-reduction}, \Cref{lemma3}, and the output rule~\eqref{eq:onvi-certificate},
\[
\mathrm{WPNG}(\pi^{\mathrm{out}})
\le\mathrm{NashGap}_{\mathcal G_{\bm\Lambda}}(\pi^{\mathrm{out}})
\le C_{t_*}=\min_{t\in[T]}C_t
\le\frac1T\sum_{t=1}^T C_t
\le\frac{3K_T}{\sqrt T}
\le15\sqrt2\,H^2S\sqrt{\frac{AL}{T}}.
\]
The sufficient condition on $T$ follows by requiring the last expression to be at most $\epsilon$. Each episode contains $H$ transitions, which gives the stated number of environment transitions.
\end{proof}

\section{Proofs for \Cref{sec:twophase}}

This section proves Theorem~\ref{thm:pf-main}; its precise form is Theorem~\ref{thm:pf-uniform}. The proof has three steps. First, Lemma~\ref{lem:value_diff} bounds the value error of the learned model, simultaneously for all players, preferences, and product Markov policies, by an error envelope: a Bellman recursion of the reward and transition radii under the true kernel, clipped at the remaining horizon. This bound depends on neither the player nor the preference, because the reward radius is shared by all players and objectives and the transitions do not depend on preferences; hence no union bound over preferences is needed. Second, Lemma~\ref{lem:optimism-2} shows that the envelope at the final counts is at most twice the exploration value $\bar V_1^t(s_1)$ of every episode $t$, and hence at most twice their average. Third, Lemma~\ref{lem:bonus_sum} bounds the sum of the exploration values over $T$ episodes by $\tilde O(H^2S\sqrt{AT})$. Theorem~\ref{thm:pf-uniform} combines the three steps and transfers the value bound to Nash gaps, WPNG, and PNG.

\paragraph{Setting.} We analyze Algorithm~\ref{alg:combined} under the observation model of Appendix~\ref{sec:nash}. Each episode starts at $s_1$ and contains $H$ transitions. At $(h,s,a)$, where $a\in\mathcal A$ is a joint action, the learner observes the next state and the reward vectors of all players, and the conditional joint law of these observations depends only on $(h,s,a)$. We write $A=\prod_{k=1}^N|\mathcal A_k|$ and let $\Pi$ be the class of product time-dependent randomized Markov policies.

\paragraph{Counts, estimates, and radii.} For $u=0,\ldots,T$, let $N_h^u(s,a)$, $\hat{\mathbf r}_{k,h}^u(s,a)$, and $\hat P_h^u(\cdot\mid s,a)$ be the visit count and the empirical estimates after $u$ episodes, with zero reward and uniform transition estimates at unvisited triples. The final empirical MOMG $\hat{\mathcal G}$ uses $\hat{\mathbf r}^T$ and $\hat P^T$. With $L=\log(16NMSAH(T+1)/\delta)$, the radii used in episode $t$ are
\begin{equation}\label{eq:pf-radii}
\Psi_h^t(s,a)=\min\!\left\{1,\sqrt{\frac{2L}{N_h^{t-1}(s,a)\vee1}}\right\},\qquad
\Phi_h^t(s,a)=\min\!\left\{H,H\sqrt{\frac{2SL}{N_h^{t-1}(s,a)\vee1}}\right\}.
\end{equation}
Lemma~\ref{lem:concentration} shows that $\Psi_h^t$ bounds the error of every reward coordinate of every player, and that $\Phi_h^t$ bounds the error of the expected next-state value of every function with values in $[0,H]$. The reward radius is the same for all players and coordinates. We write $\Psi_h^{T+1}$ and $\Phi_h^{T+1}$ for the radii computed from the final counts $N_h^T$; no episode $T+1$ is executed. Since visit counts never decrease, $\Psi_h^{T+1}\le\Psi_h^t$ and $\Phi_h^{T+1}\le\Phi_h^t$ for every $t\le T$.

\paragraph{Exploration values.} In episode $t$, Algorithm~\ref{alg:combined} uses the exploration reward $\bar r_h^t=\max\{\Phi_h^t/H,\Psi_h^t\}\in[0,1]$ and computes, with $\bar V_{H+1}^t=0$,
\[
\bar Q_h^t=\min\bigl\{H-h+1,\ \bar r_h^t+\hat P_h^{t-1}\bar V_{h+1}^t+\Phi_h^t\bigr\},
\qquad \bar V_h^t(s)=\max_{a\in\mathcal A}\bar Q_h^t(s,a).
\]
These values depend only on the data collected before episode $t$. The algorithm then executes, at each visited state, a joint action that attains the maximum, using a fixed tie-breaking rule. A joint action is a tuple of individual actions, so this exploration policy is a deterministic product policy.

\paragraph{Scalarized values and the error envelope.} For player $k$ and preference $\boldsymbol\lambda\in\Delta_M$, write $r_{k,\boldsymbol\lambda,h}=\boldsymbol\lambda^\top\mathbf r_{k,h}$ and $\hat r_{k,\boldsymbol\lambda,h}=\boldsymbol\lambda^\top\hat{\mathbf r}_{k,h}^T$. Let $Q^\pi_{k,\boldsymbol\lambda,h}$ and $V^\pi_{k,\boldsymbol\lambda,h}$ be the scalarized action values and values of $\pi\in\Pi$ in $\mathcal G$, and $\hat Q^\pi_{k,\boldsymbol\lambda,h}$ and $\hat V^\pi_{k,\boldsymbol\lambda,h}$ those in $\hat{\mathcal G}$. We write $P_hf(s,a)=\sum_{s'}P_h(s'\mid s,a)f(s')$, and similarly for $\hat P_h^u$. For nonnegative radii $\rho=(\rho_h)_{h\in[H]}$, define $e^\pi_{H+1}[\rho]\equiv0$ and
\begin{equation}\label{eq:pf-envelope}
e^\pi_h[\rho](s,a)=\min\bigl\{H-h+1,\ \rho_h(s,a)+P_he^\pi_{h+1}[\rho](s,a)\bigr\},\qquad
e^\pi_h[\rho](s)=\sum_a\pi_h(a\mid s)\,e^\pi_h[\rho](s,a).
\end{equation}
This stagewise-clipped Bellman recursion under the true transition kernel defines an error envelope. It is used only in the analysis, takes values in $[0,H-h+1]$, and is nondecreasing in $\rho$. All statements below hold pointwise on a fixed realization of the data, so the policies and preferences to which they are applied may themselves be chosen from the data.

\begin{lemma}[Uniform model concentration]\label{lem:concentration}
There is an event $\mathcal E_{\mathrm{mod}}$ of probability at least $1-\delta/2$ on which, simultaneously for every update time $u\in\{0,\ldots,T\}$, triple $(h,s,a)$, player $k$, coordinate $i$, and function $f:\mathcal S\to[0,H]$,
\begin{align*}
|\hat r_{k,i,h}^u(s,a)-r_{k,i,h}(s,a)|
&\le\min\!\left\{1,\sqrt{\frac{2L}{N_h^u(s,a)\vee1}}\right\},\\
|(\hat P_h^u-P_h)f(s,a)|
&\le\min\!\left\{H,H\sqrt{\frac{2SL}{N_h^u(s,a)\vee1}}\right\}.
\end{align*}
In particular, $\Psi_h^t$ and $\Phi_h^t$ bound the errors of the estimates after $t-1$ episodes, and $\Psi_h^{T+1}$ and $\Phi_h^{T+1}$ bound those of the final estimates.
\end{lemma}
\begin{proof}
The argument is the one in the proof of Lemma~\ref{lemma1}. The observations at a fixed triple are i.i.d.\ regardless of when the triple is visited, so for every fixed $n\in[T]$ the first $n$ of them form an i.i.d.\ sample. For such a sample, Hoeffding's inequality bounds the error of every reward coordinate, and the $\ell_1$ bound in that proof together with~\eqref{eq:holder} bounds $(\hat P-P)f$ simultaneously for all $f:\mathcal S\to[0,H]$, including functions computed from the same data. A union bound over $n$, triples, and reward coordinates has failure probability at most
\[
2NMSAHTe^{-4L}+SAHT\,2^Se^{-4SL}\le\delta/2.
\]
The resulting bounds hold at every sample size, hence at every adaptive count $N_h^u(s,a)\le T$, including the final counts. Clipping is justified by the deterministic bounds $|\hat r-r|\le1$ and $|(\hat P-P)f|\le H$, which also cover unvisited triples, where the radii equal $1$ and $H$.
\end{proof}

The next lemma is the main step. It turns the concentration bounds into a bound on the value error of every policy in the final empirical game. The bound is the error envelope with the final radii; it is the same for every player and every preference, which is why no union bound over preferences is needed.

\begin{lemma}[Uniform simulation bound]\label{lem:value_diff}
On $\mathcal E_{\mathrm{mod}}$, simultaneously for all $k\in[N]$, $\boldsymbol\lambda\in\Delta_M$, and $\pi\in\Pi$,
\[
\bigl|V^\pi_{k,\boldsymbol\lambda,1}(s_1)-\hat V^\pi_{k,\boldsymbol\lambda,1}(s_1)\bigr|
\le e^\pi_1[\Psi^{T+1}+\Phi^{T+1}](s_1).
\]
\end{lemma}
\begin{proof}
Fix a dataset in $\mathcal E_{\mathrm{mod}}$ and any $k$, $\boldsymbol\lambda$, and $\pi$, and write $e_h=e^\pi_h[\Psi^{T+1}+\Phi^{T+1}]$. Since the weights $\lambda_i$ are nonnegative and sum to one, Lemma~\ref{lem:concentration} gives
\[
|r_{k,\boldsymbol\lambda,h}-\hat r_{k,\boldsymbol\lambda,h}|
\le\sum_{i=1}^M\lambda_i\Psi_h^{T+1}=\Psi_h^{T+1}.
\]
Set $\Delta_h=Q^\pi_{k,\boldsymbol\lambda,h}-\hat Q^\pi_{k,\boldsymbol\lambda,h}$, with $\Delta_{H+1}=0$. Subtracting the Bellman equations of $\mathcal G$ and $\hat{\mathcal G}$, and adding and subtracting $P_h\hat V^\pi_{k,\boldsymbol\lambda,h+1}$, gives
\[
\Delta_h=(r_{k,\boldsymbol\lambda,h}-\hat r_{k,\boldsymbol\lambda,h})
+(P_h-\hat P_h^T)\hat V^\pi_{k,\boldsymbol\lambda,h+1}
+P_h\bigl(V^\pi_{k,\boldsymbol\lambda,h+1}-\hat V^\pi_{k,\boldsymbol\lambda,h+1}\bigr).
\]
Empirical scalar rewards lie in $[0,1]$, so $0\le\hat V^\pi_{k,\boldsymbol\lambda,h+1}\le H-h$, and Lemma~\ref{lem:concentration} bounds the absolute values of the first two terms by $\Psi_h^{T+1}$ and $\Phi_h^{T+1}$. Moreover, both action values lie in $[0,H-h+1]$, so $|\Delta_h|\le H-h+1$.

We show $|\Delta_h(s,a)|\le e_h(s,a)$ by backward induction on $h$. Both sides vanish at $h=H+1$. Suppose the claim holds at $h+1$. Then, for every $s'$,
\[
\bigl|V^\pi_{k,\boldsymbol\lambda,h+1}(s')-\hat V^\pi_{k,\boldsymbol\lambda,h+1}(s')\bigr|
\le\sum_{a'}\pi_{h+1}(a'\mid s')\,|\Delta_{h+1}(s',a')|
\le e_{h+1}(s'),
\]
so the decomposition gives $|\Delta_h|\le\Psi_h^{T+1}+\Phi_h^{T+1}+P_he_{h+1}$. Combined with $|\Delta_h|\le H-h+1$, this is $|\Delta_h|\le e_h$. Averaging under $\pi_1(\cdot\mid s_1)$ proves the claim. The bound does not depend on $k$ or $\boldsymbol\lambda$, so it holds simultaneously for all players, preferences, and policies on the single event $\mathcal E_{\mathrm{mod}}$.
\end{proof}

Phase~1 cannot compute the envelope, which involves the unknown kernel $P$ and the final counts. The next lemma shows that the exploration values computed during Phase~1 nevertheless control it. Although $\bar V^t$ is computed with the empirical kernel $\hat P^{t-1}$, the bonus $\Phi_h^t$ in its recursion pays for the difference between $\hat P^{t-1}$ and $P$.

\begin{lemma}[Exploration values bound the error envelope]\label{lem:optimism-2}
On $\mathcal E_{\mathrm{mod}}$, for every $t\in[T]$ and $\pi\in\Pi$,
\[
e^\pi_1[\Psi^{T+1}+\Phi^{T+1}](s_1)\le e^\pi_1[\Psi^t+\Phi^t](s_1)\le2\bar V_1^t(s_1).
\]
Consequently, $e^\pi_1[\Psi^{T+1}+\Phi^{T+1}](s_1)\le\frac2T\sum_{t=1}^T\bar V_1^t(s_1)$.
\end{lemma}
\begin{proof}
The first inequality holds because $\Psi_h^{T+1}+\Phi_h^{T+1}\le\Psi_h^t+\Phi_h^t$ pointwise and the envelope is nondecreasing in its radii.

For the second, write $e_h=e^\pi_h[\Psi^t+\Phi^t]$ and show $e_h(s,a)\le2\bar Q_h^t(s,a)$ by backward induction on $h$. Both sides vanish at $h=H+1$. Suppose that $e_{h+1}\le2\bar V_{h+1}^t$ pointwise. Since $0\le e_{h+1}\le H-h$, Lemma~\ref{lem:concentration} applied to $f=e_{h+1}$ gives $P_he_{h+1}\le\hat P_h^{t-1}e_{h+1}+\Phi_h^t$. Hence
\[
\Psi_h^t+\Phi_h^t+P_he_{h+1}
\le\Psi_h^t+2\Phi_h^t+2\hat P_h^{t-1}\bar V_{h+1}^t
\le2\bigl(\bar r_h^t+\Phi_h^t+\hat P_h^{t-1}\bar V_{h+1}^t\bigr),
\]
where the second step uses $\Psi_h^t\le\bar r_h^t$. The left-hand side is the expression clipped in~\eqref{eq:pf-envelope}, and the right-hand side is twice the expression clipped in the definition of $\bar Q_h^t$. Clipping the left-hand side at $H-h+1$ and the right-hand side at $2(H-h+1)$ preserves the inequality, which gives $e_h(s,a)\le2\bar Q_h^t(s,a)$. Averaging under $\pi_h(\cdot\mid s)$ and using $\bar V_h^t(s)=\max_a\bar Q_h^t(s,a)$ gives $e_h\le2\bar V_h^t$ and completes the induction. Averaging the second inequality over $t$ proves the last claim.
\end{proof}

It remains to show that the exploration values are small on average. Along the executed trajectory, each $\bar V_1^t(s_1)$ telescopes into the bonuses collected in episode $t$ plus a martingale term, and counting visits bounds the collected bonuses.

\begin{lemma}[Cumulative exploration value]\label{lem:bonus_sum}
There is an event $\mathcal E_{\mathrm{mar}}$ with probability at least $1-\delta/2$ such that, on $\mathcal E_{\mathrm{mod}}\cap\mathcal E_{\mathrm{mar}}$,
\[
\sum_{t=1}^T\bar V_1^t(s_1)\le R_T
:=2H(2H+1)S\sqrt{2LAT}+H\sqrt{2HT\log(2/\delta)}.
\]
\end{lemma}
\begin{proof}
Let $\mathcal F_{t,h}$ contain all completed episodes and the current episode up to the state $s_h^t$ and the chosen joint action $a_h^t$, before the current rewards and next state are generated. It contains the model, radii, and exploration values frozen before episode $t$. Define
\[
Z_{t,h}=P_h\bar V_{h+1}^t(s_h^t,a_h^t)-\bar V_{h+1}^t(s_{h+1}^t).
\]
Under the sampling law, $\mathbb E[Z_{t,h}\mid\mathcal F_{t,h}]=0$ and $|Z_{t,h}|\le H$, so in chronological order these variables form a martingale-difference sequence. The Azuma--Hoeffding inequality gives an event $\mathcal E_{\mathrm{mar}}$, of probability at least $1-\delta/2$, on which
\begin{equation}\label{eq:pf-martingale}
\sum_{t=1}^T\sum_{h=1}^HZ_{t,h}\le H\sqrt{2HT\log(2/\delta)}.
\end{equation}
This step does not use $\mathcal E_{\mathrm{mod}}$.

Now work on $\mathcal E_{\mathrm{mod}}\cap\mathcal E_{\mathrm{mar}}$. Because the executed joint action attains the maximum in the definition of $\bar V_h^t$, we have $\bar V_h^t(s_h^t)=\bar Q_h^t(s_h^t,a_h^t)$ along the trajectory. Dropping the clipping in $\bar Q_h^t$ and applying Lemma~\ref{lem:concentration} to $f=\bar V_{h+1}^t$, which takes values in $[0,H-h]$, gives
\begin{align*}
\bar V_h^t(s_h^t)
&\le\bar r_h^t(s_h^t,a_h^t)+\hat P_h^{t-1}\bar V_{h+1}^t(s_h^t,a_h^t)+\Phi_h^t(s_h^t,a_h^t)\\
&\le\bar r_h^t(s_h^t,a_h^t)+2\Phi_h^t(s_h^t,a_h^t)+\bar V_{h+1}^t(s_{h+1}^t)+Z_{t,h}.
\end{align*}
Summing over $h$ telescopes, since $\bar V_{H+1}^t=0$. Summing over episodes and using~\eqref{eq:pf-martingale} gives
\begin{equation}\label{eq:pf-exploration-sum}
\sum_{t=1}^T\bar V_1^t(s_1)
\le\sum_{t=1}^T\sum_{h=1}^H(\bar r_h^t+2\Phi_h^t)(s_h^t,a_h^t)
+H\sqrt{2HT\log(2/\delta)}.
\end{equation}

It remains to count the collected bonuses. Since $S\ge1$, the radii in~\eqref{eq:pf-radii} satisfy
\[
\frac{\Phi_h^t}{H}
=\min\left\{1,\sqrt{\frac{2SL}{N_h^{t-1}\vee1}}\right\}
\ge\min\left\{1,\sqrt{\frac{2L}{N_h^{t-1}\vee1}}\right\}
=\Psi_h^t,
\]
so $\bar r_h^t=\Phi_h^t/H$, and the collected bonus in~\eqref{eq:pf-exploration-sum} is $(2+1/H)\sum_{t,h}\Phi_h^t(s_h^t,a_h^t)$. Because the reward radius is shared, there is no sum over players or objectives. For a triple visited $n\ge2$ times, the counts preceding the visits are $0,1,\ldots,n-1$, and
\[
\sum_{i=1}^n\frac1{\sqrt{(i-1)\vee1}}
=2+\sum_{i=2}^{n-1}i^{-1/2}
\le2+\int_1^{n-1}x^{-1/2}\,dx
=2\sqrt{n-1}\le2\sqrt n.
\]
The bound also holds for $n=0,1$. Grouping visits by $(h,s,a)$ and applying the Cauchy--Schwarz inequality with $\sum_{h,s,a}N_h^T(s,a)=HT$, we obtain
\[
\sum_{t,h}\Phi_h^t(s_h^t,a_h^t)
\le2H\sqrt{2SL}\sum_{h,s,a}\sqrt{N_h^T(s,a)}
\le2H\sqrt{2SL}\sqrt{HSA\cdot HT}
=2H^2S\sqrt{2LAT}.
\]
Multiplying by $2+1/H$ and substituting into~\eqref{eq:pf-exploration-sum} proves the claim.
\end{proof}

\begin{theorem}[Uniform guarantees for the learned game]\label{thm:pf-uniform}
Let $\hat{\mathcal G}$ be the final empirical MOMG produced by Algorithm~\ref{alg:combined}, and put
\[
\eta_T:=\frac{2R_T}{T}.
\]
With probability at least $1-\delta$, all of the following statements hold simultaneously:
\begin{enumerate}
\item For every product Markov policy $\pi\in\Pi$, player $k$, and preference $\boldsymbol\lambda\in\Delta_M$,
\[
\bigl|V^\pi_{k,\boldsymbol\lambda,1}(s_1)-\hat V^\pi_{k,\boldsymbol\lambda,1}(s_1)\bigr|\le\eta_T.
\]
\item For every preference profile $\boldsymbol\Lambda\in(\Delta_M)^N$ and every Nash equilibrium $\hat\pi$ of the empirical scalarized game $\hat{\mathcal G}_{\boldsymbol\Lambda}$,
\[
\mathrm{WPNG}_{\mathcal G}(\hat\pi)
\le\mathrm{NashGap}_{\mathcal G_{\boldsymbol\Lambda}}(\hat\pi)
\le2\eta_T.
\]
\item For every product Markov policy $\pi\in\Pi$,
\[
\bigl|\mathrm{WPNG}_{\mathcal G}(\pi)-\mathrm{WPNG}_{\hat{\mathcal G}}(\pi)\bigr|\le2\eta_T,
\]
and the same bound holds with $\mathrm{PNG}$ in place of $\mathrm{WPNG}$.
\end{enumerate}
The policies, preferences, and empirical equilibrium selections may depend on the collected data.
\end{theorem}
\begin{proof}
The event $\mathcal E_{\mathrm{mod}}\cap\mathcal E_{\mathrm{mar}}$ has probability at least $1-\delta$; fix a dataset in it. For the first claim, Lemmas~\ref{lem:value_diff}, \ref{lem:optimism-2}, and~\ref{lem:bonus_sum} give, for every $\pi$, $k$, and $\boldsymbol\lambda$,
\[
\bigl|V^\pi_{k,\boldsymbol\lambda,1}(s_1)-\hat V^\pi_{k,\boldsymbol\lambda,1}(s_1)\bigr|
\le e^\pi_1[\Psi^{T+1}+\Phi^{T+1}](s_1)
\le\frac2T\sum_{t=1}^T\bar V_1^t(s_1)
\le\frac{2R_T}{T}=\eta_T.
\]
The right-hand side depends on neither the player, the policy, nor the preference.

For the second claim, fix $\boldsymbol\Lambda$, an empirical equilibrium $\hat\pi$, a player $k$, and a deviation $\pi'_k$, and write $\sigma=(\pi'_k,\hat\pi_{-k})$. Evaluating all values at $s_1$ and inserting the empirical values gives
\begin{align*}
V^\sigma_{k,\boldsymbol\lambda_k,1}-V^{\hat\pi}_{k,\boldsymbol\lambda_k,1}
&=\bigl(V^\sigma_{k,\boldsymbol\lambda_k,1}-\hat V^\sigma_{k,\boldsymbol\lambda_k,1}\bigr)
+\bigl(\hat V^\sigma_{k,\boldsymbol\lambda_k,1}-\hat V^{\hat\pi}_{k,\boldsymbol\lambda_k,1}\bigr)
+\bigl(\hat V^{\hat\pi}_{k,\boldsymbol\lambda_k,1}-V^{\hat\pi}_{k,\boldsymbol\lambda_k,1}\bigr)\\
&\le\eta_T+0+\eta_T.
\end{align*}
The first claim bounds the first and third terms, and the empirical Nash condition makes the middle term nonpositive. Taking the supremum over deviations and the maximum over players proves the bound on the scalarized Nash gap. For every deviation-gain vector $\mathbf d_k$, $\min_i d_{k,i}\le\boldsymbol\lambda_k^\top\mathbf d_k$, so the same bound controls WPNG.

For the third claim, fix $\pi$, $k$, and $\pi'_k$, and let $\mathbf d_k$ and $\hat{\mathbf d}_k$ be the corresponding deviation-gain vectors in $\mathcal G$ and $\hat{\mathcal G}$. Applying the first claim with the coordinate preferences $\boldsymbol\lambda=\mathbf e_i$ bounds the error of each of the two values entering a gain, hence
\[
\|\mathbf d_k-\hat{\mathbf d}_k\|_\infty\le2\eta_T,
\qquad
\Bigl|\min_i d_{k,i}-\min_i\hat d_{k,i}\Bigr|\le2\eta_T,
\]
where the second inequality holds because the coordinate minimum is $1$-Lipschitz in the sup norm. Taking suprema over deviations and maxima over players preserves this additive error. Since $\mathrm{WPNG}(\pi)=\max_k\sup_{\pi'_k}\min_id_{k,i}$, the third claim follows for every $\pi$. Finally, for every $\mathbf d\in\mathbb R^M$,
\[
\inf_{\boldsymbol\lambda\in\Delta_M^\circ}\boldsymbol\lambda^\top\mathbf d
=\min_{\boldsymbol\lambda\in\Delta_M}\boldsymbol\lambda^\top\mathbf d
=\min_id_i,
\]
by density of the relative interior and continuity of the linear function. Hence PNG and WPNG coincide in both games, and the PNG bound follows.
\end{proof}

Using $H,S,A\ge1$ and $\log(2/\delta)\le L$ gives
\[
2\eta_T=\frac{4R_T}{T}
=8H(2H+1)S\sqrt{\frac{2LA}{T}}+4H\sqrt{\frac{2H\log(2/\delta)}{T}}
\le28\sqrt2\,H^2S\sqrt{\frac{AL}{T}}.
\]
Consequently, any integer $T$ satisfying
\begin{equation}\label{eq:pf-sample-complexity}
T\ge\frac{1568\,H^4S^2A}{\epsilon^2}\log\frac{16NMSAH(T+1)}{\delta}
\end{equation}
ensures all gap guarantees in Theorem~\ref{thm:pf-uniform} at accuracy $\epsilon$. This is $\tilde O(H^4S^2A/\epsilon^2)$ exploration episodes, or $\tilde O(H^5S^2A/\epsilon^2)$ environment transitions, the same order as the fixed-profile bound of Theorem~\ref{thm:NASH-reg}. Apart from the joint-action factor $A$, the numbers of players and objectives enter only through $L$. If the planning oracle returns an $\alpha$-Nash equilibrium of the empirical game, the decomposition in the proof of the second claim gives a true scalarized Nash gap of at most $\alpha+2\eta_T$.

\end{document}